\documentclass[lettersize,journal,10pt]{IEEEtran}
\usepackage{amsmath,amsfonts}
\usepackage{algorithmic}
\usepackage{algorithm}
\usepackage{array}
\usepackage{textcomp}
\usepackage{stfloats}
\usepackage{url}
\usepackage{verbatim}
\usepackage{graphicx}
\usepackage{cite}
\usepackage[nolist,nohyperlinks]{acronym}
\begin{acronym}
    \acro{AP}{access point}
    \acro{CCDF}{complementary cumulative density function}
    \acro{CDF}{cumulative density function}
    \acro{PDF}{probability density function}
    \acro{PPP}{Poisson point process}
    \acro{BS}{base station}
    \acro{LOS}{line of sight}
    \acro{SINR}{signal-to-interference-plus-noise ratio}
    \acro{SIR}{signal-to-interference ratio}
    \acro{SNR}{signal-to-noise ratio}
    \acro{mm-wave}{millimeter wave}
    \acro{PV}{Poisson-Voronoi}
    \acro{LOS}{line-of-sight}
    \acro{NLOS}{non line-of-sight}
    \acro{PGFL}{probability generating functional}
    \acro{PLCP}{Poisson line Cox process}
    \acro{MBS}{macro base station}
    \acro{PLT}{Poisson line tessellation}
    \acro{SBS}{small cell base station}
    \acro{RAT}{radio access technique}
    \acro{RCS}{radar cross section}
    \acro{5G}{fifth generation}
    \acro{PLP}{Poisson line process}
    \acro{UE}{user equipment}
    \acro{BPP}{binomial point process}
    \acro{BLP}{binomial line process}
    \acro{BLCP}{binomial line Cox process}
    \acro{RSU}{road-side unit}
    \acro{MD}{meta distribution}
    \acro{SF}{signal fraction}
    \acro{SG}{stochastic geometry}
    \acro{FMCW}{frequency-modulated continuous waveforms}
    \acro{MHCP}{Matern hard-core process}
    \acro{PLE}{path-loss exponent}
\end{acronym}

\usepackage{subfig}
\usepackage{xcolor}
\usepackage{float}
\usepackage{svg}
\usepackage{pdfpages}
\usepackage{comment}

\usepackage{graphicx,wrapfig,lipsum}
\usepackage{rotating}
\usepackage{amsmath, amssymb, amsthm}
\usepackage{stmaryrd}
\usepackage{tikz}
\usepackage{verbatim}
\usepackage{url}
\usepackage[utf8]{inputenc}
\usepackage[english]{babel}

\newtheorem{theorem}{Theorem}[]
\newtheorem{corollary}{Corollary}[]

\newtheorem{remark}{Remark}[]

\newtheorem{lemma}[]{Lemma}

\usepackage{multicol}

\usepackage{scalerel}
\usepackage{multicol}
\usepackage{setspace}
\newtheorem{definition}{Definition}
\usepackage{booktabs}
\usepackage{bbm}
\usepackage[normalem]{ulem}
\usepackage{color}
\usepackage{dsfont}
\usepackage{bm}
\usepackage{setspace}
\usetikzlibrary{automata}
\usetikzlibrary{arrows}
\usetikzlibrary{shapes.geometric, arrows}
\usepackage[]{footmisc}
\usetikzlibrary{positioning}
\usetikzlibrary{calc}
\usetikzlibrary{arrows}
\usepackage{dirtytalk}
\allowdisplaybreaks

\usepackage{lipsum}

\usepackage[margin = 0.3cm]{caption}
\usepackage{mathtools}
\usepackage{csquotes}

\makeatletter
\newcommand{\vast}{\bBigg@{3}}
\newcommand{\Vast}{\bBigg@{4}}
\makeatother

\usepackage{multirow}

\begin{document}

\title{Modeling and Statistical Characterization of Large-Scale Automotive Radar Networks}



\author{
\IEEEauthorblockN{Mohammad Taha Shah, Gourab Ghatak, Ankit Kumar, and Shobha Sundar Ram}
\thanks{M.T. Shah and A. Kumar are with the Bharti School of Telecommunication Technology and Management, IIT Delhi, New Delhi, India, 110016; Email: \{tahashah, bsy227531\}@dbst.iitd.ac.in. G. Ghatak is with the Department of Electrical Engineering, IIT Delhi, New Delhi, India, 110016; Email: gghatak@ee.iitd.ac.in. S. S. Ram is with the Department of Electronics and Communication Engineering, IIIT Delhi, New Delhi, India, 110020; Email: shobha@iiitd.ac.in.}
}
\maketitle

\begin{abstract}
The impact of discrete clutter and co-channel interference on the performance of automotive radar networks has been studied using stochastic geometry, in particular, by leveraging Poisson point processes (PPPs). However, such characterization does not take into account the impact of street geometry and the fact that the location of the automotive radars are restricted to the streets as their domain rather than the entire Euclidean plane. In addition, the structure of the streets may change drastically as a vehicle moves out of a city center towards the suburban areas. Consequently, not only the radar performance change but also the radar parameters and protocols must be adapted for optimum performance. In this paper, we propose and characterize line and Cox process-based street and point models to analyze large-scale automotive radar networks. We consider the classical Poisson line process (PLP) and the newly introduced Binomial line process (BLP) model to emulate the streets, and the corresponding PPP-based Cox process to emulate the vehicular nodes. In particular, the BLP model effectively considers the spatial variation of street geometry across different parts of the city. We derive the effective interference set experienced by an automotive radar, the statistics of distance to interferers, and characterize the detection probability of the ego radar as a function of street and vehicle density. Finally, leveraging the real-world data on urban streets and vehicle density across different cities of the world, we present how the radar performance varies in different parts of the city as well as across different times of the day. Thus, our study equips network operators and automotive manufacturers with system design insights to optimize automotive radar networks.
\end{abstract}

\begin{IEEEkeywords}  
Stochastic geometry, Automotive radar, Poisson line Cox process, Binomial line Cox process.
\end{IEEEkeywords}

\section{Introduction}
Millimeter-wave automotive radars have proliferated on roads to support advanced driver assistance systems to reduce road accidents and improve passenger comfort and efficiency. These radars support several important applications, including automatic cruise control, obstacle detection, and blind spot detection~\cite{lu2014connected,bilik2019rise}. 
A comprehensive introduction to automotive radar is presented in~\cite{grimes1974automotive, hartenstein2008tutorial, meinel2013automotive, waldschmidt2021automotive}. Based on the advantages these radars offer, it is predicted that every car on the road will carry at least one radar in the near future. However, the downside to this development is that as more vehicles get equipped with radars, mutual interference between them will begin to limit the performance since they occupy the same frequency band~\cite{goppelt2010automotive, alland2019interference}.  Several potential solutions for mitigating interference have been suggested based on resource-based multiplexing in time, frequency, space, and modulation, or adaptive beamforming. However, in order to design and develop effective interference-mitigation solutions, researchers require accurate and tractable models for automotive radar networks. Currently, the state-of-the-art statistical models of automotive radar networks are based on simplistic assumptions of the locations of vehicular nodes, e.g., assuming them to be located only on highways~\cite{al2017stochastic, ghatak2022radar}. We address this limitation by adapting stochastic \ac{SG} to capture complex urban street topologies. Specifically, we utilize line-based Cox processes to model the spatial transition from dense city centers to suburban areas, maintaining analytical tractability for network planning.

\subsection{Related Work}
A plethora of recent work has emerged on automotive radar networks, focusing on modeling and characterizing interference among vehicular radars using tools from \ac{SG}. Early studies such as~\cite{brooker2007mutual,goppelt2011analytical} employed deterministic models to simulate mutual interference between radars, assuming fixed vehicle positions and idealized road layouts. While these models enabled accurate simulations of specific scenarios, they lacked generalizability due to their reliance on static configurations and simplified propagation assumptions. Authors in~\cite{schipper2015simulative} extended this approach by incorporating traffic flow patterns along roadways, offering realistic simulations at the expense of computational complexity.

To overcome the limitations of deterministic modeling, researchers have increasingly turned to \ac{SG}, as \ac{SG}-based models offer several advantages, including closed-form expressions for \ac{SIR} distribution, efficient parameter tuning, and scalable performance analysis across varying urban densities. Authors in~\cite{al2017stochastic} were among the first to apply \ac{SG} techniques to automotive radar networks, modeling the distribution of vehicles as a \ac{PPP} and estimating the mean \ac{SIR} for detection metrics. Authors in~\cite{munari2018stochastic} later used the strongest interferer approximation to determine the radar detection range and false alarm rate under similar assumptions. Fang \textit{et al.}~\cite{fang2020stochastic} refined this model further by incorporating fluctuating radar cross-sections (RCS) using Swerling-I and Chi-square models, yielding more realistic detection probability curves than those based on constant RCS values. These works primarily focused on two-lane highway-like setups, limiting their applicability to complex urban environments. Authors in~\cite{chu2020interference} introduced a marked point process model to address multi-lane interference characterization, while in~\cite{mishra2020stochastic} and~\cite{kui2021interference}, authors explored the use of \ac{MHCP} to capture vehicle spacing constraints and physical separation requirements. Although these models provide higher fidelity, they often come at the cost of reduced analytical tractability, necessitating Monte Carlo simulations. Besides the \ac{SG} based modeling for interference characterization, several studies have proposed strategies leveraging vehicle-to-everything (V2X) communication. Huang \textit{et al.}~\cite{huang2019v2x} evaluated interference scenarios and proposed a centralized spectrum allocation scheme based on real-time vehicle location data. Likewise, authors in~\cite{zhang2020vanet} introduced a VANET-assisted TDMA protocol to coordinate radar spectrum access, while Aydogdu \textit{et al.}~\cite{aydogdu2019radchat} explored decentralized networking protocols for adaptive interference control. In~\cite{wang2023performance}, authors introduced frequency division multiplexing and hopping techniques to mitigate interference under multi-lane setups modeled via \ac{MHCP}. Despite promising results, many of these approaches rely on cooperative behavior and external infrastructure, which may not always be feasible.

More recently, the scope of \ac{SG} in vehicular networks has expanded to address Integrated Sensing and Communication (ISAC). Wang \textit{et al.}~\cite{wang2024network} and Armeniakos \textit{et al.}~\cite{armeniakos2025stochastic} have developed comprehensive network-level frameworks to evaluate the trade-offs between sensing accuracy and communication rates, highlighting the \say{mutual benefit} of shared hardware. Similarly, authors in~\cite{nabil2024beamwidth} and~\cite{kumar2024stochastic} investigated beamwidth optimization and resource partitioning in radar-aided mmWave cellular networks to minimize beam training overhead. Parallel to ISAC, advanced blockage modeling has gained traction. Authors in~\cite{xu2024stochastic} utilized \ac{SG} to analyze coverage in RIS-assisted mmWave systems under random blockage distributions, while Nazar \textit{et al.}~\cite{nazar2025multi} proposed multi-modal sensor fusion (camera, radar, LiDAR) for proactive blockage prediction in dynamic vehicular environments. Unlike these works, which primarily focus on communication link maintenance, proactive blockage prediction, or sensing-communication trade-offs, our framework explicitly derives the \textit{mutual interference geometry} unique to radar-to-radar interactions. 

\ac{SG} has also been applied to analyze the impact of clutter scatterers on radar detection performance. Ram et al.~\cite{ram2020estimating} modeled discrete clutter using \acp{PPP} and derived detection metrics under different propagation conditions. Subsequent works~\cite{ram2021optimization, ram2022estimation, ram2022optimization, singhal2023leo} extended this framework to optimize radar waveform parameters and integrate sensing with communication systems. However, these studies typically assume uniform placement of vehicles over the Euclidean plane rather than being constrained to streets, which limits their realism in dense urban settings. In reality, automotive radars are constrained to roads, and their spatial distribution varies significantly between city centers and suburban areas. Authors in~\cite{ghatak2022radar} addressed this issue partially by introducing a highway-centric layout with directional alignment, but the model was limited in scope. 

\subsection{Motivation}
While previous studies on \ac{SG}-based spatial networks have established robust baselines for homogeneous or highway-centric environments, capturing the spatial heterogeneity of city-scale deployments remains a challenge. For instance, in real-life scenarios, automotive radars are likely to be randomly located on streets rather than off streets. Therefore, the prior assumption that an interfering radar could be randomly located anywhere near the ego radar is incorrect. In this work, we factor in the structure of city streets (downtown versus suburban city areas) and the time of the day (peak versus off-peak hours) to analyze the performance of automotive radar networks. Specifically, we explore the impact of street geometry on radar interference using line processes, which are random collections of lines in a two-dimensional Euclidean plane. Furthermore, the distribution of cars on the street is modeled as a point process with the lines as their domain. Together, this forms a doubly stochastic Cox model of automotive radars to study the effect of interference on the detection performance of an ego radar. While line processes have been employed in vehicular communication networks to model street-constrained user distributions~\cite{chetlur2019coverage, jeyaraj2021transdimensional, choi2023coverage}, our work addresses a distinct problem, automotive radar interference in dense urban environments. Radar networks differ from communication systems in their directional, bounded sensing sectors, co-channel mutual interference structure, and SIR-based detection metrics. Moreover, prior works rely solely on the homogeneous \ac{PLP}, which cannot capture the finite, spatially varying street density observed in real cities. To bridge this gap, we employ  \ac{BLCP}~\cite{shah2024binomial}, a non-homogeneous, finite-support model that accurately emulates city-center-to-suburb transitions and validate both \ac{PLCP} and \ac{BLCP} against real-world urban data from global metropolises. This combination of novel modeling, radar-specific interference characterization, and empirical validation distinguishes our work from existing street-constrained network analyses.

The key methodological distinction between these models lies in how they represent street geometry: \ac{PLCP} assumes an infinite, homogeneous street layout, making it well-suited for small-scale, uniform urban environments; in contrast, \ac{BLCP} captures finite, bounded street structures with spatially varying density, better reflecting the transition from high-density city centers to low-density suburbs. While classical Cox processes such as the Thomas cluster process, log-Gaussian Cox process, or Poisson-Voronoi Cox process are designed for planar clustering and offer limited tractability, the \ac{BLCP} and \ac{PLCP} are tailored specifically to vehicular networks constrained to street geometries. The \ac{BLCP}, in particular, enables location-dependent performance analysis and better captures suburban transitions, complementing the \ac{PLCP}'s homogeneous modeling of dense urban cores. By building upon and extending these theoretical foundations, our paper offers a novel and analytically tractable framework that improves upon existing literature by incorporating realistic urban street geometry and enabling performance evaluation in spatially homogeneous and heterogeneous environments.
\subsection{Organization and Contributions}
The key contributions of our work are as follows:
\begin{itemize}
    \item \textbf{Modeling realistic urban street geometry:} We introduce a novel framework that captures both dense city centers and spatially varying suburban areas using two line-driven Cox processes: the \ac{PLCP} and the \ac{BLCP}. The \ac{PLCP} provides a compact model for uniform street layouts (e.g., city cores), while the \ac{BLCP} captures finite, bounded street distributions that reflect real-world transitions from high-density downtown to sparse suburbs.
    \item \textbf{Deriving interference statistics under practical constraints:} First, we define the interfering sets for the models and derive the distance statistics between the interfering radars and ego radar. Then, based on the derived interfering set, we characterize the statistics of the interference experienced by the ego radar. This enables us to analytically derive the detection success probability of an ego radar under mutual interference, incorporating LOS/NLOS propagation, vehicular intensity, and finite radar range. Key system insight from our analysis shows detection probability in \ac{BLCP} varies by approximately 40\% between the city center  and outskirts.
    \item \textbf{Validating against real-world road data:} A major distinguishing feature of this work is the validation of both models using real-world road network data from four global cities i.e., New Delhi, Paris, Washington, and Johannesburg - via the OSMnx Python module. This empirical validation demonstrates the applicability of our theoretical framework in capturing time-of-day variations in traffic density and radar performance. Key system insight from our analysis shows that detection probability varies by 30\% between peak and off-peak hours. 
\end{itemize}
\begin{table}[t]
\centering
\small{
\begin{tabular}{|l|l|}
\hline
Notation         & Description \\ \hline \hline
$\mathcal{P}_{\rm P}$ & Poisson line process\\
$\lambda_{\rm L}$ & Intensity of PLP $({\rm m}^{-1})$ \\
$\mathcal{P}_{\rm B}$ & Binomial line process\\
$\Phi_{\rm P}$ & Poisson line Cox process\\
$\Phi_{\rm B}$ & Binomial line Cox process\\
$n_{\rm B}$ & Number of lines of the {BLP} $\mathcal{P}_{\rm B}$ \\
$R_{\rm g}$ & Radius of circle in which {BLP} $\mathcal{P}_{\rm B}$ is restricted $({\rm m})$\\
$r_0$ & Distance of ego radar from origin $({\rm m})$\\
$L_i$ & $i-$th line of a {PLP} or {BLP}\\
$\Phi_{{\rm L}_i}$ & \begin{tabular}[c]{@{}l@{}}Poisson point process on\\ on $i-$th line of a {PLP} or {BLP}\end{tabular} \\
$\lambda$ & Intensity of $\Phi_{{\rm L}_i}$ $({\rm m}^{-1})$\\
$k$ & \begin{tabular}[c]{@{}l@{}}$k = {\rm P}$ represents PLP/PLCP,\\ and $k = {\rm B}$ represents BLP/BLCP.\end{tabular} \\
$\mathcal{D}_k$ & Generating domain of $k$ line process\\
$R$ & Distance of ego radar to target $({\rm m})$\\
$\Omega$ & Half power beamwidth of automotive radars $({\rm deg})$\\
$\mathbf{a}$ & Boresight direction of automotive radars\\
$R_k$ & \begin{tabular}[c]{@{}l@{}}Maximum distance between ego\\ radar and interfering radar $({\rm m})$\end{tabular} \\
$\mathcal{S}^{\rightarrow}_{\mathbf{u},k}$ & \begin{tabular}[c]{@{}l@{}}Radar Sector generated by any radar\\ point for axis vector $\mathbf{a} = (0,1)$.\end{tabular} \\
$\Phi_k^{\rm I}$ & Set of interfering Cox points for PLCP/BLCP\\
$h_{\mathbf{w}_k}$ & Fading power \\
$u_i$ & Intersection distance of $i-$th line from origin\\
$d_i$ & Intersection distance of $i-$th line from ego radar\\
$||w_{k,i}||$ & \begin{tabular}[c]{@{}l@{}}Distance between ego radar and\\ interfering radars present on $L_i$ $({\rm m})$\end{tabular} \\
$\xi_k$ & SIR received at the ego radar \\
$\gamma$ & SIR threshold \\
$p_{{\rm D}, k}(\gamma)$ & Detection probability at threshold $\gamma$\\ 
\hline
\end{tabular}
}
\caption{Summary of notations used in the paper.}
\label{tab:notations}
\end{table}
The overall paper is organized as follows. In the following section~\ref{sec:SysModel}, we define the street geometry models and channel model. In Section~\ref{sec:IntDis}, we define the set of points falling within a radar sector and, consequently, the interfering vehicles set. Following which we derive the maximum and minimum distance between the ego radar and interfering radar present on $i\textsuperscript{th}$ line. Finally we derive the detection probability in Section~\ref{sec:RadarProb}, and plot various results in section~\ref{sec:Results} against various system parameters in accordance with real-world road and traffic network parameters.

While this paper provides a fundamental characterization of the automotive radar detection performance by taking into account complex street geometry, leveraging the results of this paper, the companion paper~\cite{shah2024fine} extends this framework to a fine-grained analysis of the radar performance using meta-distributions. We invite the reader to refer to~\cite{shah2024fine} for a follow-up discussion on this.

\subsection{Notations}
To differentiate the notations between the two Cox process models, we will use subscript of `$k = {\rm P}$' for \ac{PLCP} and `$k = {\rm B}$' for \ac{BLCP} models. For example, line processes are identified using calligraphic letters such as $\mathcal{P}$, thus $\mathcal{P}_{\rm P}$ denotes \ac{PLP} and $\mathcal{P}_{\rm B}$ denotes \ac{BLP}. Point processes are identified using the symbol $\Phi_k$. A \ac{BLP} usually consists of $n_{\rm B}$ lines. Probability operators are represented by $\mathbb{P}$ and expectation by $\mathbb{E}$. A complete list of important notations is presented in Table~\ref{tab:notations}.

\section{System Model}
\label{sec:SysModel}
\subsection{Network Geometries}
If we consider a snapshot of road geometries of some of the biggest cities in the world, e.g., see New Delhi in Fig.~\ref{fig:fig_1} (a) and Paris, Washington, and Johannesburg in the first column of Table~\ref{MainTable}, we observe a random distribution of streets, both in terms of density and orientation. From the perspective of an ego radar, there can be a significant variation in street geometry and vehicle density at any time. Therefore, instead of considering each distribution of automotive radars on the roads as a separate scenario, we treat them as an instance of an underlying spatial stochastic process. Accordingly, we model the road network as a line process. In particular, we consider two models. First, we consider a network of streets as a homogeneous \ac{PLP}, typical of road geometries within a small area/district within a city. The second model considers a network of streets as inhomogeneous \ac{BLP} with greater road density at the center (city centers/downtown) and lower road density at the peripherals (suburbs). 
\begin{figure*}[t]
\centering
\includegraphics[trim={0cm 0cm 0cm 0cm},clip,width=1\textwidth]{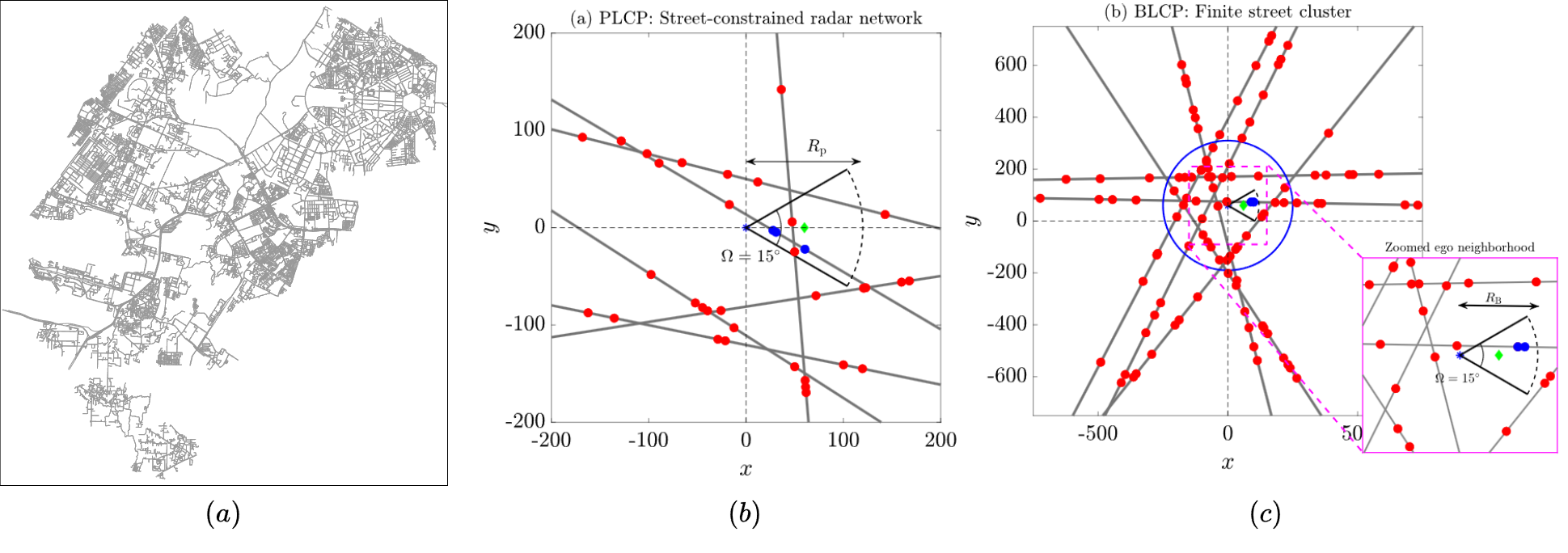}
\caption{(a) Road map of New Delhi city. (b) and (c) Single realization of the PLCP and BLCP model where the ego radar represented as a blue astrix is located at the origin in case of PLCP and at location $r_0 = 60$ from the origin in case of BLCP. Interfering vehicles are illustrated in blue and non-interfering vehicles in red dots. For (c), to emphasize local interference behavior, a magnified inset shows the immediate neighborhood of the ego radar, including the target location, nearby vehicles, and sector boundaries.
}
\label{fig:fig_1} 
\end{figure*}

\subsubsection{\ac{PLCP}}
We consider a network of streets modeled as a homogeneous \ac{PLP}, $\mathcal{P}_{\rm P} = \{L_1, L_2, \dots\}$, i.e., a stochastic set of lines in the two-dimensional Euclidean $x-y$ plane. Any line, $L_i$, of $\mathcal{P}_{\rm P}$ is uniquely characterized by its distance $r_i$ from the origin and the angle $\theta_i$ between the normal and the $x$-axis as shown in Fig.~\ref{fig:fig_2}. The pair of parameters $(\theta_i,r_i)$ corresponds to a point, $\mathbf{q}_i$, in the representation space $\mathcal{D}_{\rm P} \equiv [0,2\pi) \times (0,\infty)$ which is an open cylinder, and is called the domain set of $\mathcal{P}_{\rm P}$. Thus, there is a one-to-one correspondence between $\mathbf{q}_i, i = 1:I$ in $\mathcal{D}_{\rm P}$ and the lines, $L_i, i = 1:I$ in $\mathbb{R}^2$. The number of generating points, $I$, in any $S \subset \mathcal{D}_{\rm P}$ follows the Poisson distribution $\Phi_{\mathcal{D}_{\rm P}}$ with parameter $\lambda_{\rm L}|S|$, where $|S|$ represents the Lebesgue measure of $S$. The parameter $\lambda_{\rm L}|S|$ represents the urban road density, which varies from city to city and within a city from downtown to suburban areas. Our analysis for \ac{PLP} focuses on the perspective of an automotive radar mounted on a typical ego vehicle located at the origin of the two-dimensional plane, as shown in Fig.~\ref{fig:fig_1} (b). The radars have a half-power beamwidth of $\Omega$, and the target is assumed to be located at a distance of $R$ from the radar on the same street within the main beam of the ego radar.

The street containing the ego radar is denoted by $L_0$ with parameters $(\theta_0,r_0) = (0,0)$. As per the Palm distribution of a \ac{PLP}, the statistics of the rest of the process remain unaltered~\cite{dhillon2020poisson}. The locations of the vehicles with mounted radars on the $L_{i}\textsuperscript{th}$ street follow a one-dimensional \ac{PPP}, $\Phi_{{\rm L}_i}$, with intensity $\lambda$. Naturally, we can use a higher $\lambda$ to model the traffic conditions during peak hours compared to off-peak hours. The \ac{PPP} of vehicles on any street is assumed to be independent of the corresponding distributions on the other streets. Hence, the complete distribution of the vehicles (other than the ego vehicle/radar) is a homogeneous \ac{PLCP} $\Phi$, on the domain $\mathcal{P}_{\rm P}$ where the \ac{PLCP} is defined as 
    $\Phi_{\rm P} = \bigcup_{{\rm L}_i \in \mathcal{P}_{\rm P}} \Phi_{{\rm L}_i}$.

Fig.~\ref{fig:fig_1}(b) shows one realization of the \ac{PLCP}, where the ego radar is present at the origin on the road aligned along the $y$ axis. The other streets in the neighborhood are shown as lines in the Fig. The \ac{PLCP} points representing automotive radars are situated on these streets, and they may or may not contribute to the interference experienced by the ego-radar. This will be formally characterized subsequently. The ego radar will experience interference if and only if both the ego radar and interfering radar fall into each other's radar sectors simultaneously. To model the interference from vehicles from the opposite side of the road, we consider on line $L_0$ the location of vehicles follows a one-dimensional \ac{PPP} having the same intensity as the remaining lines. Therefore, the overall spatial stochastic process is represented as $\Phi_{{\rm P}_0}$, which as per the Palm conditioning, is the superposition of $\Phi_{\rm P}$, and an independent 1D \ac{PPP} on line $L_0$, i.e., $\Phi_{{\rm P}_0} = \Phi_{\rm P} \cup \Phi_{L_0}$.

\subsubsection{\ac{BLCP}}
Contrary to a \ac{PLP}, a \ac{BLP} $\mathcal{P}_{\rm B}$ is a finite collection of lines, i.e., $n_{\rm B}$ in the two-dimensional Euclidean plane, defined as, $\mathcal{P}_{\rm B} = \left\{L_1, L_2, \dots, L_{n_{\rm B}}\right\}$. Each line of $\mathcal{P}_{\rm B}$ corresponds to a point in the \ac{BPP} $\Phi_{\mathcal{D}_{\rm B}}$. The set $\mathcal{P}_{\rm B}$ is generated by points on a finite cylinder $\mathcal{D}_{\rm B}:=$ [$0,\pi$) $\times$ $[-R_{\rm g}, R_{\rm g}]$, which is the generating set of $\mathcal{P}_{\rm B}$. The point on the cylinder is denoted by $(\theta_i, r_i) \in \mathcal{D}_{\rm B}$ and corresponds to a line $L_i \in \mathcal{P}_{\rm B}$. Thus, in Euclidean space, the generating points of the lines are restricted to the disk $\mathcal{C}\left((0,0), R_{\rm g}\right)$. The normal to the line $L_i$ is formed by drawing a line segment from the origin to $(\theta_i, r_i)$. It is important to note that due to the finite domain of the generating set, the resulting line process is non-homogeneous in the Euclidean plane. Like \ac{PLCP}, we define on each line $L_i$ of $\mathcal{P}_{\rm B}$, an independent 1D \ac{PPP} $\Phi_{{\rm L}_i}$ with intensity $\lambda$. A \ac{BLCP} $\Phi_{\rm B}$, is the collection of all such points on all lines of the \ac{BLP}, i.e., 
    $\Phi_{\rm B} = \bigcup\limits_{i=1}^{n_{\rm B}} \Phi_{{\rm L}_i}$.

Thus, \ac{BLCP} like $\ac{PLCP}$ is a Cox process of random points defined on random lines. Due to the underlying inhomogeneity, the statistics of the \ac{BLP} cannot be characterized by perspective at just one typical point, such as one located at the origin. The isotropic structure of \ac{BLCP} implies that its characteristics, when observed from a specific point, are determined by the distance from the origin rather than the orientation of the point. A test point at $\left(0, r_0\right)$ is considered without limiting our generalization. Now, from the palm perspective of the \ac{BLCP}~\cite{shah2024binomial}, conditioning on a point to be located at $\left(0, r_0\right)$ in a \ac{BLCP} is the same as taking into consideration an atom at $\left(0, r_0\right)$, and a 1-D \ac{PPP} on a line $L_0$ that passes through $\left(0, r_0\right)$. Thus, the reduced stochastic process consists of $n_{\rm B} - 1$ lines defined in the same domain. Vehicles from the opposite side of the road are modeled as 1D \ac{PPP} on line $L_0$ having the same intensity as the remaining lines. The overall spatial stochastic process is represented as $\Phi_{{\rm B}_0}$.

Fig.~\ref{fig:fig_1}(c) shows an instance of the \ac{BLCP}, where the ego radar is located at the $(0,r_0)$, and the disk of radius $R_{\rm g}$ which restricts the generating point of lines. Besides, we see a beam of green color of radius $R_{\rm B}$, which is the maximum range of the ego radar. The remaining \ac{BLCP} points represent automotive radars, which may or may not cause interference. The red dots signify the interfering radars, while the blue dots represent the non-interfering radar. For the ego radar to experience interference, both the ego radar and interfering radar must be within a distance of $R_{\rm B}$ from each other simultaneously. Similar to \ac{PLCP}, we assume that each radar has a half-power beamwidth of $\Omega$, and the target is located at a distance of $R$ from the radar on the same street.
\begin{remark}
In our analysis, we assume that the street containing the ego radar passes through the center of the city, i.e., the origin of the generating disk. While this may not always be the case in real-world deployments, this modeling choice enables the ego radar to traverse the entire urban domain radially, capturing all possible locations from the city center to the outskirts. This setup allows for comprehensive characterization of spatial effects, including how detection probability and interference levels vary with respect to urban density, beam orientation, and distance from the center.
\end{remark}

\subsection{SIR and Channel Model}
The transmitted radar signal reflected from the target at distance $R$ undergoes losses and is assumed to have a Swerling-I fluctuating radar cross-section, $\sigma_{\mathbf{c}}$, with a mean of $\bar{\sigma}$~\cite{shnidman2003expanded}. The reflected signal power decay is modeled as $L(r) = R^{-2\alpha_{\rm L}}$, where $\alpha_{\rm L}$ is the \ac{PLE} for \ac{LOS} case, as we assume that the target is always in \ac{LOS}. Let $G_t$ be the gain of the transmitting antenna and $A_{\rm e}$ the effective area of the receiving antenna aperture. Then, the ego radar receives the reflected signal power from the target vehicle with strength
\begin{align}
    S = \gamma\sigma_{\mathbf{c}} P R^{-2\alpha_{\rm L}}
\end{align}
where $\gamma = \frac{G_{\rm t}}{(4\pi)^2}A_{\rm e}$ and $P$ is the transmit power. Let the coordinates of any \ac{PLCP} point/interfering radar in the Euclidean plane be denoted by $\mathbf{w}_{\rm P} = (x,y)$, such that $x\cos\theta_i + y\sin\theta_i = r_i$ for $(\theta_i,r_i) \in \mathcal{D}_{\rm P}$. Likewise $\mathbf{w}_{\rm B} = (x,y)$ be the coordinates of any \ac{BLCP} point/interfering radar, such that $x\cos\theta_i + y\sin\theta_i = r_i$ for $(\theta_i,r_i) \in \mathcal{D}_{\rm B}$. Owing to the effects of the mm-wave channel, the interfering signals experience multi-path fading, which is represented by the Nakagami-m fading channel model due to its analytical tractability. The fading channel power $h_{\mathbf{w}_k, \rm L}$ and $h_{\mathbf{w}_k, \rm N}$ for \ac{LOS} and \ac{NLOS} case respectively, follow gamma distribution. Precisely for any \ac{LOS} interfering radar $h_{\rm L} \sim \Gamma \left(m_{\rm L}, 1/m_{\rm L }\right)$ and for \ac{NLOS} interfering radar $h_{\rm N} \sim \Gamma \left(m_{\rm N}, 1/m_{\rm N}\right)$, where $m_{\rm L}$ and $m_{\rm N}$ are fading parameters. Furthermore, the fading gain is assumed to be independent across all \ac{LOS} and \ac{NLOS} radars. The interference at the ego radar for both \ac{PLCP} and \ac{BLCP} due to any one interfering radar located at $w_{k,i}$ is then given by,
\begin{align}
    \mathbf{I}_k &= P \gamma h_{\mathbf{w}_k, \rm L} ||\mathbf{w}_k||^{-\alpha_{\rm L}} \times \mathbb{I} (i = 0) + \nonumber\\
    & \hspace*{2cm} P \gamma h_{\mathbf{w}_k, \rm N} ||\mathbf{w}_k||^{-\alpha_{\rm N}} \times \mathbb{I} (i \ne 0),
\end{align}
where $h_{\mathbf{w}_k, \rm L}$ and $h_{\mathbf{w}_k, \rm N}$ is the fading power for \ac{LOS} and \ac{NLOS} case respectively. In our model, we assume that interfering radars located on the same street $L_0$ (i.e. $i=0$) as the ego radar are in \ac{LOS}, and thus their \ac{PLE} is denoted by $\alpha_{\rm L}$. Conversely, interfering radars situated on intersecting streets (i.e. $i \ne 0$) are considered to be in \ac{NLOS}, with their \ac{PLE} is denoted by $\alpha_{\rm N}$. As highlighted in~\cite{maccartney2013path}, mm-wave channels may exhibit significantly different \ac{PLE} values under \ac{LOS} and \ac{NLOS} conditions, underscoring the importance of distinguishing between these propagation scenarios in radar interference analysis. Now, assuming that all the automotive radars share the same power and gain characteristics, the \ac{SIR} at the ego radar is
\begin{align}
    \xi_k = \frac{\gamma \sigma_{\mathbf{c}} P R^{-2\alpha_{\rm L}}}{\sum_{\mathbf{w}_k \in \Phi_k^{\rm I}}4\pi \gamma P h_{\mathbf{w}_k, l} ||\mathbf{w}_k||^{-\alpha_{l}}},
    \label{eq:eq_3}
\end{align}
where $k \in \{{\rm P},{\rm B}\}$, $l \in \{{\rm L}, {\rm N}\}$ and $\Phi_k^{\rm I}$ is the interfering set defined in next section. It is important to note that $\xi_{\rm B}$ is a function of $r_0$. Our analytical framework assumes identical transmit power and waveforms to establish a fundamental baseline for spatial interference characterization. In practical deployments, waveform diversity (e.g., varying FMCW chirp slopes or orthogonal PMCW codes) reduces the cross-correlation between the ego and interfering signals. From a modeling perspective, this introduces an orthogonality factor $\epsilon \in [0,1)$ into the interference summation in~\eqref{eq:eq_3}, effectively scaling down the aggregate interference power and improving the detection probability. Similarly, adaptive power control, where transmit power scales with target distance would reduce the interference contribution from nearby vehicles detecting close-range targets. While these mechanisms would shift the absolute performance curves favorably, the relative spatial trends governed by the street geometry and urban density derived in this work remain structurally invariant.

For our analysis of \ac{PLCP} or \ac{BLCP}, we make the crucial assumption that the radar has a maximum allowable range or operational range, which we denote as $R_k$. This implies that the ego radar cannot detect any target beyond the $R_k$ range and that any other radar not inside the $R_k$ range of the ego radar sector will not interfere with the ego radar's detection performance. Thus, the distance between the ego radar and interfering radars is limited, i.e., $||\mathbf{w}_k|| \leq R_k$.
\begin{remark}
The distinction of \ac{LOS} propagation on the street $L_0$ and \ac{NLOS} conditions on intersecting streets is a logical assumption to consider. In dense urban environments, due to unobstructed alignment along the roadway corridor, the street $L_0$ has direct visibility (\ac{LOS}) in typical city street networks. However intersecting streets are prone to \ac{NLOS} propagation because of natural blockage at street corners  This distinction between \ac{LOS} and \ac{NLOS} scenarios is critical for accurately modeling radar interference in urban environments, as it reflects the realistic spatial distribution of signal propagation paths and their associated \ac{PLE} for \ac{LOS} and \ac{NLOS}.
\end{remark}
\begin{remark}
While the current analysis adopts a deterministic \ac{LOS}/\ac{NLOS} classification based on urban street geometry, more advanced models can incorporate ball blockage model to better reflect the stochastic nature of mm-wave propagation. Inspired by earlier frameworks, a probabilistic model introduces the concept of an \say{equivalent \ac{LOS} ball} of radius $d$, within which a radar-link has a \ac{LOS} probability $p_{\rm L}(r)$, and beyond which the \ac{NLOS} probability increases, denoted $p_{\rm N}(r)$. This leads to a distance-dependent characterization of blockage, with $p_{\rm L}(r)$ potentially following an exponential decay function. Under this approach, each potential interferer is probabilistically marked as \ac{LOS} or \ac{NLOS}. While we retain the deterministic assignment to preserve analytical tractability and consistency with the geometric framework, we numerically quantify the deviation between our deterministic assumption and a fully probabilistic blockage model in Section V-A.
\end{remark}

\begin{figure}[t]
    \centering
    \includegraphics[width=0.65\linewidth]{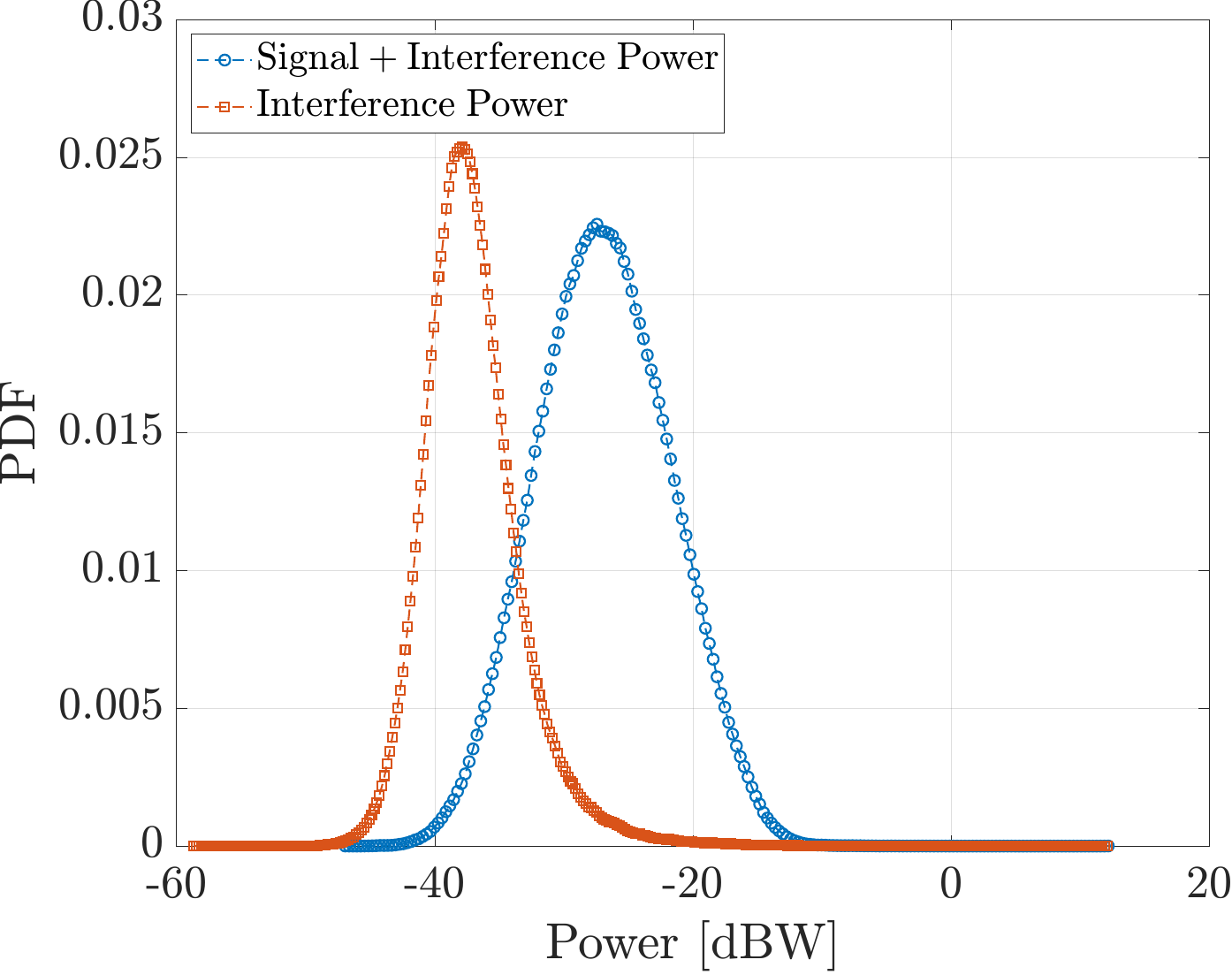}
    \caption{PDF of signal + interference power and only interference power received at ego radar, if the target is at a distance uniformly distributed between $5$ to $15\, {\rm m}$, $\Omega = 30^\circ$, $\lambda_{\rm L} = 0.01\, {\rm m}^{-1}$, $\lambda = 0.05\, {\rm m}^{-1}$, and transmit power is $1$ dB.}
    \label{fig:si_and_i}
\end{figure}

\subsection{Importance of \ac{SIR}}
Let us now illustrate the significance of characterizing the \ac{SIR} in radar system design. In traditional radar system analysis, key performance metrics are the probability of false alarm and the probability of detection, together referred to as the receiver operating characteristics (ROC). The probability of detection characterizes the event that the received signal power added with the noise and interference exceeds a predefined threshold. Whereas, the probability of false alarm quantifies the event that the sum of the interference and noise itself exceeds that threshold, leading to an erroneous detection decision in the absence of a true target. While these metrics remain fundamental, the overarching objective of this work is to gain deeper insights into the statistical behavior of the \ac{SIR}, which directly influences detection performance. Specifically, we aim to characterize the distribution of the \ac{SIR} experienced by the ego radar in a networked environment, as it provides a foundational understanding of how the signal strength relative to interference impacts detection outcomes.

To illustrate this, Fig.~\ref{fig:si_and_i} presents the \acp{PDF} of the signal plus interference strength received at ego radar, along with that of only the interference power, obtained via Monte Carlo simulations. These \acp{PDF} highlight the separation between the two distributions, which is critical for determining an appropriate detection threshold. In general, when designing a detector in the absence of prior information about the target state, one can choose the threshold based on the desired trade-off between detection and false alarm probabilities. A larger separation between the two \acp{PDF} allows for greater flexibility in selecting the threshold while keeping the false alarm rate within acceptable bounds. Although classical metrics like false alarm probability and ROC analysis are valuable tools for evaluating detection performance, they are inherently tied to the underlying \ac{SIR} distribution. By focusing on the characterization of the CDF of the \ac{SIR}, we provide a more fundamental perspective on how the relative strengths of the signal and interference influence detection outcomes. 

\section{Characterization of the Interference Process}
\label{sec:IntDis}
We observe from Fig~\ref{fig:fig_1} (b) and (c) that all the Cox points do not cause the interference. Instead, interference is caused at the ego radar only when both the ego radar and the interfering radar are simultaneously within each other's radar sectors. For example, in Fig~\ref{fig:fig_2}, we see the radar sectors (soon to be defined) emitted by four vehicles, namely \textit{A}, \textit{B}, \textit{C}, \textit{D}, and \textit{E}. The vehicle \textit{A}, located at the origin $\mathbf{0}$, represents the ego radar, and its radar sector is depicted in green. The beams of the interfering and non-interfering vehicles are represented by purple and red, respectively. Therefore, ego radar experiences interference only from \textit{B} and \textit{E}, and not from  \textit{C} and \textit{D}. In the subsequent discussion, we characterize the interfering distance and determine its effect on the system performance. 

We shall first define the subset of $\mathbb{R}^2$ defined by the conic beam (henceforth referred to as the radar sector) and then determine the collection of Cox points located within the radar sector formed by the radar beam.

\subsection{Interfering Set}
The orientation of the vehicles on the street depends on the generating angle of the street and the direction of their movement. Thus, the boresight direction of any radar on $i\textsuperscript{th}$ line is given by the two unit vectors: $\mathbf{a}$ and $-\mathbf{a}$, where $\mathbf{a} = (-\sin\theta_i, \cos\theta_i)^T$. We see that the parameter of the line $L_i$, i.e., $(\theta_i,r_i)$, determines the direction of the radar sector. 
For any point in the Euclidean plane $(p,q)$, the cosine of the angle formed by the displacement vector $(p,q) - (x,y)$  and the boresight direction $\mathbf{a}$ must be greater than $\cos\Omega$, ensuring that the point lies within the radar's angular field-of-view. Thus, the sector for any radar is uniquely characterized by $\mathcal{S}^{\rightarrow}_{\mathbf{u},k}$ and $\mathcal{S}^{\leftarrow}_{\mathbf{u},k}$ as a function of $\mathbf{u} = (x,y)$ and $\Omega$. Thus we can define the \textit{closed} interior region of the radar sector for any Cox point of \ac{PLCP} or \ac{BLCP} $(x,y)$ in the Euclidean plane as follows:

\begin{definition}
\label{def:def_1}
Consider a radar located at $\mathbf{u}=(x,y)\in\Phi_{k_0}$ whose boresight orientation is given by the unit vector $\mathbf{a} = [\cos\theta_i,\; \sin\theta_i]^T$. The forward sensing region of this radar is defined as the set of locations $\mathbf{z}=(p,q)\in\mathbb{R}^2$ that lie inside its field of view and within its usable range, i.e.,
\begin{align*}
    \mathcal{S}^{\rightarrow}_{\mathbf{u},k} \!=\! \Big\{\mathbf{z}\in\mathbb{R}^2 \!\colon\! \langle \mathbf{a}, \mathbf{z} - \mathbf{u}\rangle \!>\! \|\mathbf{z}-\mathbf{u}\| \cos\Omega, \|\mathbf{z}-\mathbf{u}\| \le R_k \!\Big\}.
\end{align*}
The backward sector corresponds to the opposite boresight direction $-\mathbf{a}$ and is defined analogously as
\begin{align*}
    \mathcal{S}^{\leftarrow}_{\mathbf{u},k} \!=\! \Big\{\mathbf{z}\in\mathbb{R}^2 \!\colon\! \langle -\mathbf{a}, \mathbf{z} - \mathbf{u}\rangle \!>\! \|\mathbf{z} -\mathbf{u}\| \cos\Omega, \|\mathbf{z} - \mathbf{u}\| \le R_k \!\Big\}.
\end{align*}
\end{definition}
\noindent For the \ac{PLCP} model, the ego radar is positioned at $\mathbf{u}=(0,0)$,
whereas in the \ac{BLCP} model its coordinate is $(0,r_0)$. In both cases,
the sector is spatially bounded due to the constraint $\|\mathbf{z}-\mathbf{u}\| \le R_k$.
\\ 
{\bf Example:} 
Fig.~\ref{fig:fig_2} shows the ego radar sector as $\mathcal{S}^{\rightarrow}_{\mathbf{u} = (0,0),{\rm P}}$ in green color with $\mathbf{a} = (0,1)$. Similarly, on the line $L_i$, the radar sector, $\mathcal{S}^{\leftarrow}_{\mathbf{u},{\rm P}}$, of an interfering vehicle \textit{B}, located at $(x,y) \in \Phi_{\rm P}$ is illustrated by the magenta color. We have also illustrated the two radar sectors of the vehicle \textit{C} for two boresight vectors, i.e., $\mathbf{a}$ and $-\mathbf{a}$. Thus, the ego radar at any position only experiences interference from \ac{PLCP} points in its radar sector, and the ego radar is in their radar sector. If the ego radar’s boresight is $\mathbf{a}$, only radars with boresight $-\mathbf{a}$ contribute to interference. The same arguments follow for the \ac{BLCP} model, with the only difference being the location of ego radar. This makes the analysis of \ac{BLCP} more involved in a location-dependent framework. We generalize the set of interfering points as:

\begin{definition}[Set of Interfering Radars]
\label{def:new_def2}
For any $\Phi_{k_0}$ corresponding to model $k\in\{{\rm P},{\rm B}\}$ the set of Cox points $\Phi^{\mathrm I}_k$ that cause interference at the ego radar is
\begin{align*}
    \Phi^{\mathrm I}_k = \Big\{\mathbf{u}\in \Phi_{k_0} \colon\, \Xi_{\mathbf{u}}^{\rightarrow} \cdot \Xi_{\mathbf{u}_0}^{\leftarrow} = 1 \Big\},
\end{align*}
where the ego radar is located at $\mathbf{u}_0=(0,r)$ with $r=0$ for PLCP and $r=r_0$ for BLCP. The indicator terms are given by
\begin{align*}
    \Xi_{\mathbf{u}}^{\rightarrow} &= \mathbf{1} \left(\mathbf{u}_0 \in \mathcal{S}^{\rightarrow}_{\mathbf{u},k}\right) + \mathbf{1} \left(\mathbf{u}_0 \in \mathcal{S}^{\leftarrow}_{\mathbf{u},k}\right), \\
    \Xi_{\mathbf{u}_0}^{\leftarrow} &= \mathbf{1} \left(\mathbf{u} \in \mathcal{S}^{\rightarrow}_{\mathbf{u}_0,k}\right) + \mathbf{1} \left(\mathbf{u} \in \mathcal{S}^{\leftarrow}_{\mathbf{u}_0,k}\right),
\end{align*}
where $\mathbf{1}\left(\cdot\right)$ is the indicator function.
\end{definition}
The above definition characterizes the set of automotive radars that generate interference at the ego radar. A radar located at $\mathbf{u}\in\Phi_{k_0}$ contributes to interference only if a \textbf{mutual visibility condition} is satisfied. Specifically, (i) the ego radar location $\mathbf{u}_0$ must lie inside either the forward or backward sensing sector of the radar at $\mathbf{u} = (x,y)$, and (ii) simultaneously, the radar at $\mathbf{u}$ must lie inside either the forward or backward sensing sector of the ego radar. These two conditions are encoded through the indicator variables $\Xi_{\mathbf{u}}^{\rightarrow}$ and $\Xi_{\mathbf{u}_0}^{\leftarrow}$ in Definition~\ref{def:new_def2}. Their product enforces a reciprocity constraint, ensuring that interference is possible only when the two radars are geometrically aligned and mutually within each other's effective sensing regions.

\subsection{Interfering Distance}
Our objective is to determine the impact of the interfering radars located on the ego radar's street and the remaining streets on the detection performance of the ego radar. To do so, \emph{we determine the fraction length of $L_i, \forall\, L_i \in \{\mathcal{P}_{\rm P},\mathcal{P}_{\rm B}\}$ wherein if a radar is present, it will contribute to the interference experienced at the ego radar.} First, note that the distance from the ego radar to the intersection of the $L_i\textsuperscript{th}$ line is 
\begin{align}
   d_i = u_i - r_0 = \frac{r_i}{\sin\theta_i} - r_0,
   \label{eq:eq_d}
\end{align}
where the angle of intersection is equal to the generating angle $\theta_i$ of the intersecting line and $r_0 = 0$ in the case of \ac{PLCP}. The distance between the interfering radar on $L_i$ from the intersection point is $v_{k,i}$, referred to as the interfering distance if the other radar and ego radar mutually interfere. Fig.~\ref{fig:fig_3} illustrates an ego radar with a green radar sector on $L_1$ and radars mounted on other vehicles, \textit{S} and \textit{T}, on $L_i$, with red and purple beams. Specifically, \textit{S} and \textit{T} are at a distance $v_{k,i} = b_{k,i}$ and $v_{k,i} = a_{k,i}$ from the intersection point, respectively. These two distances ($a_{k,i}$ and $b_{k,i}$) represent the bounds of $v_{k,i}$ within which another radar will interfere with the ego radar. From Fig~\ref{fig:fig_3}, any other vehicle present behind the vehicle \textit{S} will not cause interference. Likewise, any vehicle present after vehicle \textit{T} also will not cause interference. This is because the vehicles present behind \textit{S} and after \textit{T} will not mutually interfere with ego radar.

From definition~\ref{def:def_1}, the radar sector of ego vehicle is bounded as the maximum range $R_k$ to interferers is limited by either blockages or by the maximum detectable/unambiguous range of the radar. Thus, $R_k$ must be considered to find the interfering distance. The following Theorem~\ref{th:theo1} gives the maximum and minimum interfering distance i.e., $a_{{\rm B},i}$ and $b_{{\rm B},i}$ for \ac{BLCP}.
\begin{figure}[t]
    \centering
    \includegraphics[trim={4.4cm 4.3cm 4.55cm 1.85cm},clip,width = 0.33\textwidth]{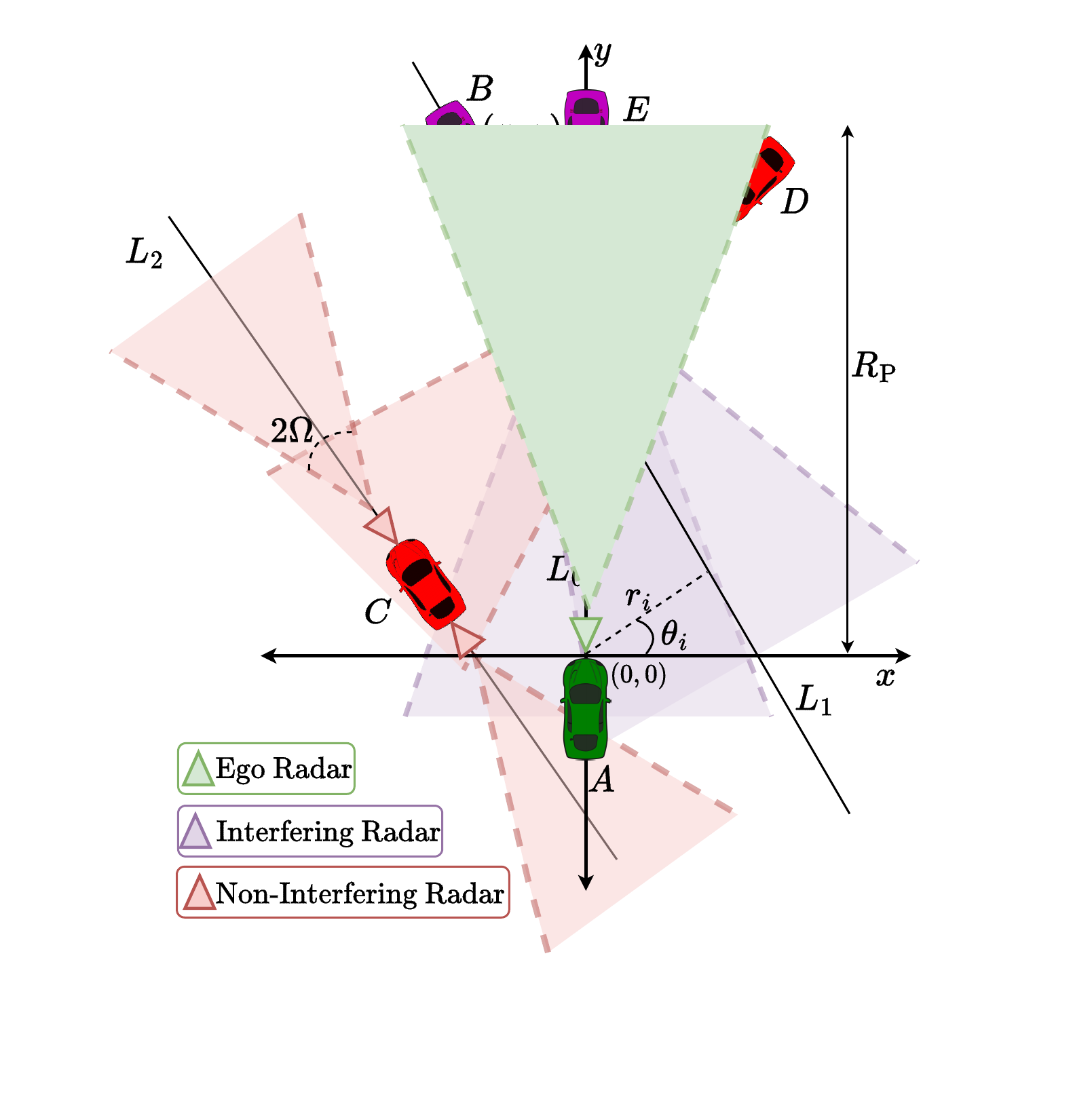}
    \caption{Illustration of a scenario showing interfering and non-interfering radar sectors w.r.t. ego radar.}
    \label{fig:fig_2}
\end{figure}
\begin{theorem}
\label{th:theo1}
For $L_i$ line with parameters $(\theta_i,r_i)$ intersecting $L_0$ at a distance $d_i$ from the ego radar, the maximum ($a_{{\rm B},i}$) and the minimum ($b_{{\rm B},i}$) distances of the interfering radars from the point of intersection on $L_i$, for \ac{BLCP} are given as:
\begin{align}
    a_{{\rm B},i} &= \nonumber \\
    &\hspace*{-0.5cm}
    \begin{cases}
        R; \hspace*{4.5cm} \mathrm{for}\; (\theta_0,r_0) = (0, 0)\\
        d_i\left(|\sin\theta_i|\cot\Omega - |\cos\theta_i|\right); \hspace*{0.72cm} \mathrm{for}\; 0 \leq d_i \leq c_2,\\
        \hspace*{0.5cm} \;\mathrm{and}\; \theta_i \in [0, 2\pi] \setminus \big\{[\Omega, \pi - \Omega] \cup [\pi + \Omega, 2\pi - \Omega]\big\}\\
        \frac{d_i|\sin\left(\theta_i-\Omega\right)|}{\sin\Omega} \hspace*{3.15cm} \mathrm{for}\; 0 \leq d_i \leq c_2,\\
        \hspace*{0.5cm} \;\mathrm{and}\; \theta_i \in \big\{[\Omega, 2\Omega] \cup [\pi+\Omega,\pi+2\Omega]\big\} \\
        \frac{d_i|\sin\left(\theta_i+\Omega\right)|}{\sin\Omega} \hspace*{3.15cm} \mathrm{for}\; 0 \leq d_i \leq c_2,\\
        \hspace*{0.5cm} \;\mathrm{and}\; \theta_i \in \big\{\pi - 2\Omega, \pi - \Omega] \cup [2\pi-2\Omega,2\pi-\Omega]\big\} \\
        \frac{d_i}{|\sin\theta_i|\cot\Omega - |\cos\theta_i|}; \hspace*{2.1cm} \mathrm{for}\; c_1 \leq d_i < 0,\\
        \hspace*{0.5cm} \;\mathrm{and}\; \theta_i \in [0, 2\pi] \setminus \big\{[\Omega, \pi - \Omega] \cup [\pi + \Omega, 2\pi - \Omega]\big\}\\
        0; \hspace*{4.65cm} \mathrm{otherwise}
    \end{cases}
    \label{eq:l_1}
\end{align}
\begin{align}
    b_{{\rm B}, i} &= \nonumber\\
    &\hspace*{-0.5cm}
    \begin{cases}
        R_{\rm B}; \hspace*{4.5cm} \mathrm{for}\; (\theta_0,r_0) = (0, 0)\\
        \sqrt{R_{\rm B}^2 - (d_i \sin\theta_i)^2} - d_i|\cos\theta_i|; \hspace*{0.37cm} \mathrm{for}\; c_1 \leq d_i \leq c_2, \\
        \hspace*{0.2cm} \mathrm{and}\; \theta_i \in [0,2\pi] \setminus \big\{[2\Omega, \pi - 2\Omega] \cup [\pi + 2\Omega, 2\pi-2\Omega] \big\} \\
        \frac{d_i \tan\Omega}{|\sin\theta_i| - \tan\Omega|\cos\theta_i|};\hspace*{2.32cm} \mathrm{for}\; 0 \leq d_i \leq c_1 \\
        \hspace*{0.2cm} \mathrm{and}\; \theta_i \in \big\{[\Omega, 2\Omega] \cup [\pi-2\Omega, \pi - \Omega]\big\}\\
        0; \hspace*{4.85cm} \mathrm{otherwise}
    \end{cases}
    \label{eq:l_2}
\end{align}
where,
\begin{align*}
    &c_1 = R_{\rm B}\sin\Omega\left(\cot\Omega - |\cot\theta|\right), \;\; c_2 = R_{\rm B}|\csc\theta|\sin\Omega
\end{align*}
\end{theorem}

\begin{proof}
See appendix~\ref{pr:thm_1}.
\end{proof}
For the \ac{PLCP} model, the values of $a_{{\rm P},i}$ and $b_{{\rm P},i}$ are structurally same as $a_{{\rm B},i}$ and $b_{{\rm B},i}$, but the only difference lies in parameters $d_i$, $R_{\rm P}$, and most importantly the cases for different angles of $\theta_i$. For the \ac{PLCP}, $d_i$ reduces to $u_i$ as $r_0 = 0$ from~\eqref{eq:eq_d}. In~\eqref{eq:l_1} and~\eqref{eq:l_2}, the set of values of $\theta_i$ for all the different cases has a larger range as compared to different cases seen in Lemma~\ref{le:lemma1}.
For example, $a_{{\rm B},i} = d_i\left(|\sin\theta_i|\cot\Omega - |\cos\theta_i|\right)$ if $ \theta_i \in [0, 2\pi] \setminus \big\{[\Omega, \pi - \Omega] \cup [\pi + \Omega, 2\pi - \Omega]\big\}$, while as in Lemma~\ref{le:lemma1}, for the same scenario, the range of $\theta_i$ is $\big\{[0,\Omega] \cup [\pi-\Omega,\pi]\big\}$. 
Because of the non-homogeneous construction of \ac{BLP}, we have to carefully consider all the different cases, which leads to a larger range of $\theta_i$. As $r_0$ in \ac{BLCP} is dynamic thus, the values of $a_{{\rm B},i}$ and $b_{{\rm B},i}$ will exist for such cases where the values of $a_{{\rm P},i}$ and $b_{{\rm P},i}$ do not exist. 

{\bf Example:} Consider the following example. In \ac{PLCP}, the ego radar is at the origin, and if $\theta_i > \pi$, then line $L_i$ does not intersect $L_0$ ahead of the ego radar, i.e., in $y >0$. On the contrary, for the same range of generating angles, the location of ego radar can be such that line $L_i$ can intersect the $L_0$ ahead of the ego radar. This complexity introduced by $r_0$ leads to a careful consideration of $\theta_i$ in determining the interfering distance in \ac{BLCP} in comparison to \ac{PLCP}.
\begin{lemma}
\label{le:lemma1}
For line $L_i$ that intersects $L_0$ at a distance $d_i$ from the ego radar, the maximum ($a_{{\rm P},i}$) and the minimum ($b_{{\rm P},i}$) distances of the interfering radars from the point of intersection on $L_i$, are given as:
\begin{align}
    a_{{\rm P},i} &= \nonumber \\
    &\hspace*{-0.5cm}
    \begin{cases}
        R; \hspace*{4.35cm} \mathrm{for}\; (\theta_0,r_0) = (0, 0)\\
        u_i\left(|\sin\theta_i|\cot\Omega - |\cos\theta_i|\right); \hspace*{0.55cm} \mathrm{for}\; 0 \leq u_i \leq c_2,\\
        \hspace*{2.4cm} \;\mathrm{and}\; \theta_i \in \big\{[0,\Omega] \cup [\pi-\Omega,\pi]\big\}\\
        \frac{u_i|\sin\left(\theta_i-\Omega\right)|}{\sin\Omega}; \hspace*{0.5cm} \mathrm{for}\; 0 \leq u_i \leq c_2, \;\mathrm{and}\; \theta_i \in [\Omega, 2\Omega] \\
        \frac{u_i|\sin\left(\theta_i+\Omega\right)|}{\sin\Omega}; \hspace*{0.5cm} \mathrm{for}\; 0 \leq u_i \leq c_2, \\
        \hspace*{3.71cm}\mathrm{and}\; \theta_i \in [\pi - 2\Omega, \pi - \Omega] \\
        \frac{u_i}{|\sin\theta_i|\cot\Omega - |\cos\theta_i|}; \hspace*{0.85cm} \mathrm{for}\; c_1 \leq u_i < 0\\
        \hspace*{2.4cm} \;\mathrm{and}\; \theta_i \in \big\{[\pi, \pi + \Omega], [2\pi-\Omega, 2\pi]\big\}\\
        0; \hspace*{4.4cm} \mathrm{otherwise}
    \end{cases}
    \label{eq:l_1_p}
\end{align}
\begin{align}
    b_{{\rm P}, i} &= \nonumber\\
    &\hspace*{-0.5cm}
    \begin{cases}
        R_{\rm P}; \hspace*{4.69cm} \mathrm{for}\; (\theta_0,r_0) = (0, 0)\\
        \sqrt{R_{\rm P}^2 - (u_i \sin\theta_i)^2} - u_i|\cos\theta_i|; \hspace*{0.5cm} \mathrm{for}\; c_1 \leq u_i \leq c_2, \\
        \hspace*{0.5cm} \mathrm{and}\; \theta_i \in \big\{[0,\Omega] \cup [\pi - \Omega, \pi+\Omega] \cup [2\pi-\Omega,2\pi]\big\} \\
        \frac{u_i \tan\Omega}{|\sin\theta_i| - \tan\Omega|\cos\theta_i|};\hspace*{2.45cm} \mathrm{for}\; 0 \leq u_i \leq c_1 \\
        \hspace*{0.5cm} \mathrm{and}\; \theta_i \in \big\{[\Omega, 2\Omega] \cup [\pi-2\Omega, \pi - \Omega]\big\}\\
        0; \hspace*{5cm} \mathrm{otherwise}
    \end{cases}
    \label{eq:l_2_p}
\end{align}
where,
\begin{align*}
    &c_1 = R_{\rm P}\sin\Omega\left(\cot\Omega - |\cot\theta|\right), \;\; c_2 = R_{\rm P}|\csc\theta|\sin\Omega.
\end{align*}
\end{lemma}
\begin{figure}[t]
\centering
\includegraphics[trim={2.5cm 1.75cm 7.5cm 3.5cm},clip,width = 0.3\textwidth]{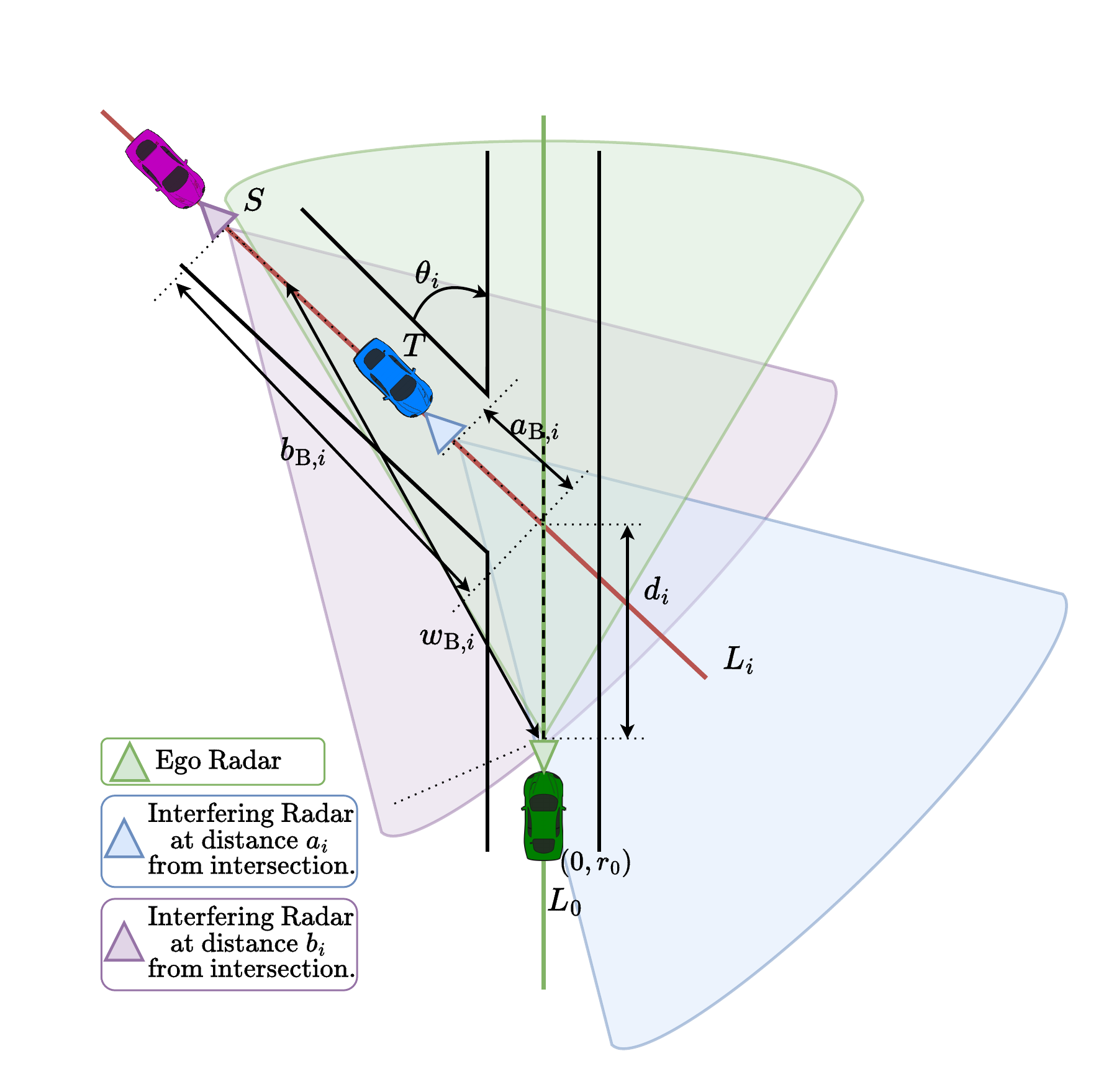}
\caption{Illustration of a scenario where two radars are present at the edge point of the line $L_i$ inducing interference.}
\label{fig:fig_3}
\end{figure}
With all the possible cases where we derived the interfering distance, i.e., $v_{k, i}$, correspondingly, the distance from the ego radar to the interfering radar is then given as
\begin{align}
    w_{k, i} &=
    \begin{cases}
        v_{k, i} & i=0 \\
        \sqrt{(d_i + v_{k, i}|\cos\theta_i|)^2 + (v_{k, i}\sin\theta_i)^2}\; &  \mathrm{otherwise}.
    \end{cases}
    \label{eq:eq_w_p}
\end{align}

To quantify the average number of interferers experienced by the ego radar, we build upon the analytical framework established in Theorem~\ref{th:theo1} and Lemma~\ref{le:lemma1}, by finding the average length of the portion of each line $L_i$ that falls within the radar's field of view. This is determined by averaging the bounds $a_{k,i}$ and $b_{k,i}$. Next, we integrate this average line-segment length across the full set of lines in the Cox process, which involves integrating over all angular and radial configurations. This integration step aggregates the contributions from all possible spatial configurations of interfering vehicles relative to the ego radar's position. Finally, we multiply the resulting total line length by the vehicular intensity $\lambda$, which represents the average number of vehicles per unit length on each line, yielding the expected number of interferers as a function of the ego radar's location $r_0$. Specifically, for the BLCP case, the expected interferer count is given by:
\begin{align*}
    \mathcal{I}_{\mathrm{B}}(r_0) = \lambda \left(\frac{1}{2\pi R_{\mathrm{B}}} \int_0^{R_{\mathrm{B}}} \int_0^{2\pi} |b_{\mathrm{B}} - a_{\mathrm{B}}|\, d\theta\, dr \right)^{n_{\mathrm{B}}},
\end{align*}
where $|b_{\mathrm{B}} - a_{\mathrm{B}}|$ captures the active interfering length per line.

\section{Detection Success Probability}
\label{sec:RadarProb}
The detection \textit{success probability} at a threshold $\beta$ is defined as the \ac{CCDF} of SIR, $p_{{\rm D},{\rm P}}(\beta) = \mathbb{P}[\rm{SIR} > \beta]$. This represents the probability that an attempted detection by the ego radar of the target located at a distance $R$ is successful. 
The analytical framework in our work follows a unified pipeline. First, the street geometry is modeled using PLCP or BLCP. Conditioned on this geometry, the effective interfering radar set is defined by joint beam, range, and LOS/NLOS constraints. The resulting distance bounds directly determine the interference integrals, which are then used to derive the SIR distribution and detection probability. This structure is common to both models and enables a transparent comparison between homogeneous and spatially varying urban environments.
\begin{theorem}
\label{th:theo3}
For the network where locations of the vehicles are modeled as Cox process $\Phi_k$, such that $k \in \{{\rm P},{\rm B}\}$ the detection success probability for an ego radar is
\begin{align}
    &p_{{\rm D},k}(\beta) = \mathbb{E}_{\mathcal{D}_k} \vast(\!\!\prod_{(\theta,r) \in \mathcal{D}_k}\!\!\!\! \mathbb{E}_{\Phi} \Bigg( \prod_{\mathbf{w}_k \in \Phi_{k_0}} \frac{1}{1+\beta^\prime ||\mathbf{w}_k||^{-\alpha}} \Bigg)\vast)
    \label{eq:main_eq}
\end{align}
where $\beta^\prime = \frac{4\pi\beta}{\bar{\sigma} R^{-2\alpha_{\rm L}}}$, and $||\mathbf{w}_k||$ is given in~\eqref{eq:eq_w_p}.
\begin{proof}
From the definition of success probability, we have
{
\begin{align}
    &p_{{\rm D},k}(\beta) = \mathbb{P}\left(\frac{\gamma \sigma_{\mathbf{c}}PR^{-2\alpha_{\rm L}}}{\sum_{\mathbf{w}_k \in \Phi_{k_0}} 4\pi \gamma P h_{\mathbf{w}_k} ||\mathbf{w}_k||^{-\alpha_l}} > \beta\right) \nonumber\\
    &= \mathbb{P} \left(\sigma_{\mathbf{c}} > \frac{\beta\left(\sum_{\mathbf{w}_k \in \Phi_{k_0}} 4\pi \gamma P h_{\mathbf{w}_k} ||\mathbf{w}_k||^{-\alpha_l}\right)}{\gamma PR^{-2\alpha_{\rm L}}} \right) \nonumber\\
    &\hspace*{-0.25cm}\overset{(a)}{=} \mathbb{E}_{\Phi_{k_0}, h_{\mathbf{w}_k}} \left[\exp{\left(-\frac{\beta(\sum_{\mathbf{w}_k \in \Phi_{k_0}} 4\pi \gamma P h_{\mathbf{w}_k} ||\mathbf{w}_k||^{-\alpha_l})}{\bar{\sigma}\gamma PR^{-2\alpha_{\rm L}}}\right)}\right] \nonumber\\
    &\overset{(b)}{=} \mathbb{E}_{\mathcal{P}_k} \vast(\prod_{(\theta,r) \in \Phi_{\mathcal{D}_k}} \mathbb{E}_{\Phi_{\rm L}} \Bigg( \prod_{\mathbf{w}_k \in \Phi_{\rm L}} \nonumber\\
    &\hspace*{2cm}\mathbb{E}_{h_{\mathbf{w}_k}} \left[\exp{\left(-\frac{4\pi \beta P h_{\mathbf{w}_k} ||\mathbf{w}_k||^{-\alpha_l}}{\bar{\sigma} PR^{-2\alpha_{\rm L}}}\right)}\right] \Bigg)\vast) \nonumber\\
    &\overset{(c)}{=} \mathbb{E}_{\mathcal{P}_k} \vast(\prod_{(\theta,r) \in \Phi_{\mathcal{D}_k}} \mathbb{E}_{\Phi_{\rm L}} \Bigg( \prod_{\mathbf{w}_k \in \Phi_{\rm L}} \left(\frac{1}{1+\frac{\beta^\prime}{m_l} ||\mathbf{w}_k||^{-\alpha_l}}\right)^{m_l} \Bigg)\vast)
    \label{eq:thm_3}
\end{align}
}
Step (a) follows from the exponential distribution of $\sigma_{\mathbf{c}}$. In step (b), the expectation operator $\mathbb{E}_{\Phi_{k_0}}$ over the Cox processes is decomposed into two parts: first, the expectation over the 1D \ac{PPP} $\Phi_L$ (an instance of $\Phi_{L_i}$) and second, the expectation over the line process $\mathcal{P}_k$. Step (c) follows by averaging over the fading variable $h_{\mathbf{w}}$.
\end{proof}
\end{theorem}
\begin{corollary}
For the \ac{PLCP} with $m_{\rm L} = 2$ and $m_{\rm N} = 1$, the detection success probability is
\begin{align}
    &p_{{\rm D},{\rm P}}(\beta) = \exp\Bigg(\!\!-\lambda \int_{R}^{R_{\rm P}} 1 - \left(\frac{1}{1+\frac{\beta^\prime}{2} {v}_{\rm P}^{-\alpha_{\rm L}}}\right)^2 {\rm d}v_{\rm P} -\lambda_{\rm L} \int_{\mathbb{R}^{+}}\!\! \nonumber\\
    &\hspace*{0.1cm} \int_0^{2\pi} 1 - \exp\bigg(-\lambda \int_{a_{\rm P}}^{b_{\rm P}} 1 - \frac{1}{1+\beta^\prime ||\mathbf{w}_{\rm P}||^{-\alpha_{\rm N}}}\,{\rm d}v_{\rm P}\bigg){\rm d}\theta\, {\rm d}r\Bigg).
    \label{eq:pd_p}
\end{align}
\begin{proof}
The corollary follows by applying Laplace functional~\cite{dhillon2020poisson} of palm \ac{PLCP} to~\eqref{eq:thm_3}.
\end{proof}
\end{corollary}
\begin{corollary}
For the \ac{BLCP}-modeled network $(k = {\rm B})$, with $m_{\rm L} = 2$ and $m_{\rm N} = 1$, the detection success probability $p_{{\rm D},{\rm B}}(r_0,\beta)$ is given as,
\begin{align}
    &p_{{\rm D},{\rm B}}(r_0,\beta) = \nonumber\\
    &\hspace*{1cm} \exp\left(-\lambda \int_{R}^{R_{\rm B}} 1 - \left(\frac{1}{1+\frac{\beta^\prime}{2} {v}_{\rm B}^{-\alpha_{\rm L}}}\right)^2{\rm d}v_{\rm B}\right) \times \Bigg[\frac{1}{2\pi R_{\rm g}} \nonumber\\
    &\int_{0}^{R_{\rm g}}\!\!\!\! \int_0^{2\pi}\!\!\! \exp{\left(-\lambda \int_{a_{\rm B}}^{b_{\rm B}} 1 - \frac{1}{1+\beta^\prime ||\mathbf{w}_{\rm B}||^{-\alpha_{\rm N}}} \,{\rm d}v_{\rm B}\right)}{\rm d}\theta\, {\rm d}r\Bigg]^{n_{\rm B}-1}.
    \label{eq:pd_b}
\end{align}
\end{corollary}
Thus, the detection success probability of BLCP $p_{{\rm D},{\rm B}}(r_0,\beta)$ is a function of $r_0$ as $\mathbf{w}_{\rm B}$ and limits of integration depend on ego radar's location, unlike PLCP. 
\begin{remark}
While our current analysis assumes a fixed $R$, the framework can be generalized by first conditioning on $R$ and then by including a distribution of $R$ (e.g., uniformly or exponentially distributed). Equivalently $R$ can a function of the point process, such as the distance to the nearest vehicle. Both cases can be accommodated via appropriate outer integrals applied to the conditioned results. This flexibility enables generalization to complex scenarios, while preserving the analytical structure of the model.
\end{remark}

\begin{figure*}[t]
\centering
\subfloat[]
{\includegraphics[width=0.24\textwidth]{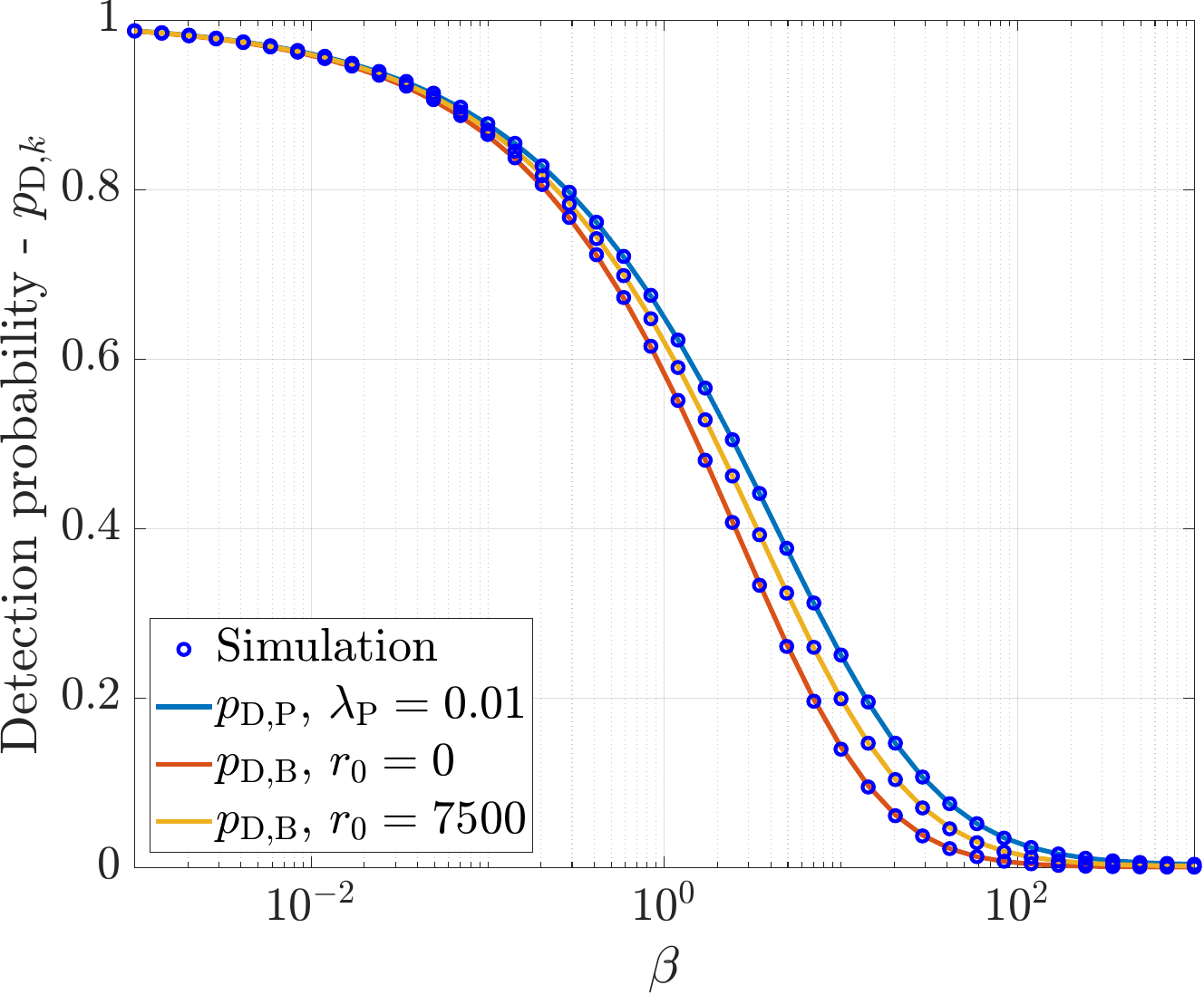}
\label{fig:result_mc}}
\hfil
\subfloat[]
{\includegraphics[width=0.24\textwidth]{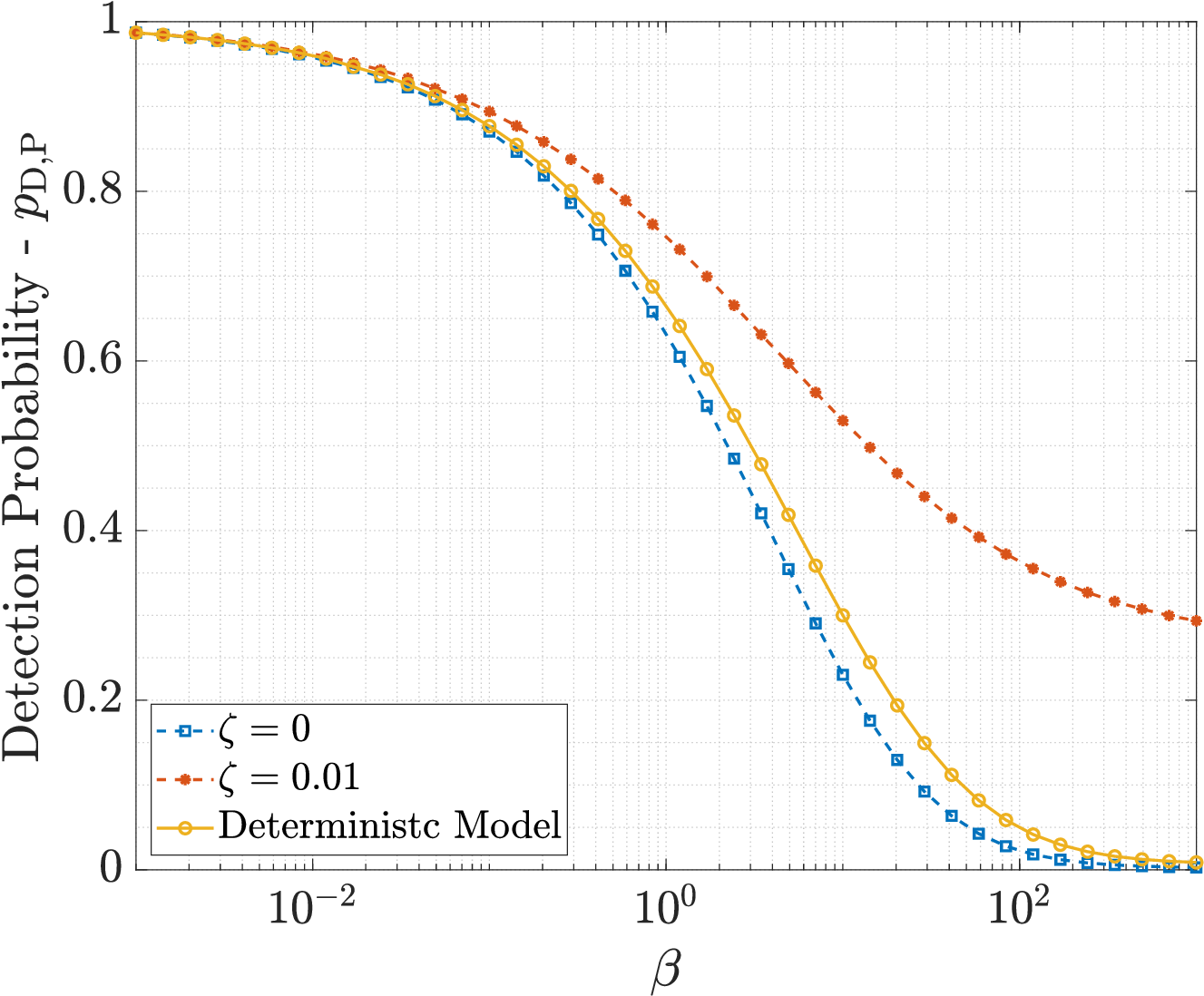}
\label{fig:prob_LOS_PLCP}}
\hfil
\subfloat[]
{\includegraphics[width=0.24\textwidth]{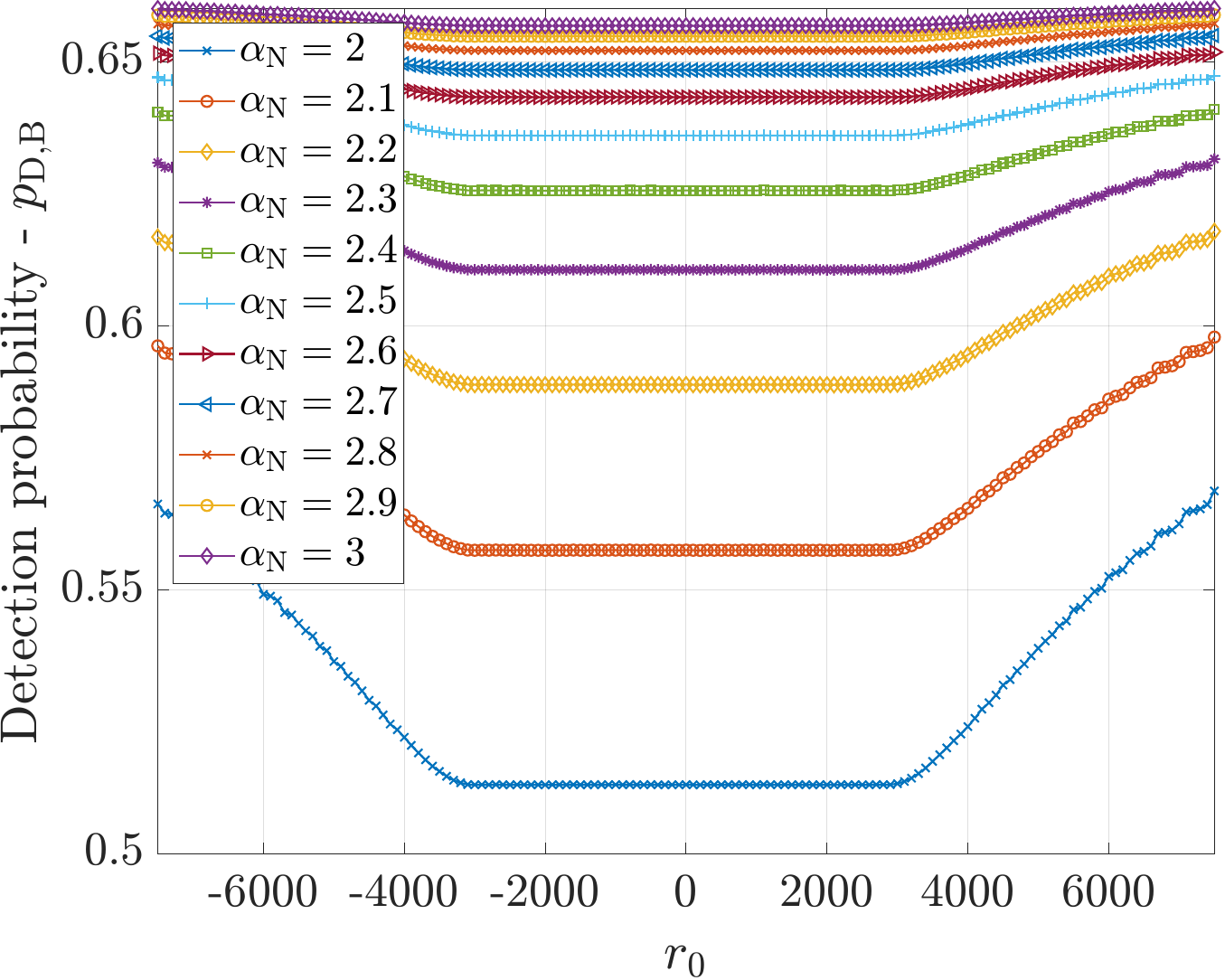}
\label{fig:result_n2}}
\caption{(a) Comparison of Monte Carlo simulations and analytical results for both \ac{PLCP} and \ac{BLCP} models w.r.t $\beta$, Probability of successful detection with respect to (b) $\beta$ for varying $\zeta$ in \ac{PLCP} model, (c) $r_0$ for varying $\alpha_{\rm N}$}
\label{fig:result_n} 
\end{figure*}

\vspace{-0.2cm}
\begin{remark}
For high path-loss exponents, interference is dominated by local automotive radar interferers, thus the detection statistics under BLCP converge to those under PLCP as $R_g \to \infty$, as per the Poisson limit. The BLCP's distinctive edge-effect capture is most valuable when analyzing large-scale, city-spanning spatial trends where finite-domain constraints shape the interference landscape.
\end{remark}
\vspace{-0.2cm}
\begin{remark} For a given threshold $\tau$, one can find the classical metrics of false-alarm probability $P_{\rm fa}$ and the detection probability $P_{\rm dt}$ by employing the Gil-Pelaez inversion theorem along with the Laplace transform of interference to find
\begin{align*}
    P_{{\rm fa},k} = \mathbb{P}\left(N + \mathbb{I}_k > \tau\right) = \frac{1}{2} + \frac{1}{\pi} \int_0^\infty \frac{\Im\left[ \mathcal{L}_{\mathbb{I}_k}(j t) e^{-j t \left(\tau - N\right)} \right]}{t} {\rm d}t
\end{align*}
where $\mathcal{L}_{\mathbb{I}_k}(j t) = \mathbb{E}[e^{-t \mathbb{I}_k}] $ is the Laplace transform of the $\mathbb{I}_k$ and can be found by following the steps in Theorem 1. Likewise $P_{\rm dt}$ can be found by solving
\begin{align*}
    &P_{{\rm dt},k} = \mathbb{P}\left(S + N + \mathbb{I}_k > \tau\right) \nonumber\\
    &= \mathbb{E}_{\Phi_{k_0}} \left[ \exp\left(-\frac{\tau - N - \sum_{\mathbf{w}_k \in \Phi_{k_0}}4\pi \gamma P h_{\mathbf{w}_k, l} ||\mathbf{w}_k||^{-\alpha_{l}}}{\gamma P R^{-2\alpha_L} \bar{\sigma}} \right) \right],
\end{align*}
due to the added randomness of the signal. We refer the reader to the recent work \cite{ghatak2025distribution} for a discussion on this.
\end{remark}

\subsection{Impact of Side-Lobe Interference}
The interfering set defined in Definition 2 characterizes radars that satisfy mutual sector alignment and therefore contribute dominant main-lobe interference. Practical automotive radar systems may experience sidelobe leakage from radars within the maximum radar range $R_k$ but not mutually visible to each other. To incorporate this effect, decompose the total interference as $\mathbf{I}_{{\rm tot},k} = \mathbf{I}_{{\rm ML}, k} + \mathbf{I}_{{\rm SL}, k}$, where $\mathbf{I}_{{\rm ML}, k}$ represents aggregate interference from the mutually visible interfering set, and $\mathbf{I}_{{\rm SL},k}$ captures sidelobe-induced interference. As sidelobe gains are typically several tens of decibels lower than main-lobe gains, $\mathbf{I}_{{\rm SL},k}$ can be modeled as a scaled interference term with gain factor $\eta \ll 1$. This model modifies the SIR expression by a multiplicative attenuation factor without changing the geometry-dependent distance bounds or detection probability analysis structure. Therefore, sidelobe interference only affects absolute performance levels, while spatial trends and comparative insights for PLCP and BLCP models remain unchanged.

\section{Numerical Results and Discussion}
\label{sec:Results}
The radar parameters used for the numerical results are taken from~\cite{series2014systems}. Specifically, $P = 10$ dBm, $\bar{\sigma} = 30$ dBsm, $G_{\rm t} = G_{\rm r} = 10$ dBi, and $\beta = 10$ dB. We study how the success probability $p_{{\rm D},k}$ varies with system parameters like $R$, $\Omega$, and $\lambda$. In all the plots, we assume $\Omega = 7.5^\circ$ (for \ac{PLCP} models), $\Omega = 15^\circ$ (for \ac{BLCP} models), $R=15\, {\rm m}$, $\lambda = 0.01\, {\rm m}^{-1}$, $R_{\rm g} = 1500\, {\rm m}$, $n_{\rm B} = 300$, and $R_k = 500\, {\rm m}$ unless specified otherwise. 

Fig.~\ref{fig:result_mc} presents a comparison between the analytically derived detection probability $p_{{\rm D},k}$ and the corresponding Monte Carlo simulation results for both PLCP and BLCP models. For the BLCP case where $n_{\rm B} = 300$, and $R_{\rm g} = 1500$, the $p_{{\rm D},{\rm B}}$  is shown for two representative ego radar locations: at the city center and near the outskirts. For PLCP a single plots is shown for $\lambda_{\rm L} = 0.01$. The plots validate the analytical model's accuracy across spatial regimes. In these plots for both BLCP and PLCP, the value $\lambda = 0.01$ and $\Omega = 30^\circ$.

\subsection{Impact of Path-Loss Exponent }
\subsubsection{Impact of Probabilistic LOS Model}
To validate the deterministic LOS/NLOS assumption utilized in our analytical framework having $\alpha_{\rm L} = 2$, $\alpha_{\rm N} = 4$, we compared the detection probability $p_{{\rm D},{\rm P}}$ against a probabilistic blockage model via Monte Carlo simulations in Fig.~\ref{fig:prob_LOS_PLCP}. In the probabilistic model, the LOS probability of any link of length $r$ is given by $p_{\mathrm{LOS}}(r) = e^{-\zeta r}$, where $\zeta$ is the blockage parameter determined by the building density. Interferers are randomly assigned LOS ($\alpha_{\rm L}=2$) or NLOS ($\alpha_{\rm N}=4$) states based on this probability, regardless of whether they are on $L_0$ or $L_i$. Our simulations reveal that the deterministic model serves as a tight lower bound to the probabilistic model in dense urban scenarios. Specifically, at $\zeta = 0$, the model reduces to an all LOS scenario ($p_{\rm L}(r)\equiv 1$), yielding the largest aggregate interference and therefore the lowest detection probability, this represents a worst-case, no-blockage environment. The maximum deviation in detection probability between the deterministic and probabilistic models was observed to be less than $7\%$. As $\zeta > 0$, the deterministic scenario lies between the extreme case of $\zeta=0$  and $\zeta>0$, depending on the radar range and city geometry; the plot illustrates this ordering and the relative conservatism of the deterministic approximation.

\subsubsection{Impact of varying $\alpha_{\rm N}$}
Fig.~\ref{fig:result_n2} shows the detection probability $p_{{\rm D},{\rm B}}$ as a function of the ego-radar location $r_0$ for different values of the NLOS path-loss exponent $\alpha_{\rm N}$. As $\alpha_{\rm N}$ increases, the overall detection probability improves due to the stronger attenuation of cross-street interference. Moreover, for small $\alpha_{\rm N}$, particularly near $\alpha_{\rm N} = 2$, $p_{{\rm D},{\rm B}}$ exhibits pronounced spatial variations with $r_0$ (which would be explained in further section~V-B), whereas for larger $\alpha_{\rm N}$ the curves become nearly flat, indicating strong spatial averaging of the NLOS interference field. This behavior follows from the BLCP geometry. For small $\alpha_{\rm N}$, interference decays slowly with distance and is dominated by the nearest cross-street vehicles, making the detection probability highly sensitive to the ego location. As $\alpha_{\rm N}$ increases, the contribution of distant interferers is strongly suppressed and the aggregate NLOS interference becomes isotropically averaged, resulting in spatially uniform performance. In the remainder of the paper, we fix $\alpha_{\rm N}=2$ to intentionally study a \textbf{pessimistic, interference-limited regime}, since this choice maximizes mutual interference and produces the strongest spatial sensitivity. Larger values of $\alpha_{\rm N}$ only improve detection probability and further smooth the spatial variations.

\subsection{Success Probability}
\subsubsection{PLCP} Fig.~\ref{fig:result_P_1} shows that the probability of successful detection decreases as $R$ increases due to increased path loss irrespective of the beamwidth, $\lambda_{\rm L}$, and $\lambda$. Furthermore, the $p_{{\rm D},{\rm P}}$ is higher when the interference from neighboring radars is lower due to a low vehicle density. Next, we examine the variation of success probability with $\lambda$ in Fig~\ref{fig:result_P_2}. The results indicate that $p_{{\rm D},{\rm P}}$ decreases with $\lambda$ due to an increase in the total number of interferers and a decrease of the distance between the nearest interfering radar and the ego radar at the origin. For a fixed beamwidth, $p_{{\rm D},{\rm P}}$ decreases as $\lambda_{\rm L}$ increases. 
Similar trends is observed between $p_{{\rm D},{\rm P}}$ and radar beamwidth in Fig.~\ref{fig:result_P_3}. We note that intensity has a greater impact on total interference than beamwidth.

Fig.~\ref{fig:result_P_3} shows that as the beamwidth increases the number of interfering radars within the main lobe of the ego radar increases, thereby decreasing the detection probability. For lower $\Omega$ (e.g., $< 5^\circ$) and given $\lambda$, the impact of $\lambda_{\rm L}$ is limited, since, in case of small beamwidths, the number of interferes present on the remaining lines of \ac{PLP} are negligible as compared to interferes present on ego radars street, i.e., $L_0$.

\begin{figure*}[t]
\centering
\subfloat[]
{\includegraphics[width=0.24\textwidth]{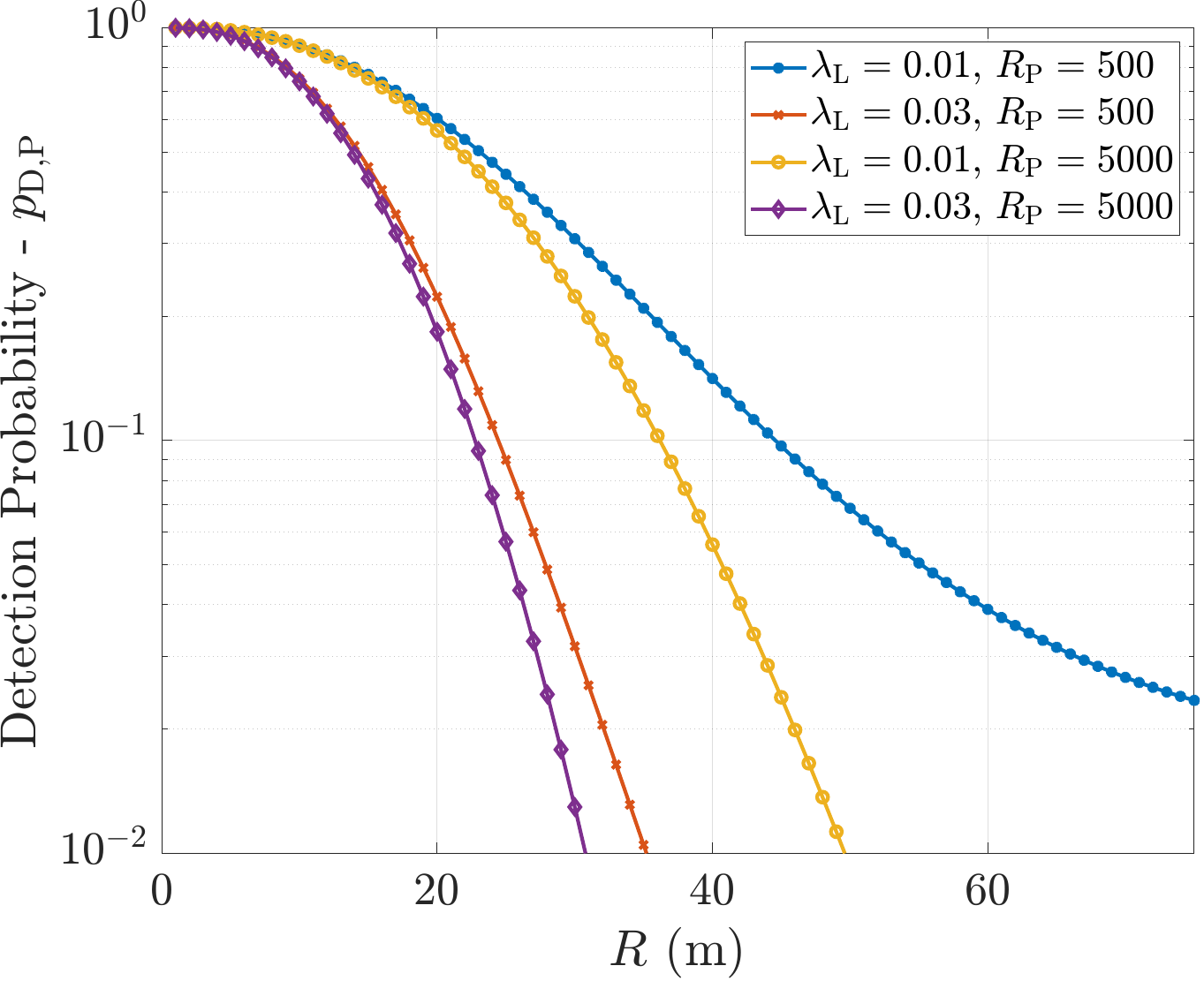}
\label{fig:result_P_1}}
\hfil
\subfloat[]
{\includegraphics[width=0.24\textwidth]{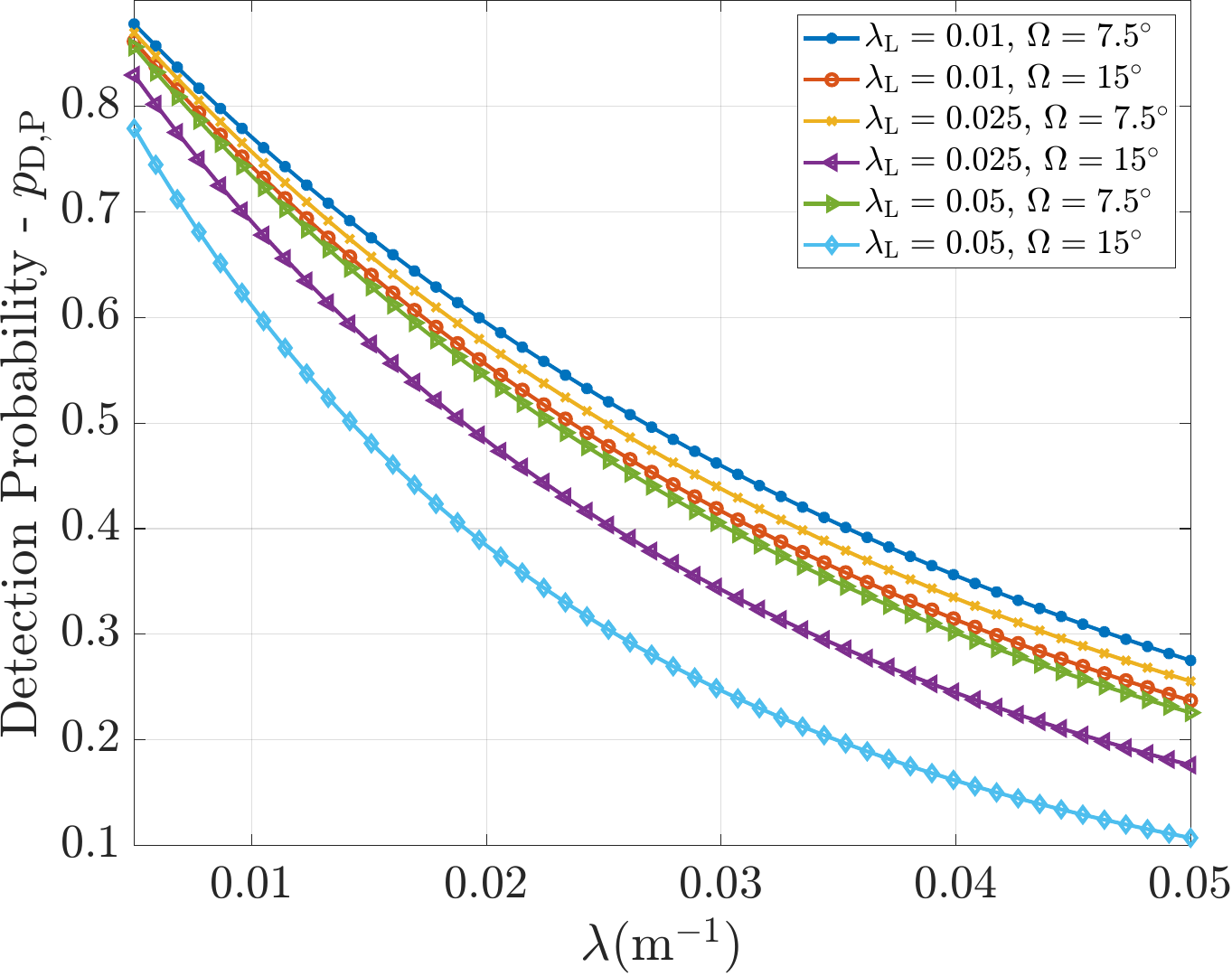}
\label{fig:result_P_2}}
\hfil
\subfloat[]
{\includegraphics[width=0.24\textwidth]{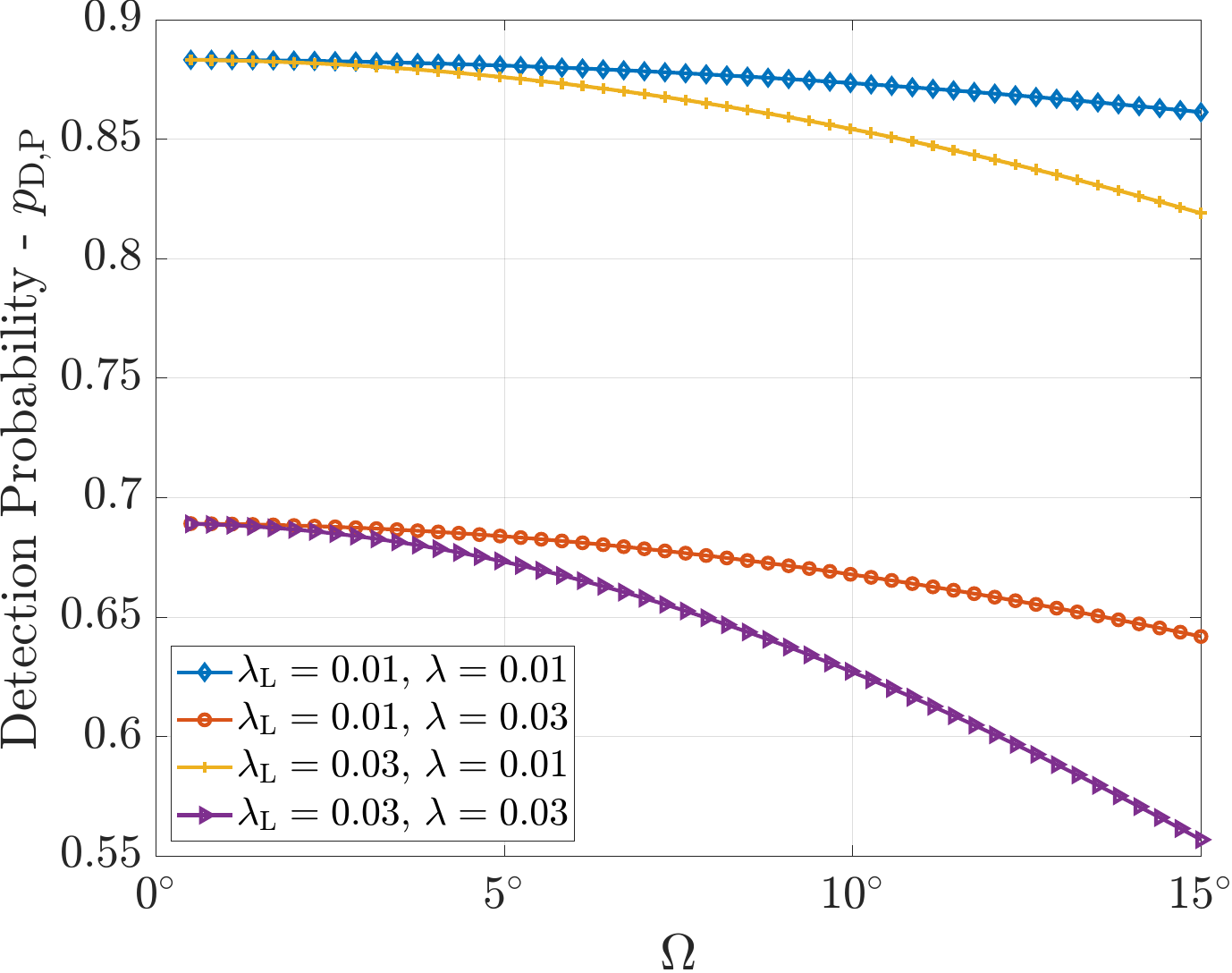}
\label{fig:result_P_3}}
\caption{Probability of successful detection $p_{{\rm D},{\rm P}}$ with respect to (b) $R$, (c) $\lambda$, and (d) $\Omega$.}
\label{fig:result_P_1_4} 
\end{figure*}

\begin{figure*}[t]
\centering
\subfloat[]
{\includegraphics[width=0.25\textwidth]{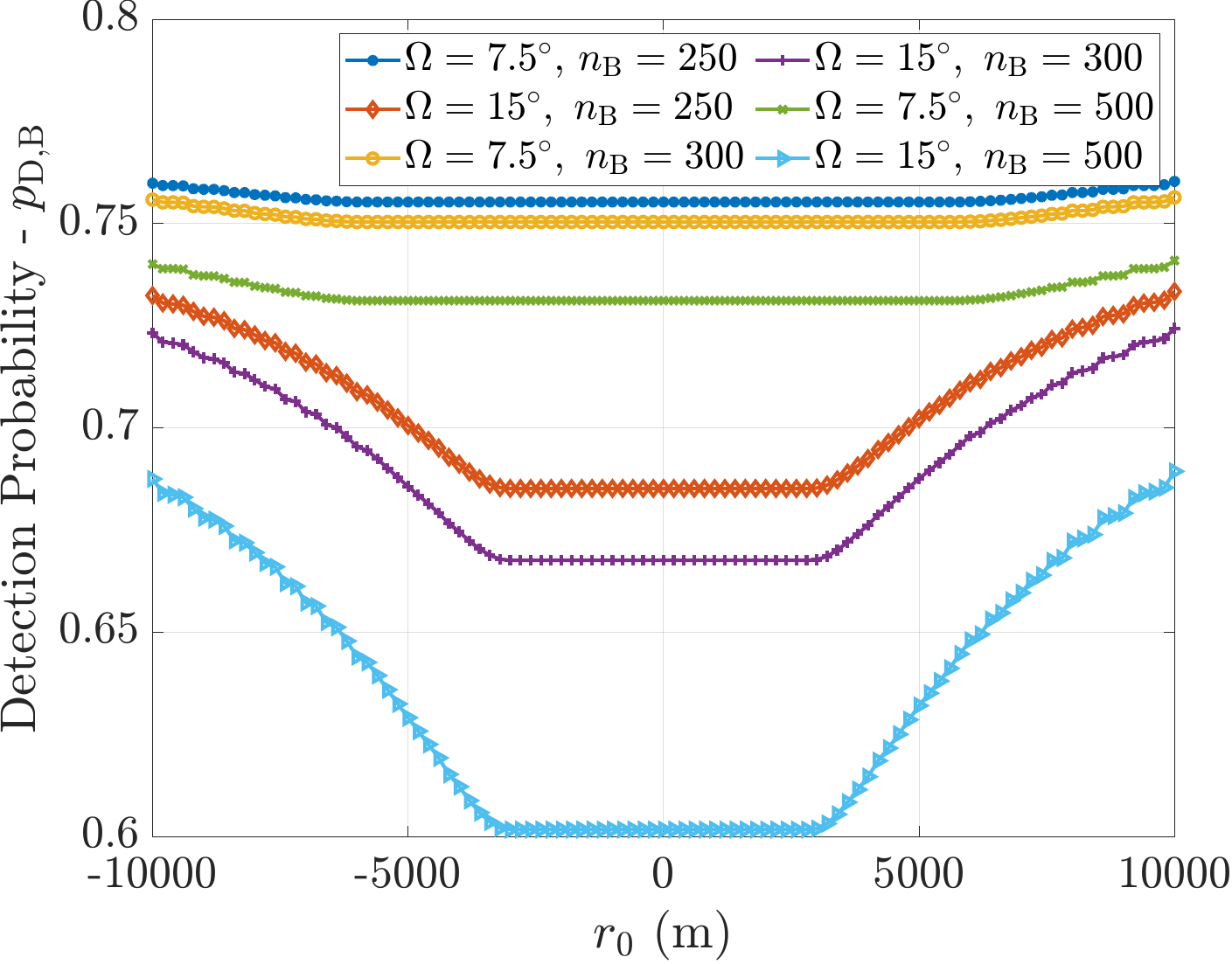}
\label{fig:result_B_1}}
\hfil
\subfloat[]
{\includegraphics[width=0.24\textwidth]{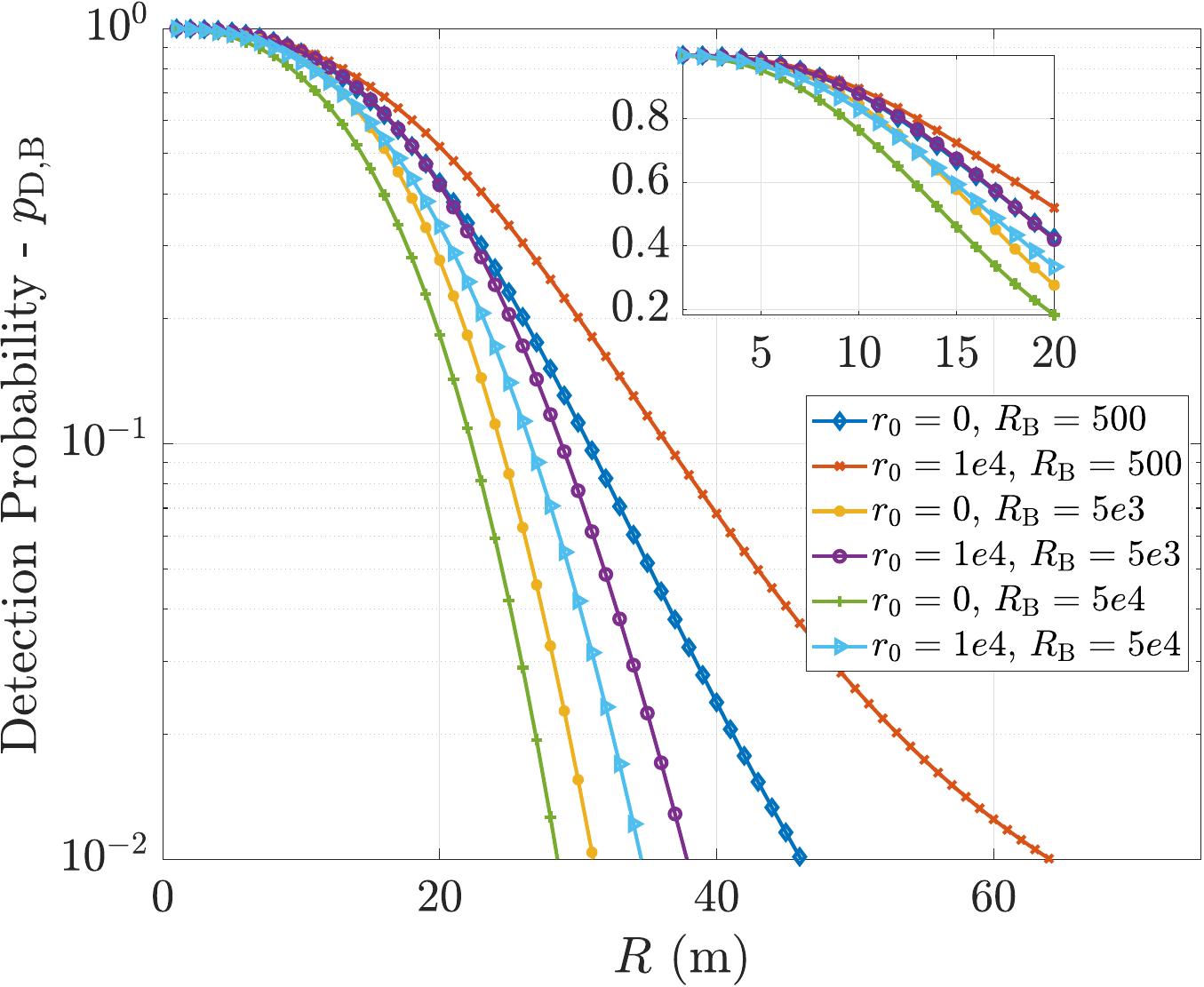}
\label{fig:result_B_2}}
\hfil
\subfloat[]
{\includegraphics[width=0.24\textwidth]{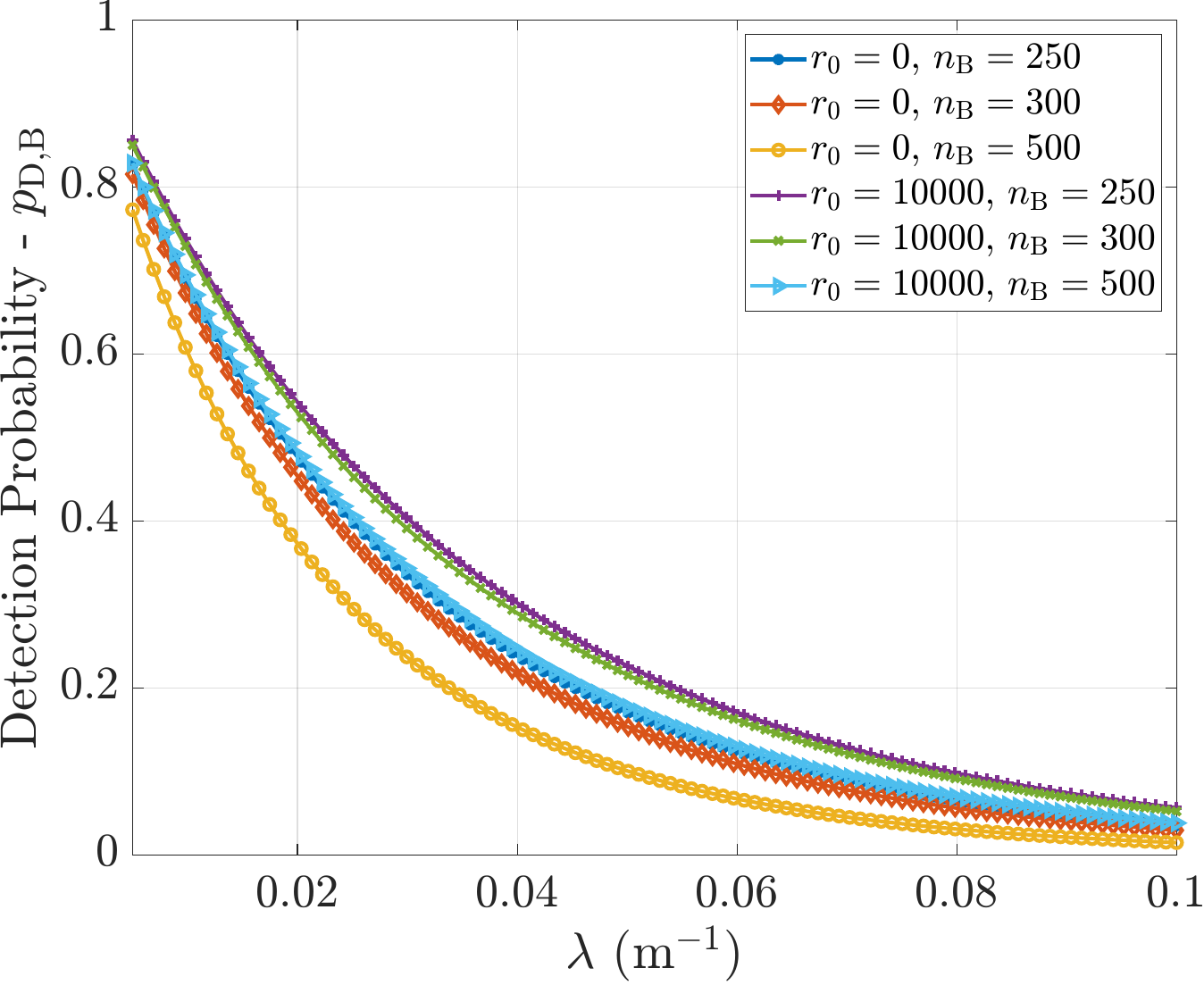}
\label{fig:result_B_3}}
\hfil
\subfloat[]
{\includegraphics[width=0.24\textwidth]{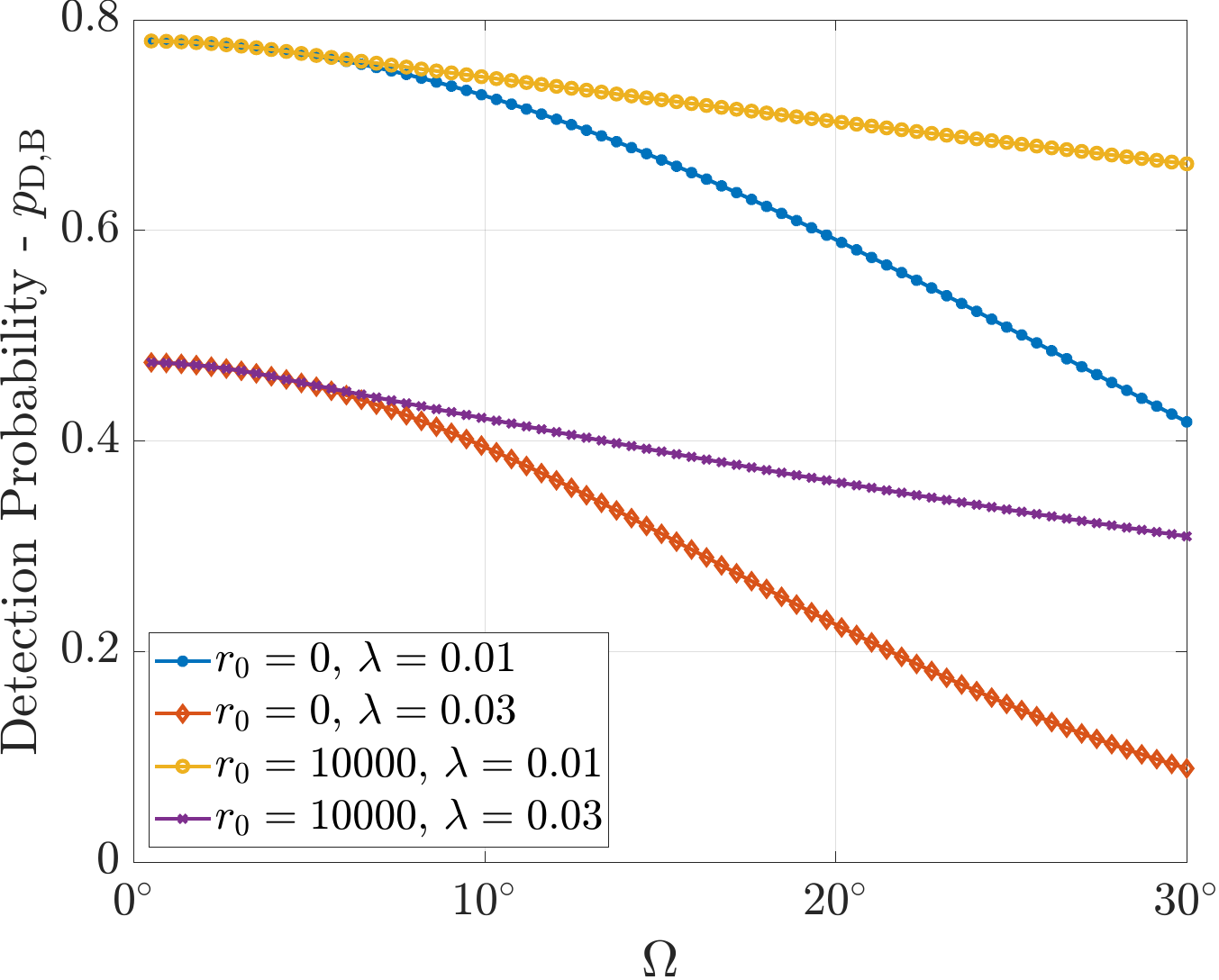}
\label{fig:result_B_4}}
\caption{Probability of successful detection $p_{{\rm D},{\rm B}}$ with respect to (a) $r_0$, (b) $R$, (c) $\lambda$, and (d) $\Omega$.}
\label{fig:result_B_1_4} 
\end{figure*}

\subsubsection{\ac{BLCP}} Unlike the detection performance of ego radar in the \ac{PLCP} framework, the \ac{BLCP} scenario presents a contrasting view of detection probability $p_{{\rm D},{\rm B}}$ against different parameters. Fig.~\ref{fig:result_B_1} shows that as the ego radar moves from the outskirts of a city ($r_0=- 10000\, {\rm m}$) to its center ($r_0=0\, {\rm m}$), then back to the outskirts ($r_0 = +10000\, {\rm m}$), the detection probability first decreases, then saturates to a certain level and subsequently increases as the vehicle leaves the city. Recall that the axis vector of the ego radar is $\mathbf{a} = (0,1)$. Although the street geometry is isotropic, the radar sector is not. Consequently, the detection performance at $-10000\, {\rm m}$ is not the same as that at $10000\, {\rm m}$. The ego radar at $-10000\, {\rm m}$ has lower $p_{{\rm D},{\rm B}}$ as compared to $10000\, {\rm m}$. This is because, at $r_0 = -10000\, {\rm m}$, the radar sector of ego radar has a higher number of interfering radars falling inside it than at $r_0 = 10000\, {\rm m}$. 

The spatial variation in $p_{{\rm D},{\rm B}}$ is fundamentally driven by the radial decay of street density in the BLCP model. At the city center, the ego radar's sector intersects the maximum number of streets, each contributing a non-negligible interfering segment (as quantified in section~\ref{sub:sub_A}). This results in the highest aggregate interference and lowest detection probability. As the ego radar moves toward the outskirts, the number of intersecting streets and the total interfering length declines due to the finite support of the BLCP, thereby improving $p_{{\rm D},{\rm B}}$. This behavior, validated against real-world road data (Table~\ref{MainTable}), reflects the practical reality that radar performance is inherently location-dependent in urban environments.

We further observe that for smaller values of beamwidth, we have a larger saturation region as compared to larger values of $\Omega$. We have assumed $R_{\rm g} = 1500\, {\rm m}$ and $R_{\rm B} = 500\, {\rm m}$. Thus the detection probability for the ego radar in the \ac{BLCP} model \textbf{should have been constant} from $r_0 = -1500\, {\rm m}$ to $r_0 = 500\, {\rm m}$ m irrespective of of $\Omega$. However, this is not the case. From Theorem 2 of~\cite{shah2024binomial}, we deduce that the  \textit{line length density} is constant in the following range of values of $r_0 = [-1500, 500]\, {\rm m}$, as the radar sector is completely within the generating circle $\mathcal{C}((0,0), R_{\rm g})$. This unusual behavior of the network performance where the saturation range of $p_{{\rm D},{\rm B}}$ is not equal to $r_0 = [-1500, 500]\, {\rm m}$ will be explained in the subsection~\ref{sub:sub_A}.

Fig.~\ref{fig:result_B_2} shows that as $R$ increases, the detection performance deteriorates due to the two-way path loss, while the interference power remains the same. Likewise, by increasing $R_{\rm B}$ from $500\, {\rm m}$ to $50000\, {\rm m}$, $p_{{\rm D},{\rm B}}$  decreases, as the total number of interferers within the radar sector increases. Now, in Fig.~\ref{fig:result_B_3}, we plot the detection probability with respect to the intensity of vehicles, $\lambda$, for two different locations of the ego radar. As evident from the plot, $p_{{\rm D},{\rm B}}$ decreases as $\lambda$ increases due to an increase in interference. Additionally, we observe that, for $r_0 = \pm 10000\, {\rm m}$ (outskirts), $p_{{\rm D},{\rm B}}$ is at a larger value than for $r_0 = 0$ (city center). This is because the line length/street density at the outskirts is lower than at the city center. Also, as $n_{\rm B}$ increases from $250$ to $500$, we observe that, for a constant value of $\lambda$, the detection probability decreases because a higher number of streets result in a larger number of vehicles.
Finally Fig.~\ref{fig:result_B_4} shows that as $\Omega$ increases from $1^\circ$ to $15^\circ$, the detection probability decreases non-linearly. We observe that at the city center, the impact of $\Omega$ on the detection performance is more significant as compared to $r_0 = 10000$. Furthermore, for a given value of $\lambda$, the detection probability for initial values of $\Omega$ is nearly the same for two values of $r_0 = \{0, 10000\}$. Indeed, similar to the PLCP, for low beamwidth values, the street containing the ego radar contributes more significantly to the interference as compared to the other streets.
In subsection~\ref{sub:sub_B}, we compare the performance of \ac{BLCP} and \ac{PLCP} frameworks to better understand the contrasting results of these two frameworks.

\begin{figure}[h]
\centering
\subfloat[]
{\includegraphics[width=0.23\textwidth]{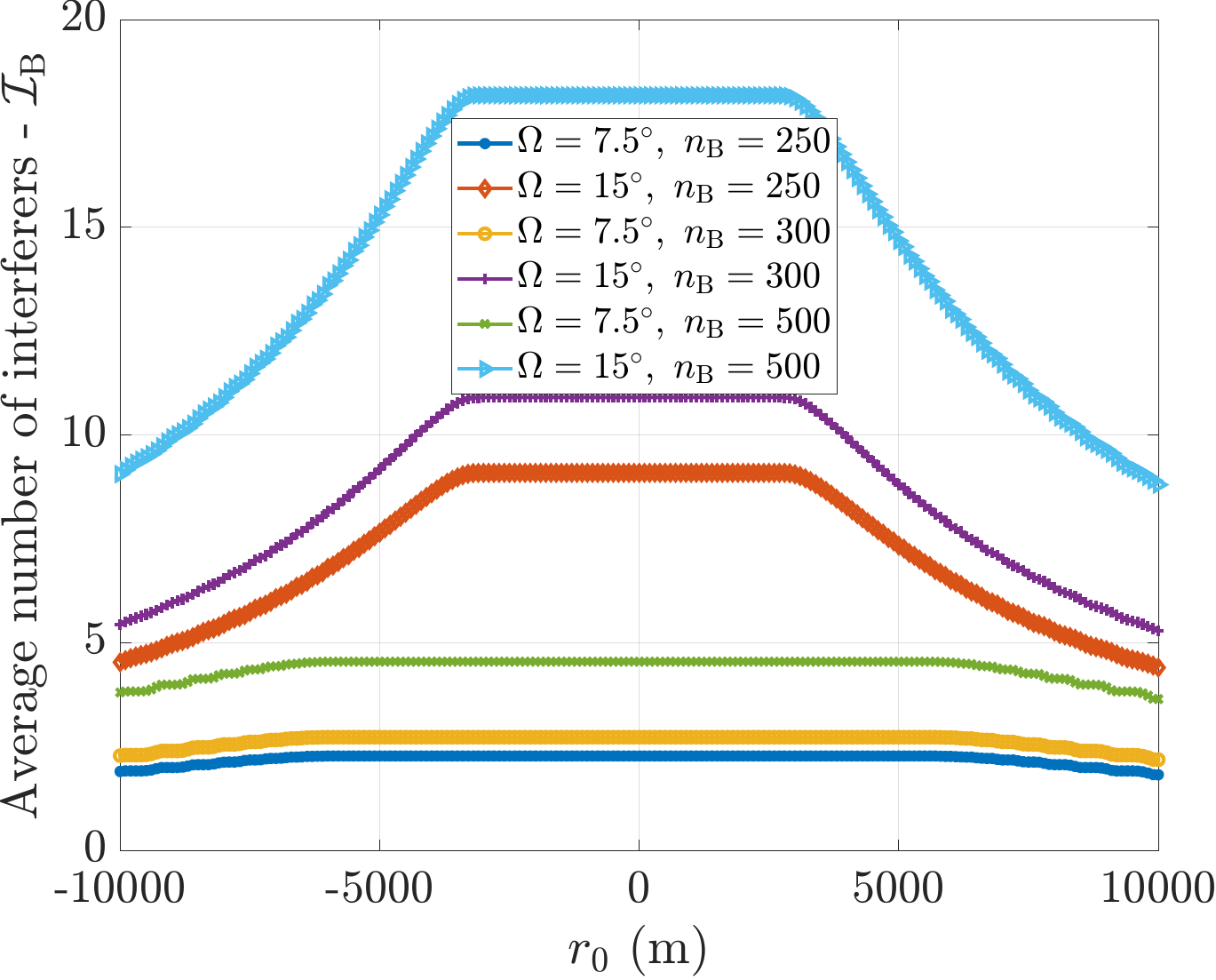}
\label{fig:result_new1}}
\hfil
\subfloat[]
{\includegraphics[width=0.23\textwidth]{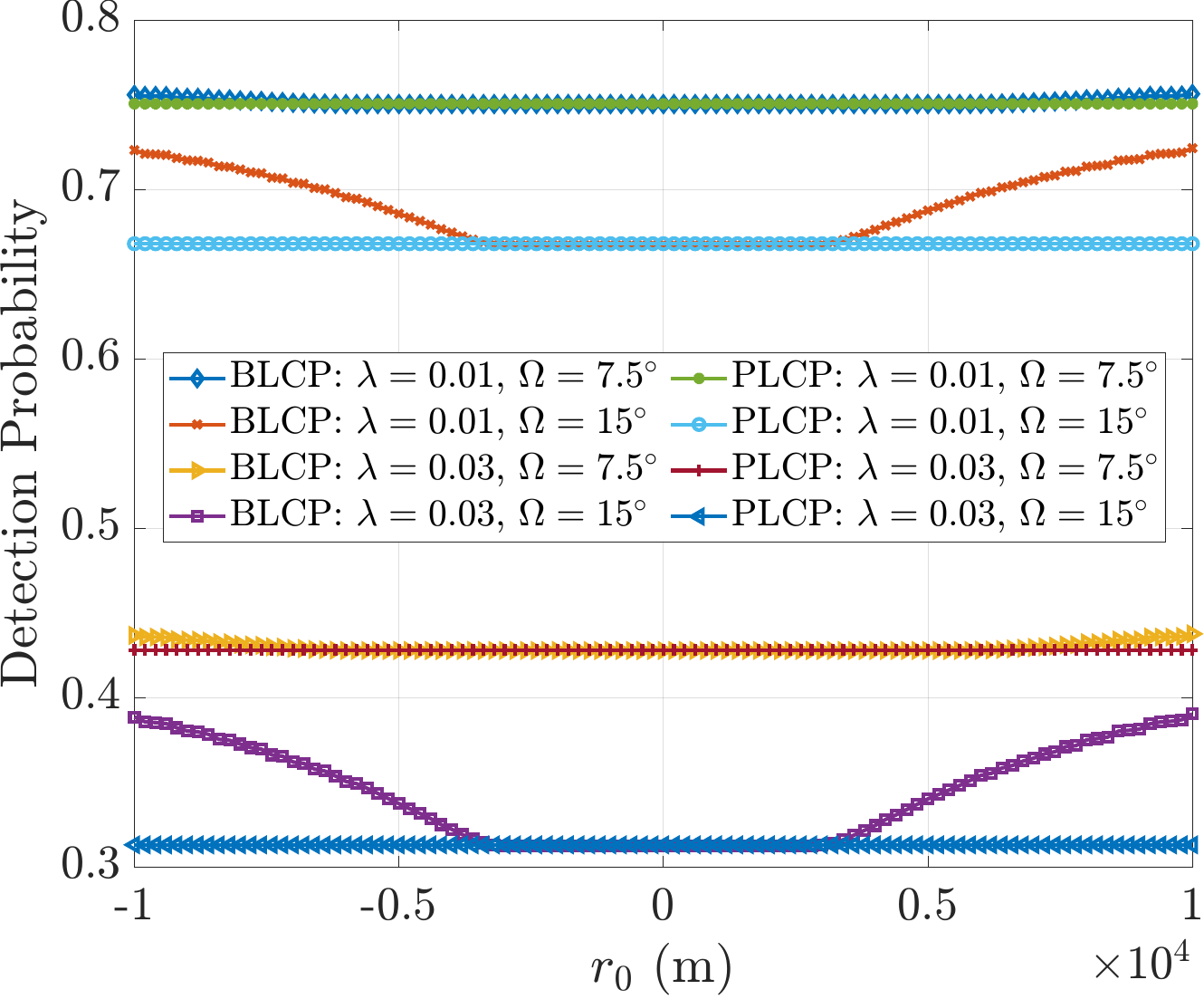}
\label{fig:result_new2}}
\caption{(a) Average number of interferes falling inside the bounded radar sector w.r.t $r_0$. (b) Comparison of $p_{{\rm D},k}$ of ego radar for \ac{PLCP} and \ac{BLCP} models w.r.t $r_0$.}
\label{fig:result_new} 
\end{figure}

\subsection{Average number of interferers}
\label{sub:sub_A}
In this subsection, we explain the phenomenon when the detection probability remains constant for a specific range of $r_0$. Fig.~\ref{fig:result_new1} plots $\mathcal{I}_{\rm B}$ w.r.t. the location of ego radar for the same values of $\Omega$ and $n_{\rm B}$ as in Fig.~\ref{fig:result_B_1}. We see that $\mathcal{I}_{\rm B}$ first increases as $r_0$ increases from $-10000\, {\rm m}$, then saturates at a value and decreases again. The region where it starts to saturate and then again decreases corresponds to the same range of values, for which $p_{{\rm D},{\rm B}}$ remains constant in Fig.~\ref{fig:result_B_1}. Accordingly, the saturation range is not equal to $[-1500, 500]\, {\rm m}$. Rather, the saturation range is a function of $\Omega$, $R_{\rm B}$ and $R_{\rm g}$, and not the function of $n_{\rm B}$ and $\lambda$. In Fig.~\ref{fig:result_new1} and Fig.~\ref{fig:result_B_1}, the region of saturation remains the same i.e. $r_0 = -3100\, {\rm m}$ to $r_0 = 2800\, {\rm m}$ for $\Omega = 7.5^\circ$ and $n_{\rm B} = \{250, 300, 500\}$. Likewise, the region of saturation is from $r_0 = -6000\, {\rm m}$ to $r_0 = 5700\, {\rm m}$ in both Fig.~\ref{fig:result_new1} and Fig.~\ref{fig:result_B_1} for $\Omega = 7.5^\circ$ and $n_{\rm B} = \{250, 300, 500\}$. \\
{\bf Example:} Let the ego radar be present at the origin with axis vector $\mathbf{a} = (0,1)$, and let $\Omega < 1^\circ$, i.e., negligible. Here, the detection probability performance does not change as the ego radar moves anywhere in the city, and it depends only on the intensity of vehicles present on line $L_0$ since the beamwidth is too narrow to be impacted by cars on other lanes. Thus, we expect a constant $p_{{\rm D},{\rm B}}$. As the beamwidth increases, the ego radar incorporates more interfering radars. For $\Omega = 2\pi$ corresponding to an isotropic radar sector, we see the saturation region equal $[-1500, 500]\, {\rm m}$.

\subsection{Detection performance of \ac{BLCP} v/s \ac{PLCP}}
\label{sub:sub_B}
Fig.~\ref{fig:result_new2} compares the detection probability $p_{{\rm D},k}$ for both the \ac{PLCP} and \ac{BLCP} based model networks. To make a fair comparison, we take the intensity of lines of \ac{PLP} as $\lambda_{\rm L} = \frac{300}{2\pi 1500}$ and number of lines in \ac{BLP} $n_{\rm B} = 300$. The detection probability presents a broad contrast between the \ac{PLCP} and \ac{BLCP} networks. In the \ac{PLCP} model, there are no insights on the detection probability depending on the location of ego radar. Due to the homogeneous nature of the \ac{PLCP} the $p_{{\rm D},{\rm P}}$ remains constant at all value of $r_0$. On the other hand, \ac{BLCP} highlights the impact of the ego radar location. Interestingly, we observe that the detection probability for the \ac{PLCP} model are the same as for the \ac{BLCP} in the saturation region, especially if the generating circle is large. The \ac{PLCP} can be used to model the network instance of a local section of the city with a homogeneous distribution of lanes, while the \ac{BLCP} can model the road structure of an entire city, and the two models can be used to study the performance of an ego radar on a micro or macro scale. The following section introduces a real-world perspective of modeling road networks using \ac{PLCP} and \ac{BLCP} framework, where we learn of how \ac{PLCP} and \ac{BLCP} can be used in modeling different perspectives of the city network.

\subsection{Sensitivity Analysis of Network Parameters}
For a deeper understanding of network geometry effects, we analyze the sensitivity of detection probability $p_{{\rm D},k}$, to small perturbations in the density parameters $\lambda, \lambda_{\rm L}$ and geometric parameters $R_k, n_{\rm B}, R_{\rm g}, \Omega$. Understanding these sensitivities is crucial for identifying the parameters that most significantly impact radar performance. The detection probability for the PLCP model, given in \eqref{eq:pd_p}, can be expressed in compact exponential form as
\begin{align}
    p_{{\rm D},{\rm P}} 
    = \exp \big(-\lambda \mathcal{A}_{\rm LOS} 
               - \lambda_{\rm L} \mathcal{A}_{\rm NLOS}(\lambda)\big),
    \label{eq:pd_compact}
\end{align}
where $\mathcal{A}_{\rm LOS}$ is the LOS interference length on the ego street and 
$\mathcal{A}_{\rm NLOS}(\lambda)$ is the aggregated NLOS interference contribution arising from all intersecting streets.

\subsubsection{Derivatives of parameters}
Differentiating \eqref{eq:pd_compact} with respect to $\lambda_{\rm L}$ yields $\frac{\partial p_{{\rm D},{\rm P}}}{\partial \lambda_{\rm L}}
    = -\mathcal{A}_{\rm NLOS}(\lambda)\, p_{{\rm D},{\rm P}}$,
which shows that the sensitivity to the road intensity is proportional to the NLOS interference footprint.

Differentiation with respect to the vehicular density $\lambda$ gives
\begin{align}
    \frac{\partial p_{{\rm D},{\rm P}}}{\partial \lambda}
    = -\left(\mathcal{A}_{\rm LOS}
    + \lambda_{\rm L} \frac{\partial\mathcal{A}_{\rm NLOS}}{\partial\lambda}\right)
    p_{{\rm D},{\rm P}},
    \label{eq:derivative_lambda}
\end{align}
where
\begin{align*}
    \frac{\partial\mathcal{A}_{\rm NLOS}}{\partial\lambda} 
    &= \int_{\mathbb{R}^{+}} \int_{0}^{2\pi} \left(\int_{a_{\rm P}}^{b_{\rm P}} 1 - \frac{1}{1+\beta^\prime ||\mathbf{w}_{\rm P}||^{-\alpha_{\rm N}}}\, {\rm d}v_{\rm P}\right)  \times \nonumber\\
    &\hspace*{-1.5cm} \left(1 - \exp \Big(- \lambda \int_{a_{\rm P}}^{b_{\rm P}} 1 - \frac{1}{1+\beta^\prime ||\mathbf{w}_{\rm P}||^{-\alpha_{\rm N}}}\, {\rm d}v_{\rm P}\Big)\right) {\rm d}\theta\,{\rm d}r .
\end{align*}
The first term in \eqref{eq:derivative_lambda} corresponds to the LOS vehicles on the ego street, while the second term quantifies the effect of vehicles on intersecting streets, weighted by the probability that these streets are occupied.

The situation for the geometric parameters is fundamentally different. The PLCP parameters $(R_{\rm P},\Omega)$ and the BLCP parameters $(R_{\rm B},\Omega, r_0,n_{\rm B},R_{\rm g})$ appear exclusively through the integration limits $a_k$ and $b_k$, which, as shown in Theorem 1 and Lemma 1, are defined by several piecewise cases involving trigonometric conditions, min-max relations, intersection geometry, and boundary changes. These limits involve trigonometric boundary transitions, case-dependent activation regions, and min-max operations, resulting in discontinuous or non-differentiable derivatives across region boundaries. Closed-form analytical derivatives with respect to these geometric parameters would therefore be extremely lengthy and offer little practical insight.

For this reason, we compute their sensitivities using numerical central differences.  For a given parameter $\eta = \left\{R_k, \Omega, r_0, n_{\rm B}, R_{\rm g} \right\}$, the partial derivative is approximated by $\frac{\partial p_{{\rm D},k}}{\partial \eta} \approx \frac{p_{{\rm D},k}(\eta + h) - p_{{\rm D},k}(\eta - h)}{2h},$
where $h$ is a small relative perturbation. To facilitate comparison across parameters with different physical units, we evaluate the normalized sensitivity $S_\eta = \frac{\eta}{p_{{\rm D},k}} \frac{\partial p_{{\rm D},k}}{\partial \eta}.$
This gives the percent change in $p_{{\rm D},k}$ per $1\%$ change in the parameter and is directly comparable across parameters with different units.

\subsubsection{Insights from numerical results}
We now present the numerical values of the partial derivatives and normalized sensitivities obtained using central-difference differentiation applied directly to the analytical detection probability expressions for both PLCP and BLCP, ensuring deterministic (noise-free) estimates. The baseline parameters were set to $\lambda = 0.01$, $\lambda_{\rm L} = 0.01$, $R_k = 500$, $\Omega = 15^\circ$, $n_{\rm B} = 500$, and $R_{\rm g} = 1500$, yielding $p_{{\rm D},{\rm P}} = 0.877$ for PLCP and $p_{{\rm D},{\rm B}} = 0.5255$ for BLCP at $r_0=1500$. 

The sensitivities for the PLCP model are summarized in Table~\ref{tab:plcp_sensitivity}. It shows that the vehicular density $\lambda$ overwhelmingly dominates all other parameters in the PLCP model. A $1\%$ increase in $\lambda$ results in an approximate $0.13\%$ reduction in $p_{{\rm D},{\rm P}}$. For a fractional change $h$ in $\lambda$, an absolute change in detection probablity is $\Delta p_{{\rm D},{\rm P}} \approx p_{{\rm D},{\rm P}} \times S_{\lambda} \times h$ i.e. for $h = 0.01$ and $S_\lambda = -0.13$, we get $\Delta p_{{\rm D},{\rm P}} = -0.00115$. Thus $p_{{\rm D},{\rm P}}$ drops from $0.877$ to about $0.875$, likewise a $10\%$ increase in $\lambda$ produces $\Delta p_{{\rm D},{\rm P}} = -0.0115$. The effects of road density $\lambda_{\rm L}$, radar range $R_k$, and beamwidth $\Omega$ are more than two orders of magnitude smaller.
\begin{figure}[t]
\centering
\subfloat[]
{\includegraphics[width=0.24\textwidth]{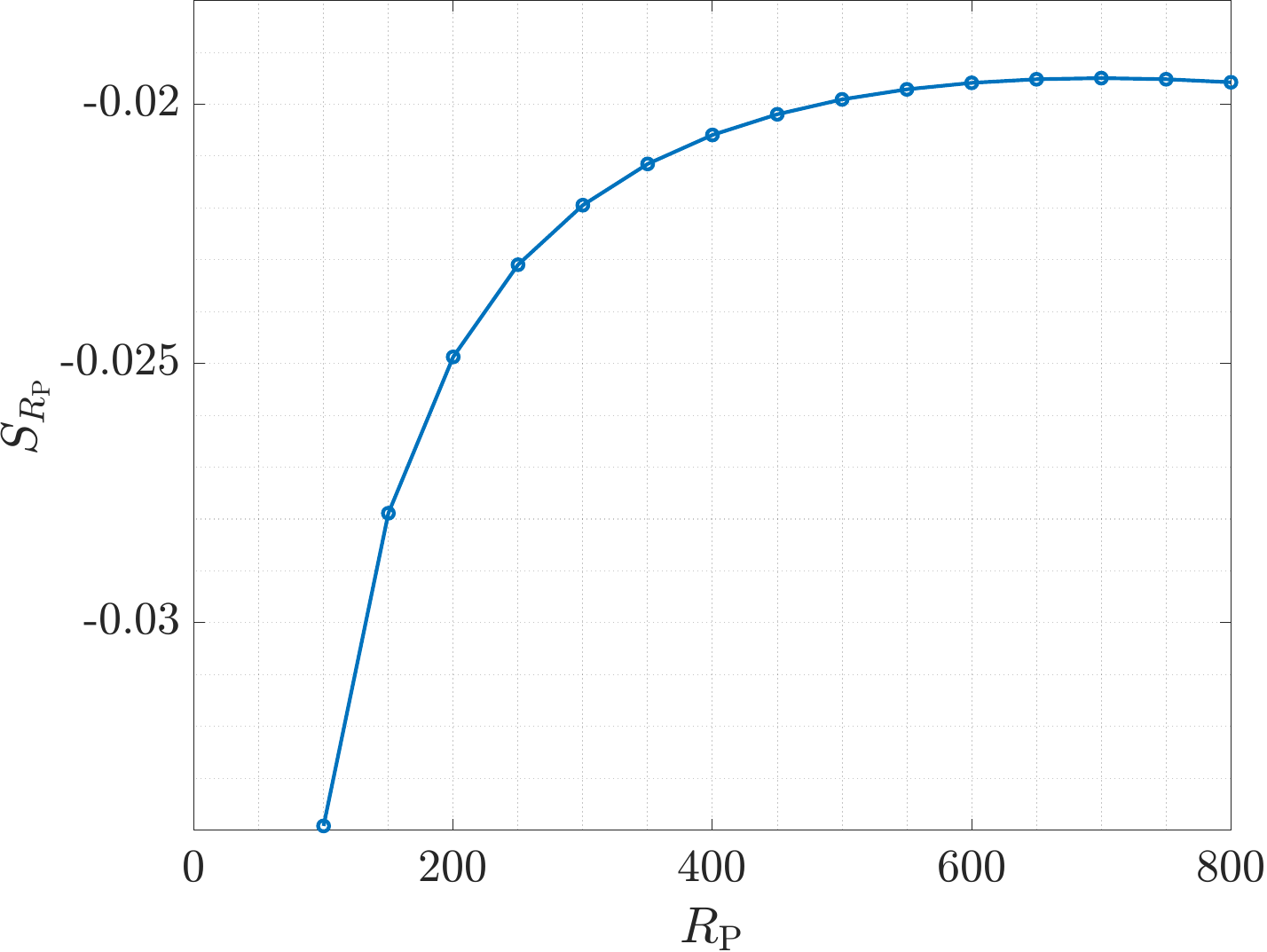}
\label{fig:r_1}}
\hfil
\subfloat[]
{\includegraphics[width=0.23\textwidth]{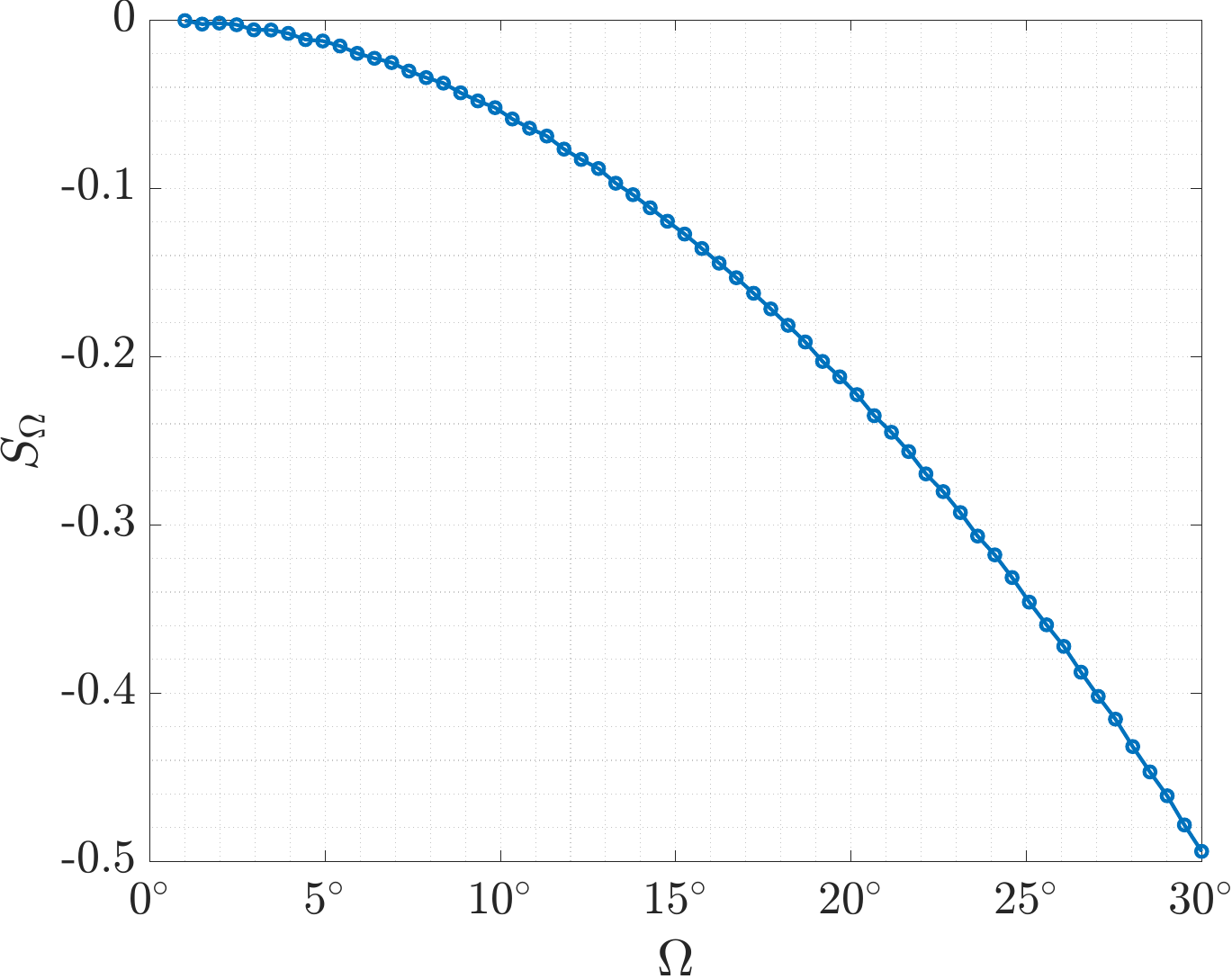}
\label{fig:r_2}}
\caption{Normalized Sensitivity v/s (a) $R_{\rm P}$, and (b) beamwidth for the PLCP model.}
\label{fig:r_1_2} 
\end{figure}
The BLCP sensitivities in Table~\ref{tab:blcp_sensitivity} reveal the same trend: vehicular density remains the dominant control parameter despite the presence of additional geometric degrees of freedom such as $r_0$, $n_{\rm B}$, and $R_{\rm g}$. These additional parameters primarily modulate local geometry but do not rival the impact of $\lambda$.

Fig.~\ref{fig:r_1} shows the sensitivity with respect to the maximum radar range $R_k$. The detection probability initially decreases rapidly with $R_k$ and then saturates, indicating that dominant interference originates within the first few hundred meters. Fig.~\ref{fig:r_2} plots the sensitivity with respect to beamwidth $\Omega$, which is consistently negative and grows in magnitude for larger $\Omega$, confirming that wide beams impose a rapidly increasing interference penalty.

Taken together, the sensitivity results indicate that vehicular density $\lambda$ is the key parameter governing radar-to-radar interference and detection performance, with normalized sensitivities that are more than two orders of magnitude larger than those associated with road intensity, radar range, or beamwidth. In dense deployments, network-level mechanisms that effectively reduce the \emph{effective} density of simultaneously active radars (e.g., duty-cycle control, scheduling, or coordinated transmission policies) are therefore the most impactful levers for improving $p_{{\rm D},{\rm P}}$. Adjusting $R_k$ and $\phi_b$ can provide additional performance tuning, particularly avoiding unnecessarily large sensing ranges and excessively wide beams; however, these parameters offer only moderate gains compared to density-aware interference management. These results reinforce that the primary mechanism governing radar detectability in vehicular environments is the density of active vehicles, beamwidth and range of automotive radars rather than the underlying street geometry.

\begin{table}[h!]
\centering
\begin{tabular}{|l|c|c|}
\hline
Parameter & Baseline Value & $S_\eta$ \\ \hline
$\lambda$          & $0.01$      & $-1.308\times 10^{-1}$ \\
$\lambda_{\rm L}$  & $0.01$    & $-6.97\times 10^{-3}$ \\
$R_{\rm P}$              & $500$        & $-3.08\times 10^{-3}$ \\
$\Omega$  & $15^\circ$     & $-1.416\times 10^{-2}$ \\ \hline
\end{tabular}
\caption{Sensitivity of PLCP detection probability at baseline $p_{{\rm D},{\rm P}} = 0.877s$.}
\label{tab:plcp_sensitivity}
\end{table}

\begin{table}[h!]
\centering
\begin{tabular}{|l|c|c|c|}
\hline
Parameter & Baseline Value & $S_\eta$ \\ \hline
$\lambda$          & $0.01$     & $-6.27\times 10^{-1}$\\
$r_0$              & $1500$         & $3.396 \times 10^{-6}$ \\
$R_{\rm B} $ & $500$     & $-1.025\times 10^{-1}$ \\
$\Omega$  & $15^\circ$    & $-5.2\times 10^{-1}$ \\
$n_{\rm B}$        & $500$      & $-2.598\times 10^{-1}$ \\
$R_{\rm g}$        & $1500$       & $2.599\times 10^{-1}$ \\\hline
\end{tabular}
\caption{Sensitivity of BLCP detection probability at baseline $p_{{\rm D},{\rm B}} = 0.5255$.}
\label{tab:blcp_sensitivity}
\end{table}

\begin{figure}[t]
    \centering
    \includegraphics[width=0.5\linewidth]{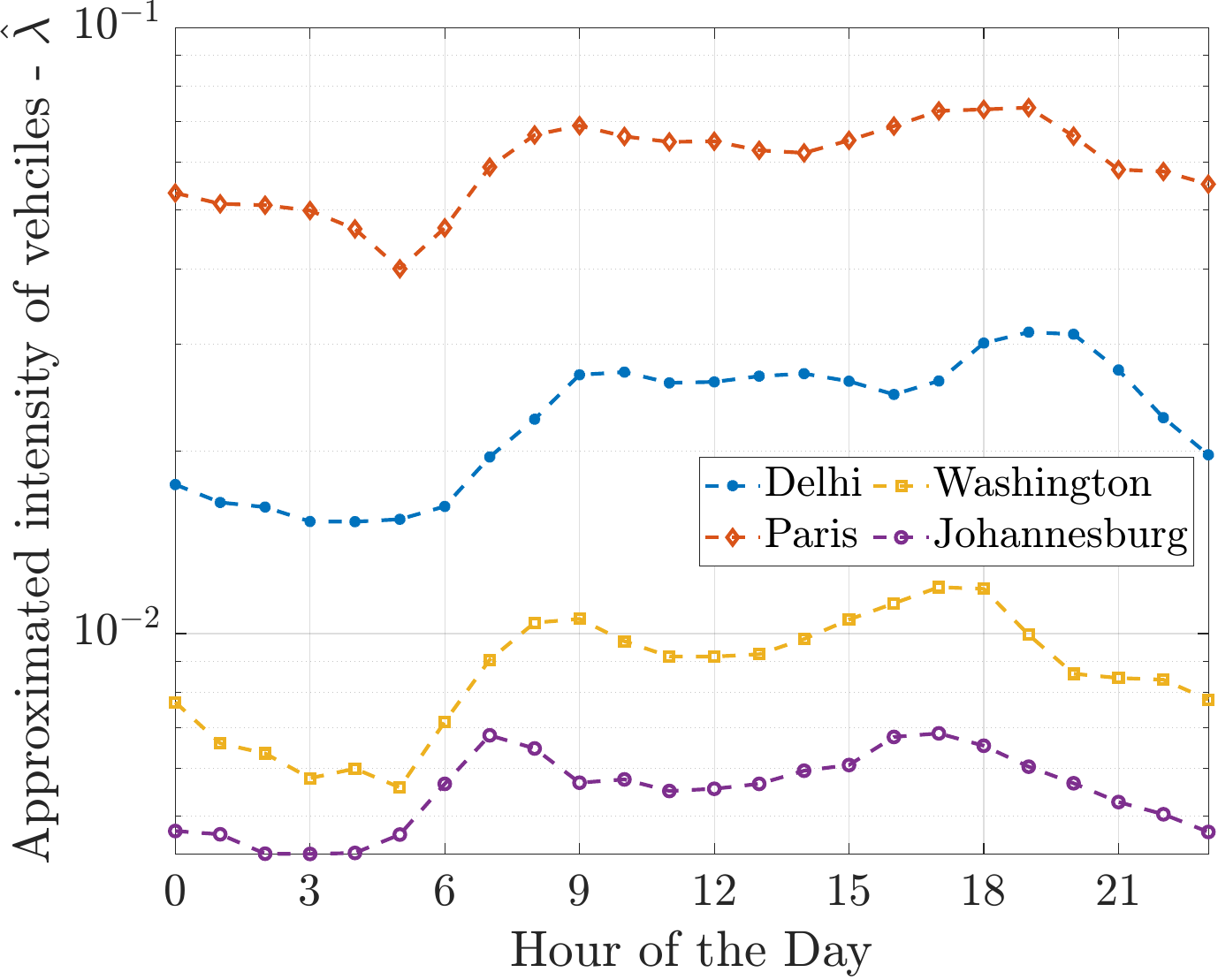}
    \caption{$\hat{\lambda}$ w.r.t. hour of the day for four cities}
    \label{fig:approxL}
\end{figure}

\subsection{Performance of ego radar in real-world road geometries}
This section analyzes the urban traffic patterns by generating detailed geospatial maps for four major metropolitan areas: New Delhi, Paris, Washington, and Johannesburg. These maps include detailed road networks and important urban infrastructure. They also extend to surrounding areas, comprehensively illustrating the urban sprawl and its effects on traffic patterns. To generate the road networks, we use the OSMnx Python module that models analyzes, and visualizes OpenStreetMap street networks~\cite{boeing2025modeling}. Using this real-world data, we approximate the parameters for the \ac{PLP} and \ac{BLP}, enabling a robust statistical representation of road patterns. In Table~\ref{MainTable}, the first column illustrates the road network of the four major cities, and in the second column, we extend the maps to cover the greater suburban areas corresponding to each of the cities. The data of the road networks shown in the first column is used to approximate the \ac{PLCP} parameter $\lambda_{\rm L}$. For \ac{PLP}, the total length of lines in a given region is the product of the $\lambda_{\rm L}$ and the region's area. By considering the geographical region to be a circle, we estimate the approximate street/line density, $\hat{\lambda}_{\rm L}$, as
\begin{align*}
    \hat{\lambda}_{\rm L} = \frac{\pi}{\texttt{street\textunderscore density\textunderscore km}}
\end{align*}
where \say{street\textunderscore density\textunderscore km} is the density of the street found using OSMnx. The $\hat{\lambda}_{\rm L}$ for the four cities are $\{0.004, 0.0052, 0.0053, 0.0037\}$ m$^{-1}$. Now, to fit the \ac{BLP} model, we use the \textit{total length of street} obtained from OSMnx data and the result of Theorem 2 from~\cite{shah2024binomial}. This Theorem derives the \textit{line length density} $\rho (r)$ of \ac{BLP}, where $r$ is the distance from the center of the generating circle. 
Given a dataset $\left\{\left(r_i, \iint \limits_{r \in \mathbb{R},\, \theta \in [0, 2\pi]} \rho_i \right) \right\}_{i=1}^{N}$, where $r_i$ represents the distance from the city center and $\iint \rho_i$ denotes the total length of streets of a city in a circle of radius $r$, we estimate the parameters $(n_{\rm B}, R_{\rm g})$ by fitting the function $\iint \rho(r)$ to the total length of streets data. By integrating $\rho (r)$ over the region of interest, we get the area's total length of \ac{BLP} lines. The total length of streets of a city in a given area can be obtained from \say{street\textunderscore length\textunderscore total} using OSMnx. In our analysis, we assume that at the center of each city, there is a square bounding box of a specific length, and we calculate the total length of lines as the dimensions of the box increase. To obtain the best-fitting parameters, we minimize the sum of squared residuals:
\begin{align*}
    \min_{n_{\rm B}, R_{\rm g}} \sum_{i=1}^{N} \left(\,\,\, \iint\limits_{r \in \mathbb{R}, \theta \in [0, 2\pi]} \!\!\!\!\!\!\!\! \rho(r_i) - \texttt{street\textunderscore length\textunderscore total} \!\right)^2.
\end{align*}
The third column of Table~\ref{MainTable} plots the total length of lines w.r.t the dimensions of the box. 
As an example in the third column and first row, the approximated parameters for Delhi, India are $\hat{n}_{\rm B} = 662$ and $\hat{R}_{\rm g} = 12.96\, {\rm Km}$. Likewise, in the second column, approximated parameters $\hat{n}_{\rm B}$ and $\hat{R}_{\rm g}$ for other cities can be found. In the third column, the four figures illustrates the effectiveness of the approximated parameters in fitting real-world data. Specifically, it shows how the total length of streets varies with the dimensions of the bounding box, with the observed data represented by blue circle markers. We also plot the total length of streets for both the \ac{PLP} and \ac{BLP} approximated parameters. These plots exhibit remarkable precision of \ac{BLP} in modeling road networks as compared to \ac{PLP}. The empirical investigation demonstrates a strong correlation between the \ac{BLP} models and real road structures.
\begin{table*}[t]
\centering
\begin{tabular}{|l|l|l|l|}
\hline
\multicolumn{1}{|c|}{\begin{tabular}[c]{@{}c@{}}Map of Cities considered\\ for \ac{PLCP}\end{tabular}} & \multicolumn{1}{c|}{\begin{tabular}[c]{@{}c@{}}Extended Map of Cities\\ considered for \ac{BLCP}\end{tabular}} & \multicolumn{1}{c|}{\begin{tabular}[c]{@{}c@{}}Fitting \ac{BLP} and \ac{PLP}\\parameters to real data\end{tabular}} & \multicolumn{1}{c|}{\begin{tabular}[c]{@{}c@{}}Detection performance \\ of \ac{PLCP} and \ac{BLCP}\end{tabular}}\\
\hline
\hline
\centering \includegraphics[trim={7cm 0cm 0cm 0cm},clip,height=0.17\textwidth,width=0.2\textwidth]{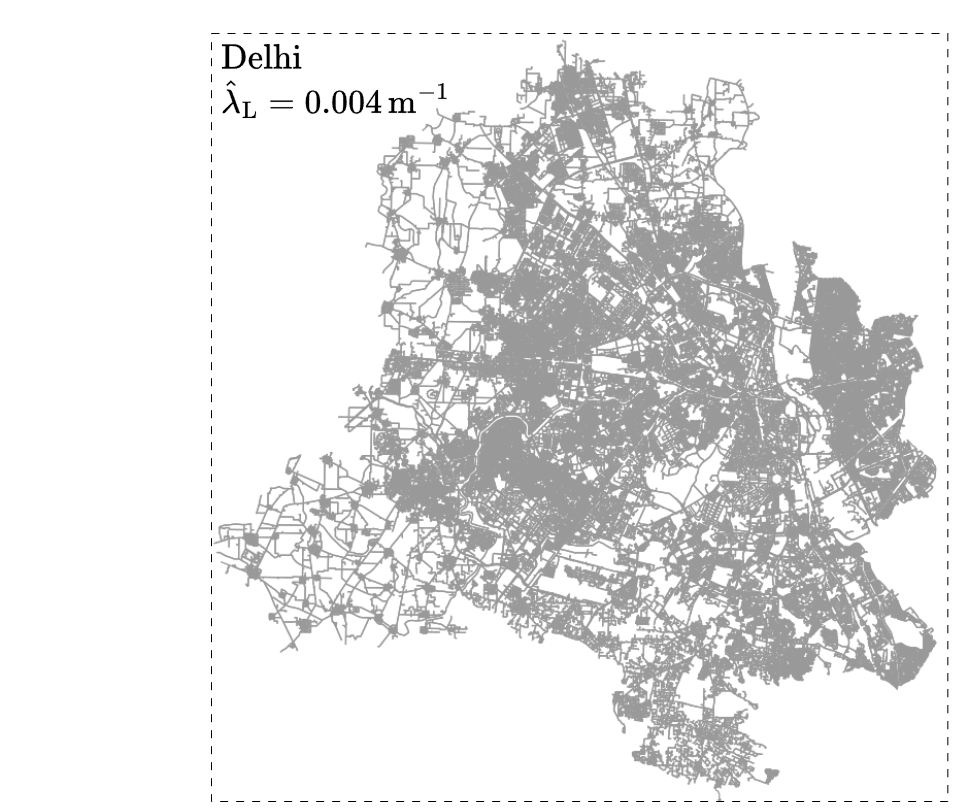} & \includegraphics[height=0.17\textwidth,width=0.2\textwidth]{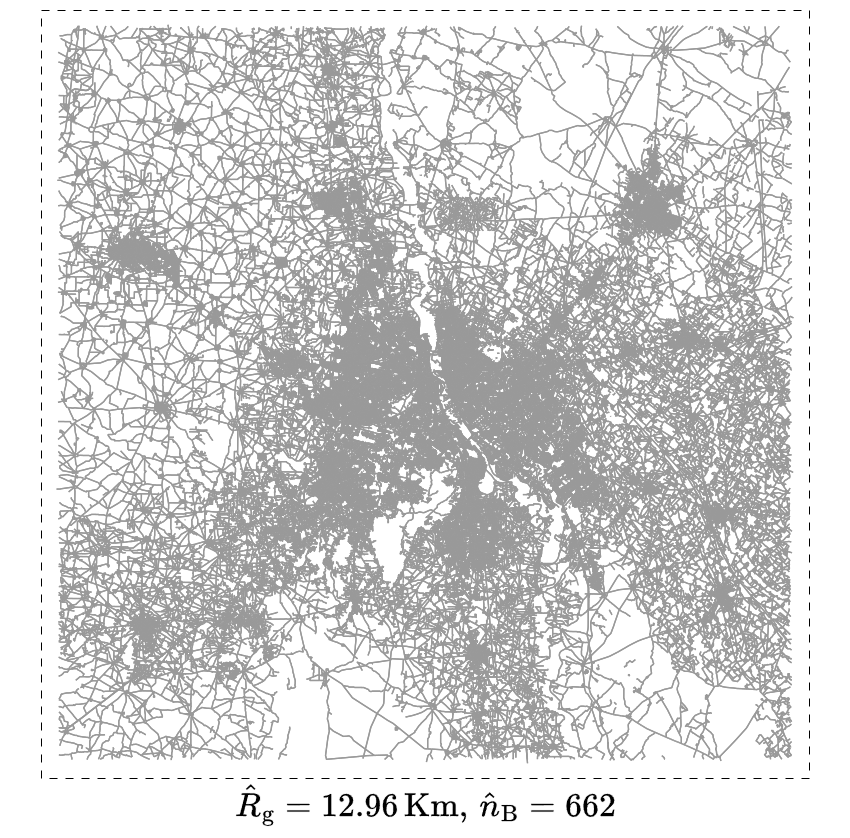} & \includegraphics[height=0.16\textwidth,width=0.2\textwidth]{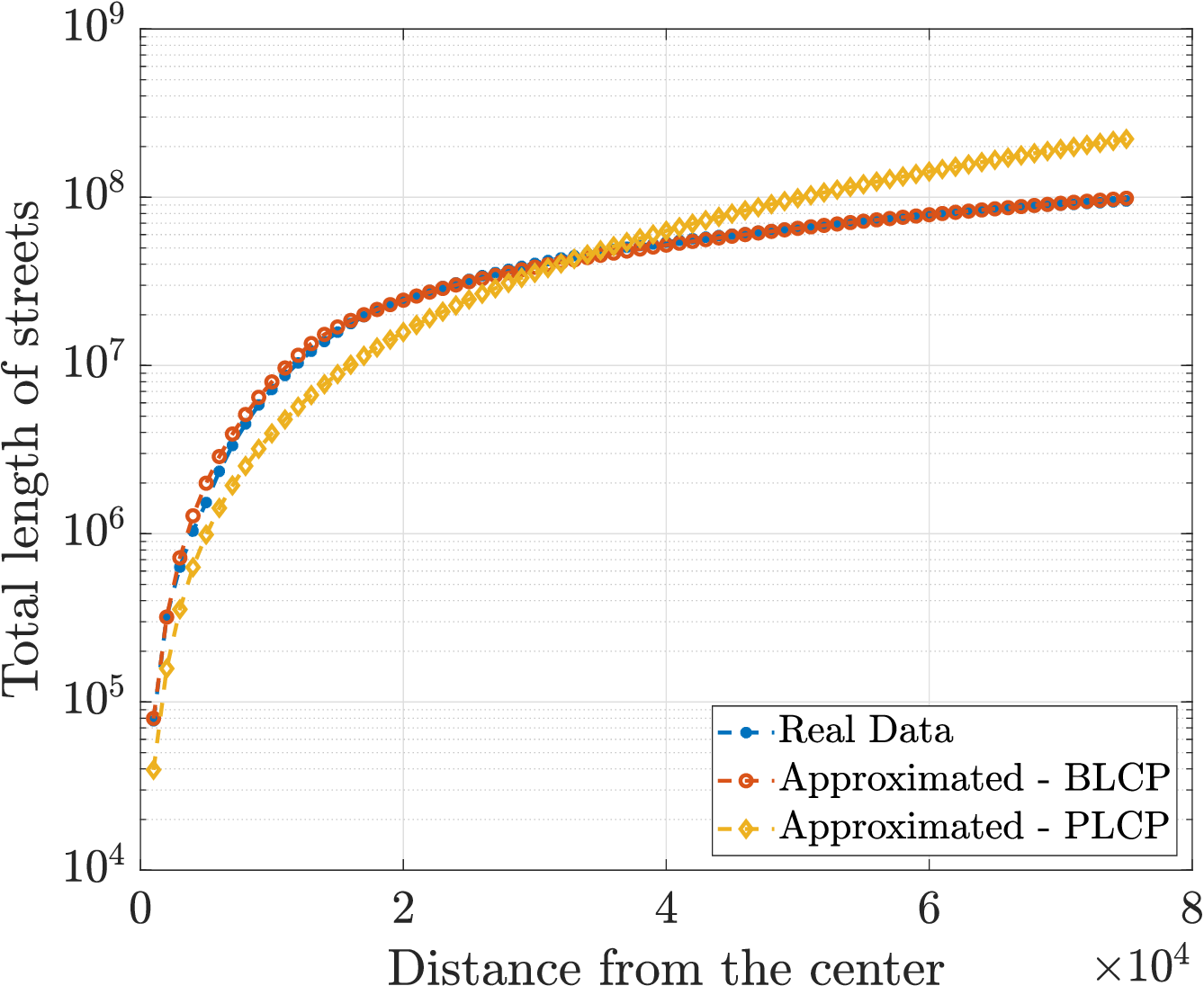} &  \includegraphics[height=0.16\textwidth,width=0.2\textwidth]{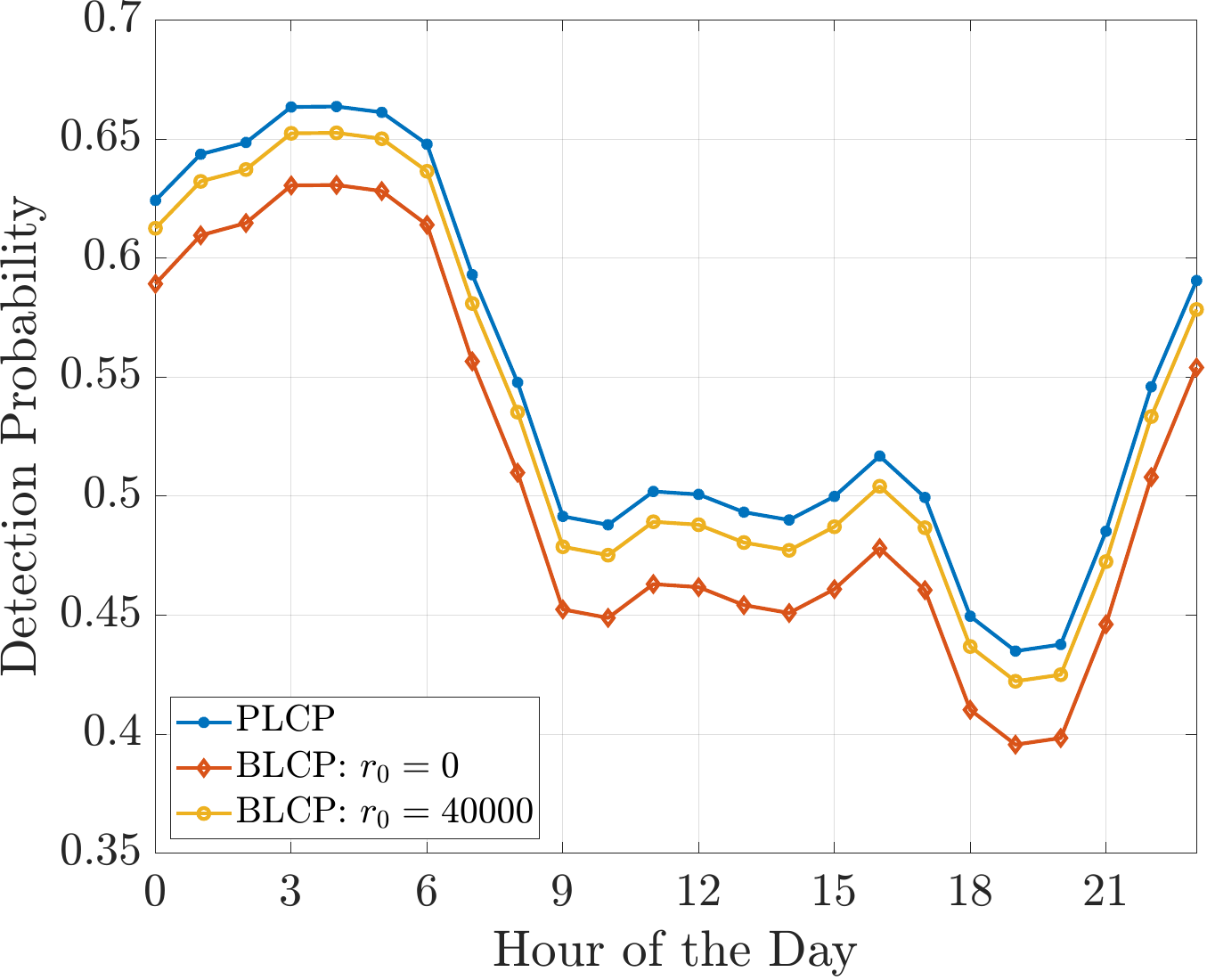}\\
\hline
\includegraphics[trim={7cm 0cm 0cm 0cm},clip,height=0.17\textwidth,width=0.2\textwidth]{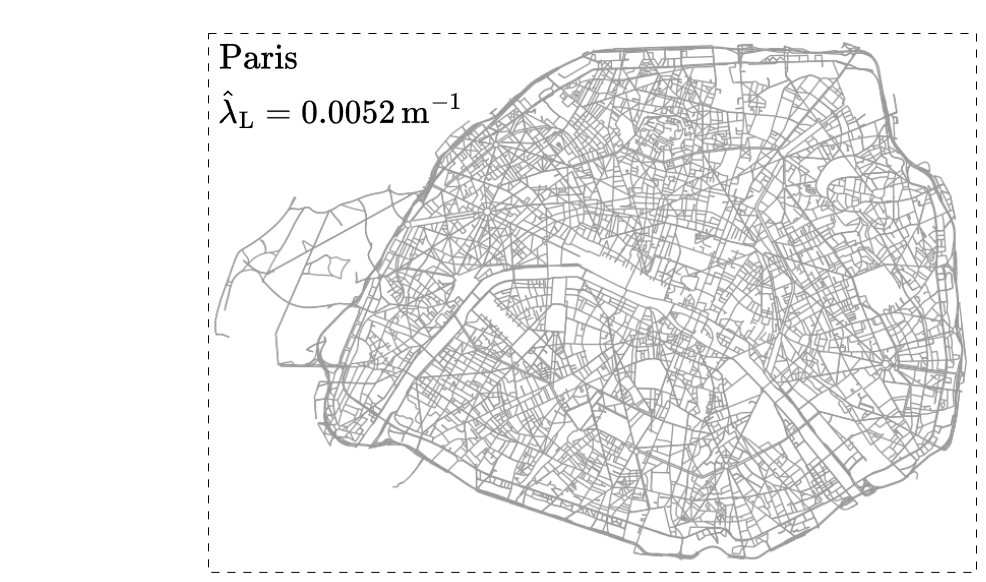} & \includegraphics[height=0.17\textwidth,width=0.2\textwidth]{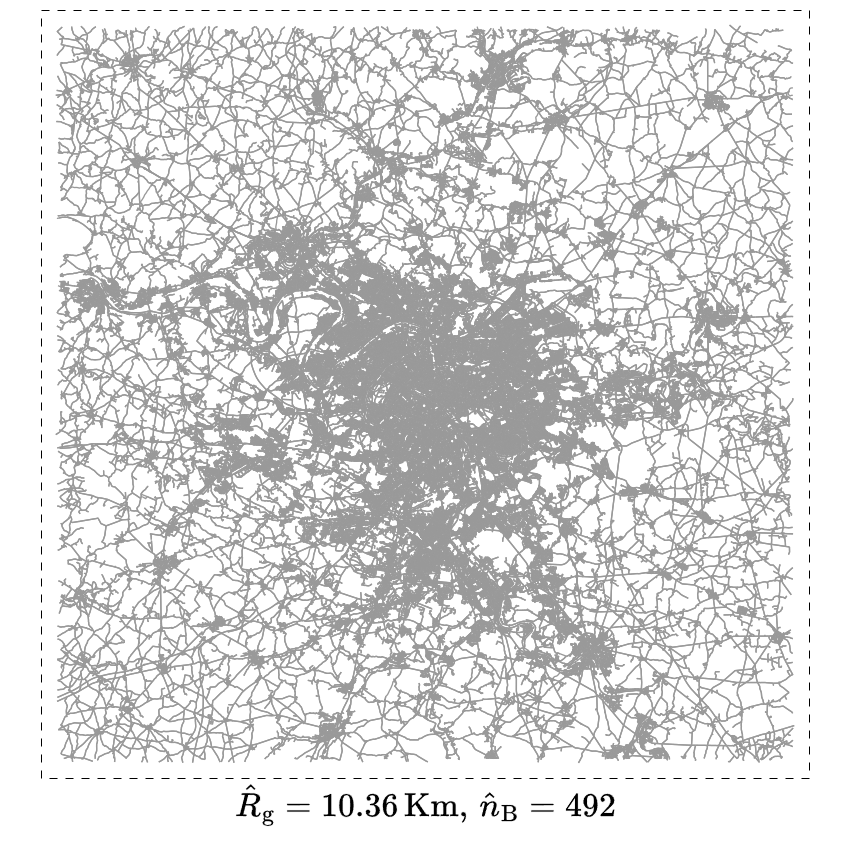} & \includegraphics[height=0.16\textwidth,width=0.2\textwidth]{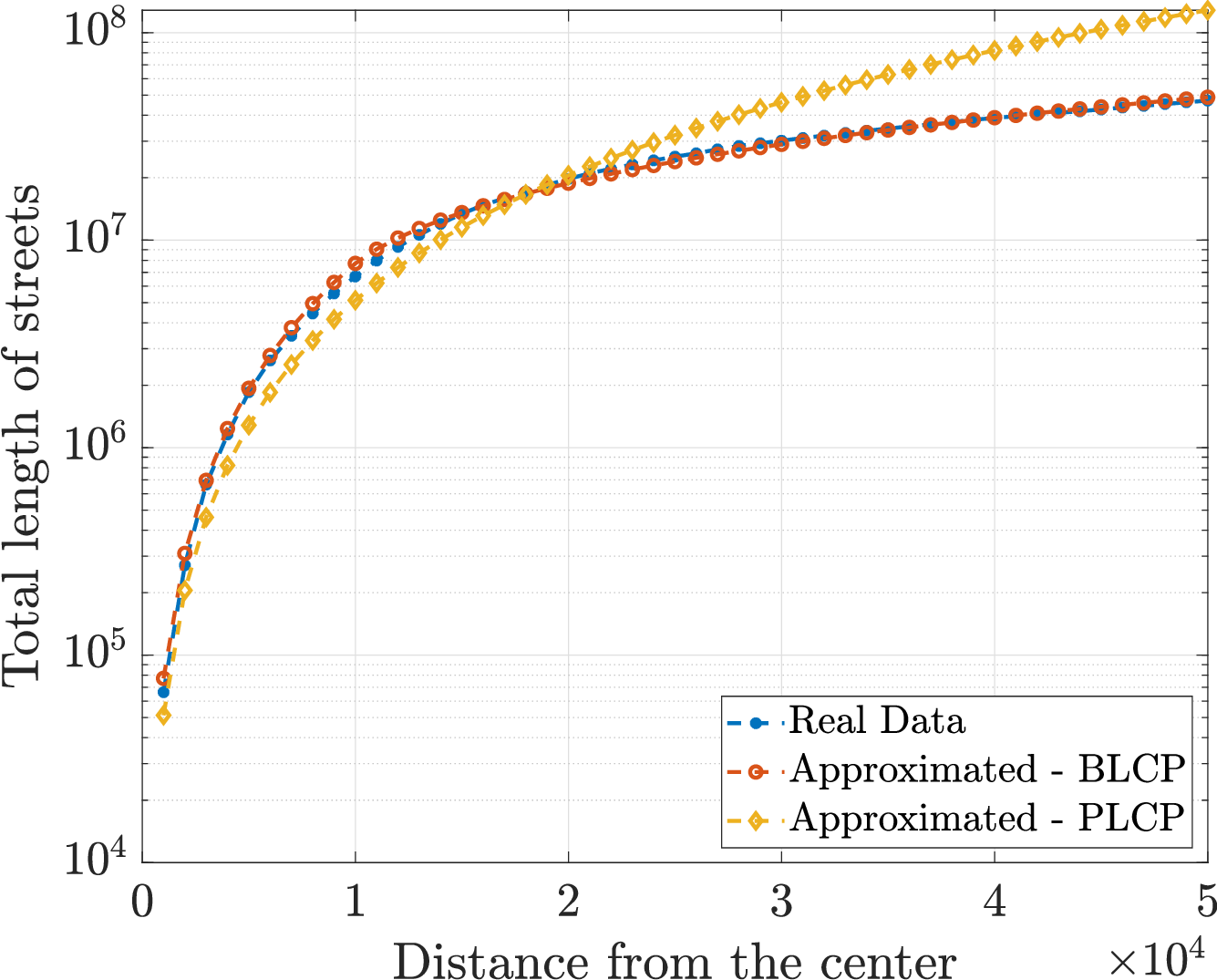} & \includegraphics[height=0.16\textwidth,width=0.2\textwidth]{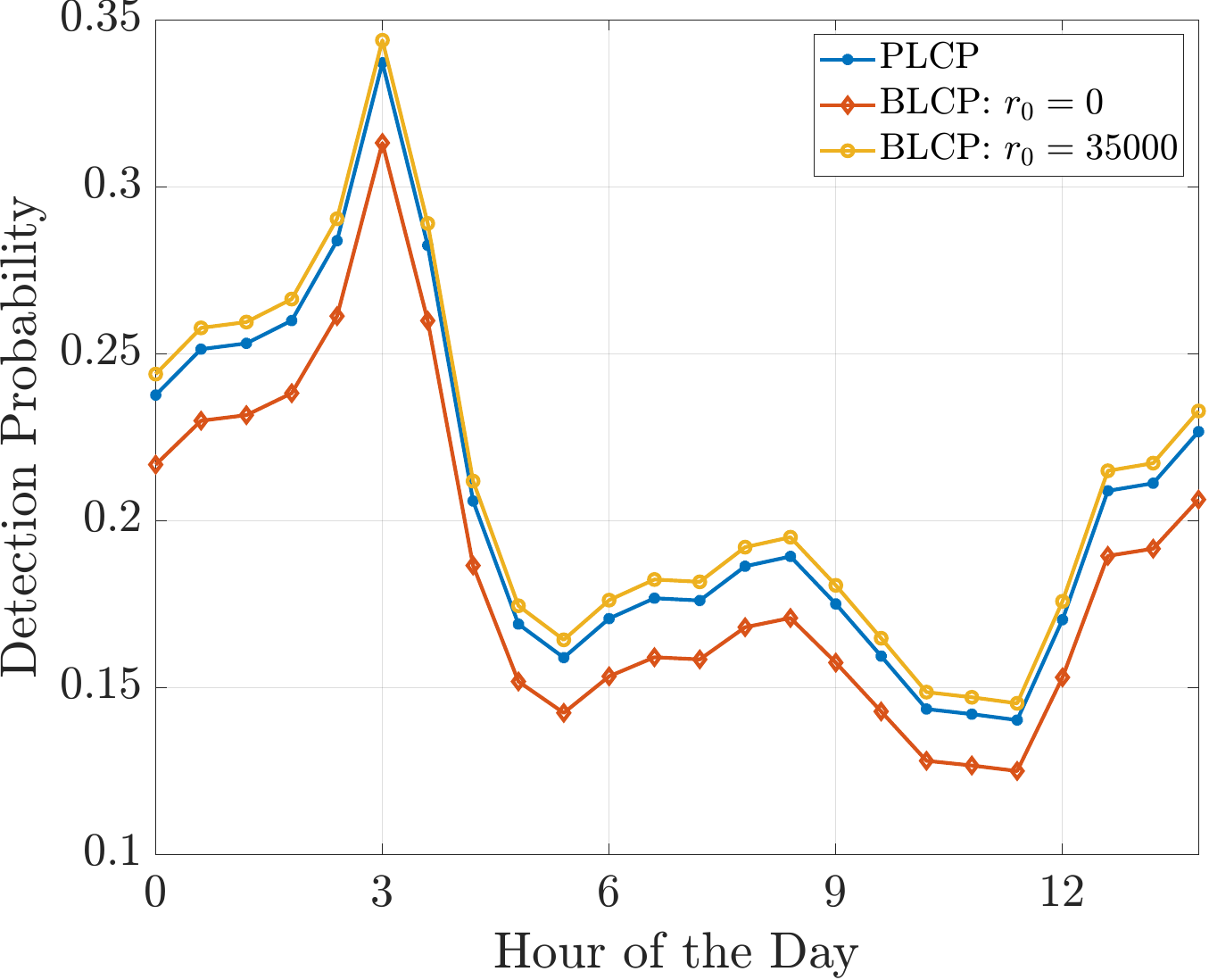} \\
\hline
\includegraphics[trim={7cm 0cm 0cm 0cm},clip,height=0.17\textwidth,width=0.2\textwidth]{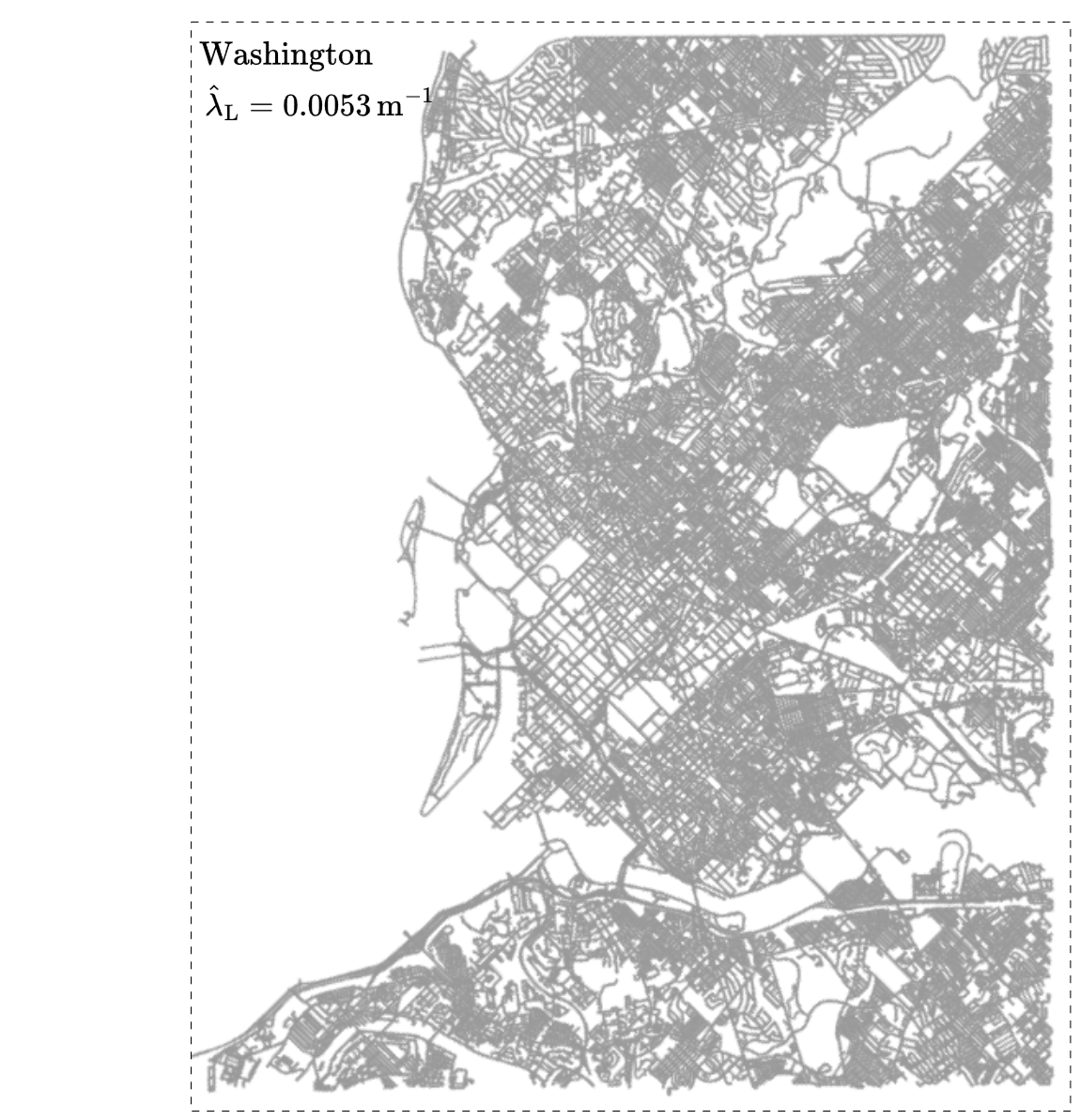} & \includegraphics[height=0.17\textwidth,width=0.2\textwidth]{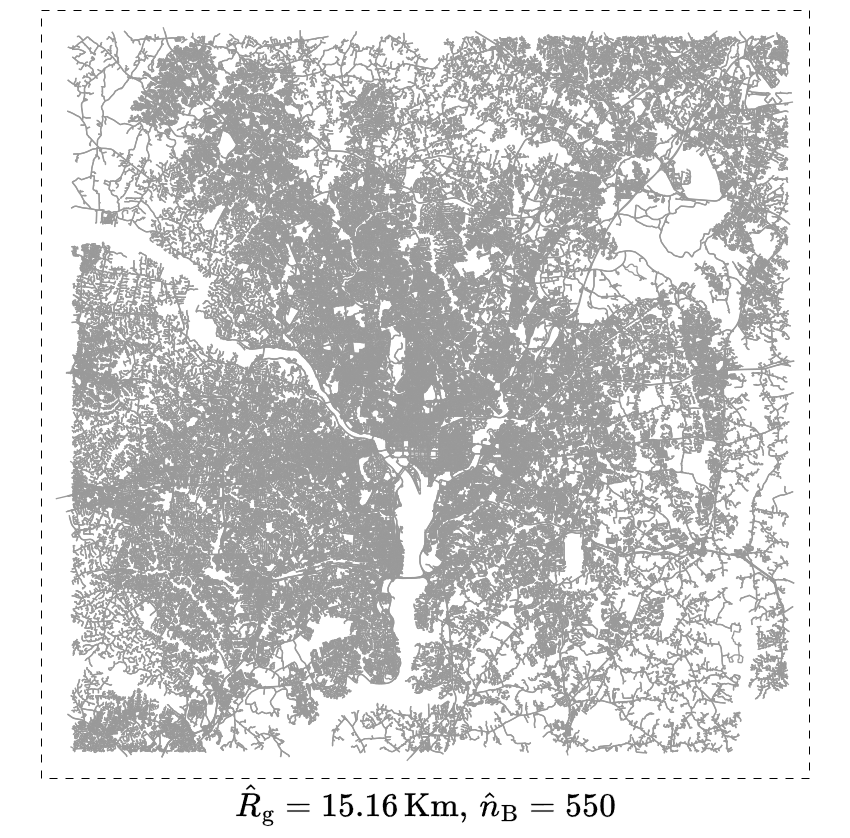} & \includegraphics[height=0.16\textwidth,width=0.2\textwidth]{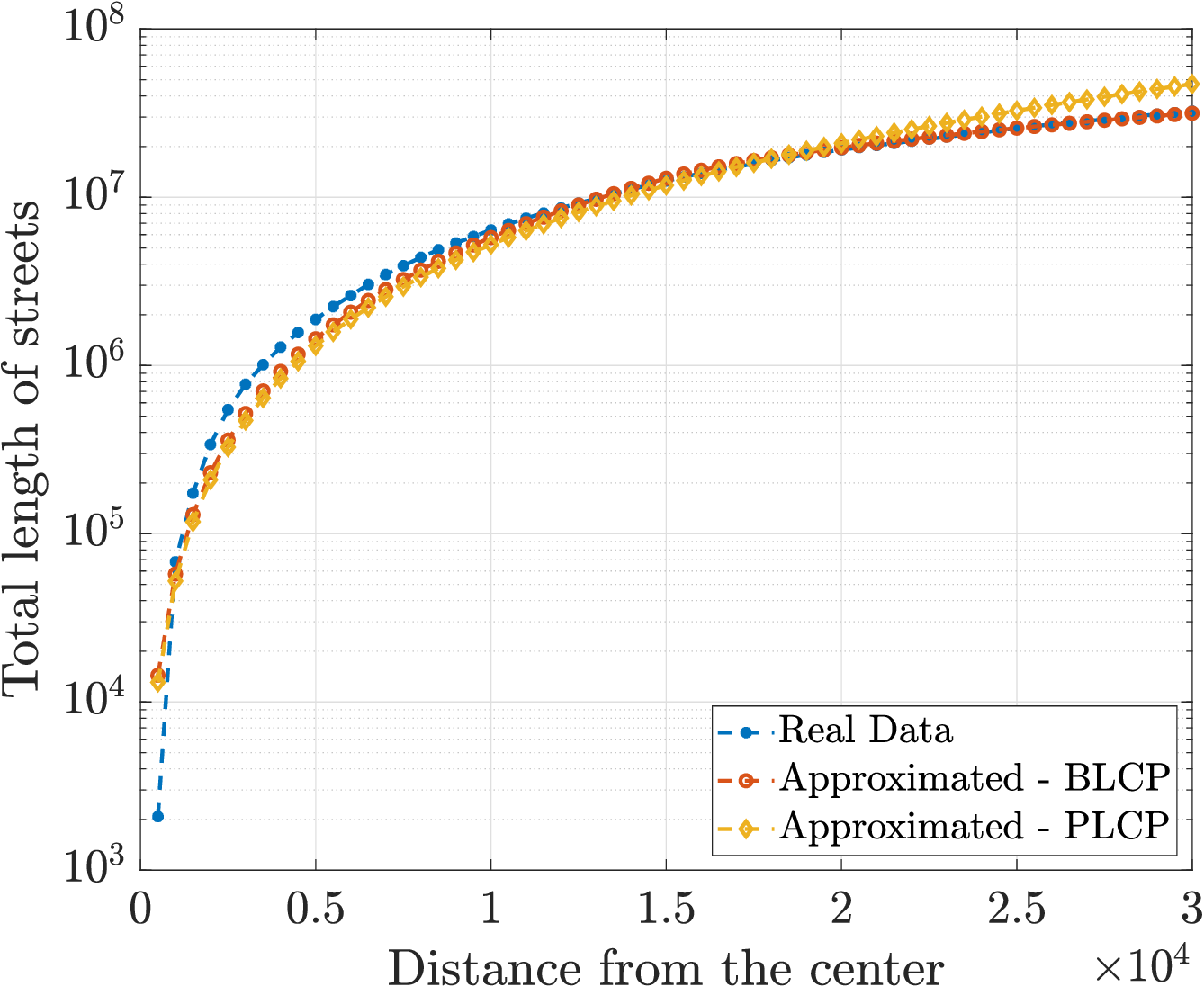} & \includegraphics[height=0.16\textwidth,width=0.2\textwidth]{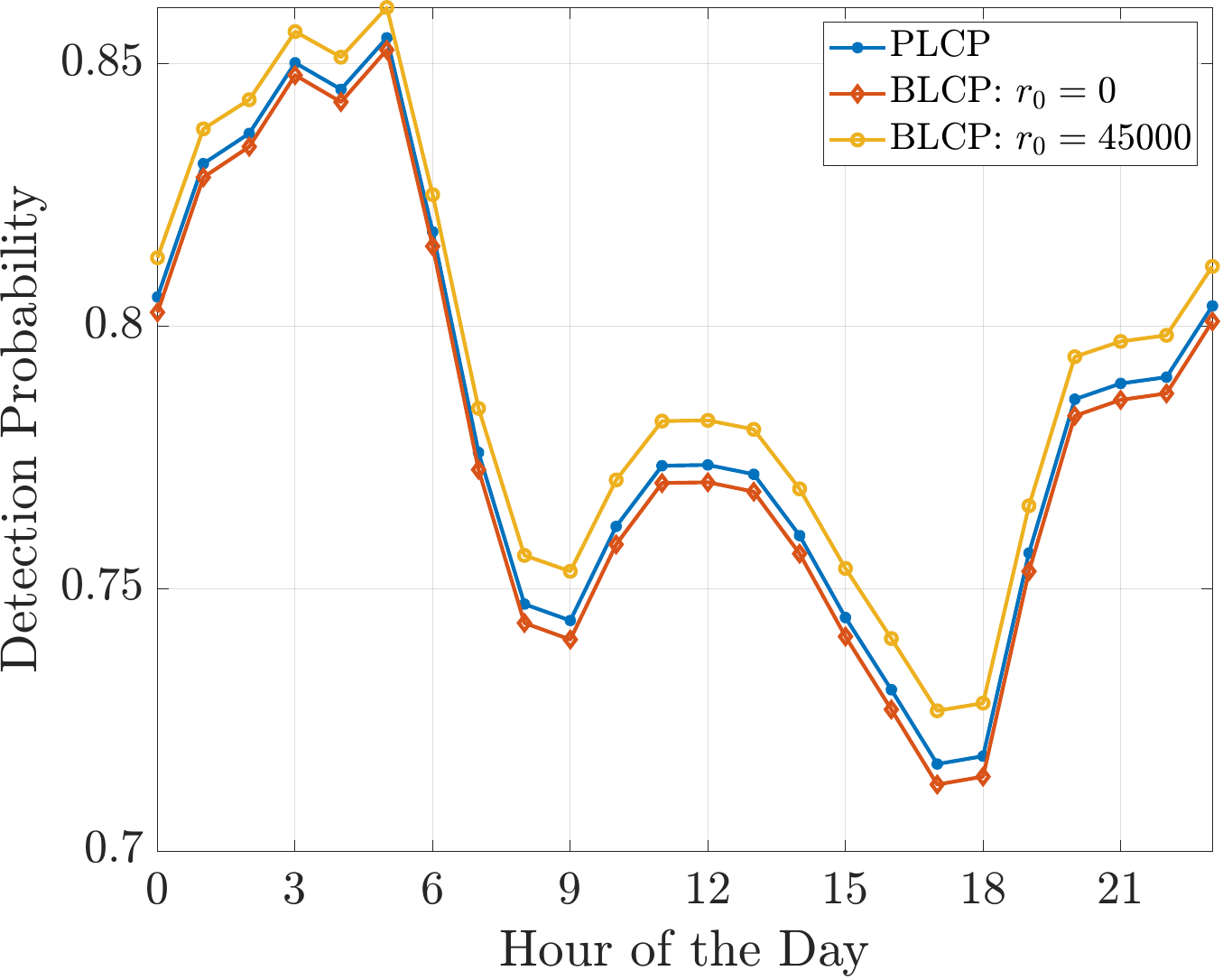} \\
\hline
\includegraphics[trim={8.5cm 0cm 0cm 0cm},clip,height=0.17\textwidth,width=0.2\textwidth]{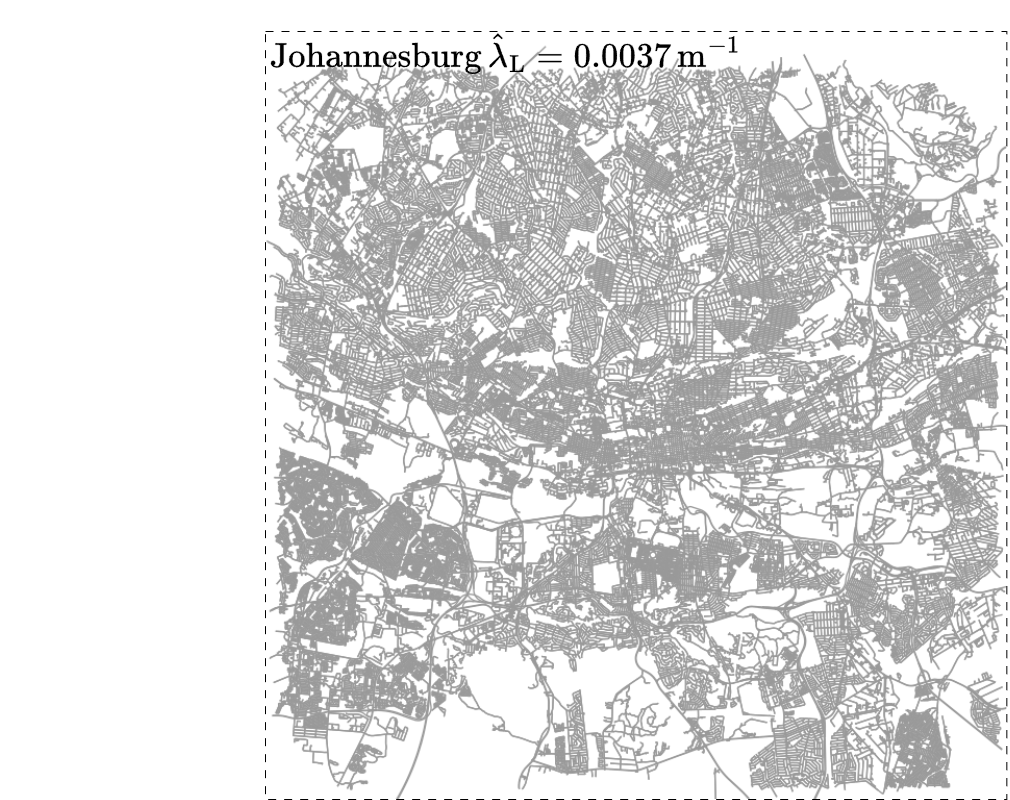} & \includegraphics[height=0.17\textwidth,width=0.2\textwidth]{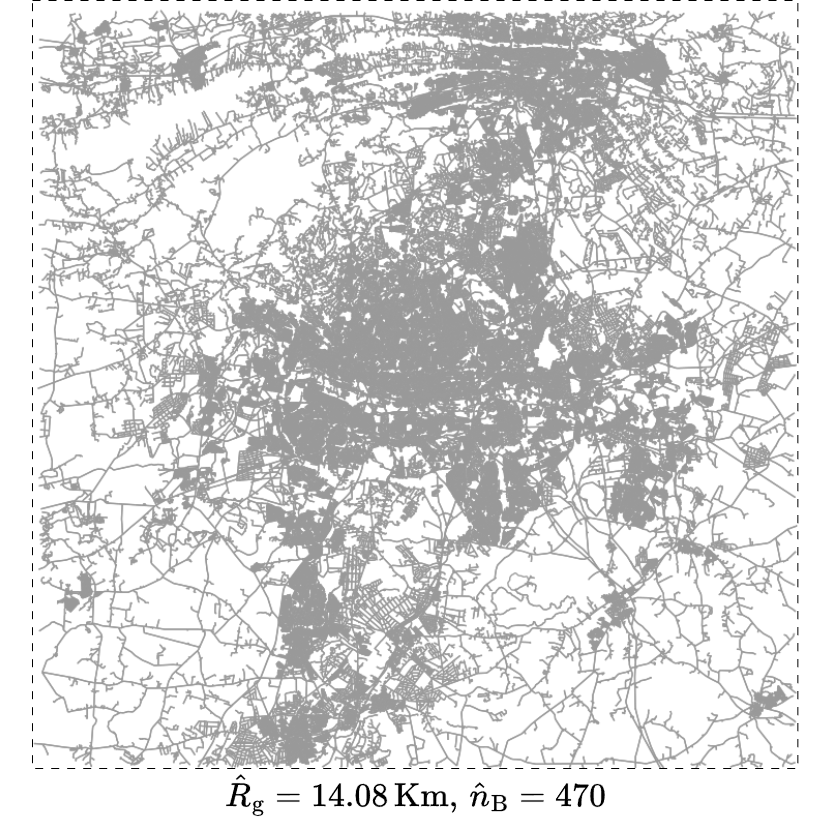} & \includegraphics[height=0.16\textwidth,width=0.2\textwidth]{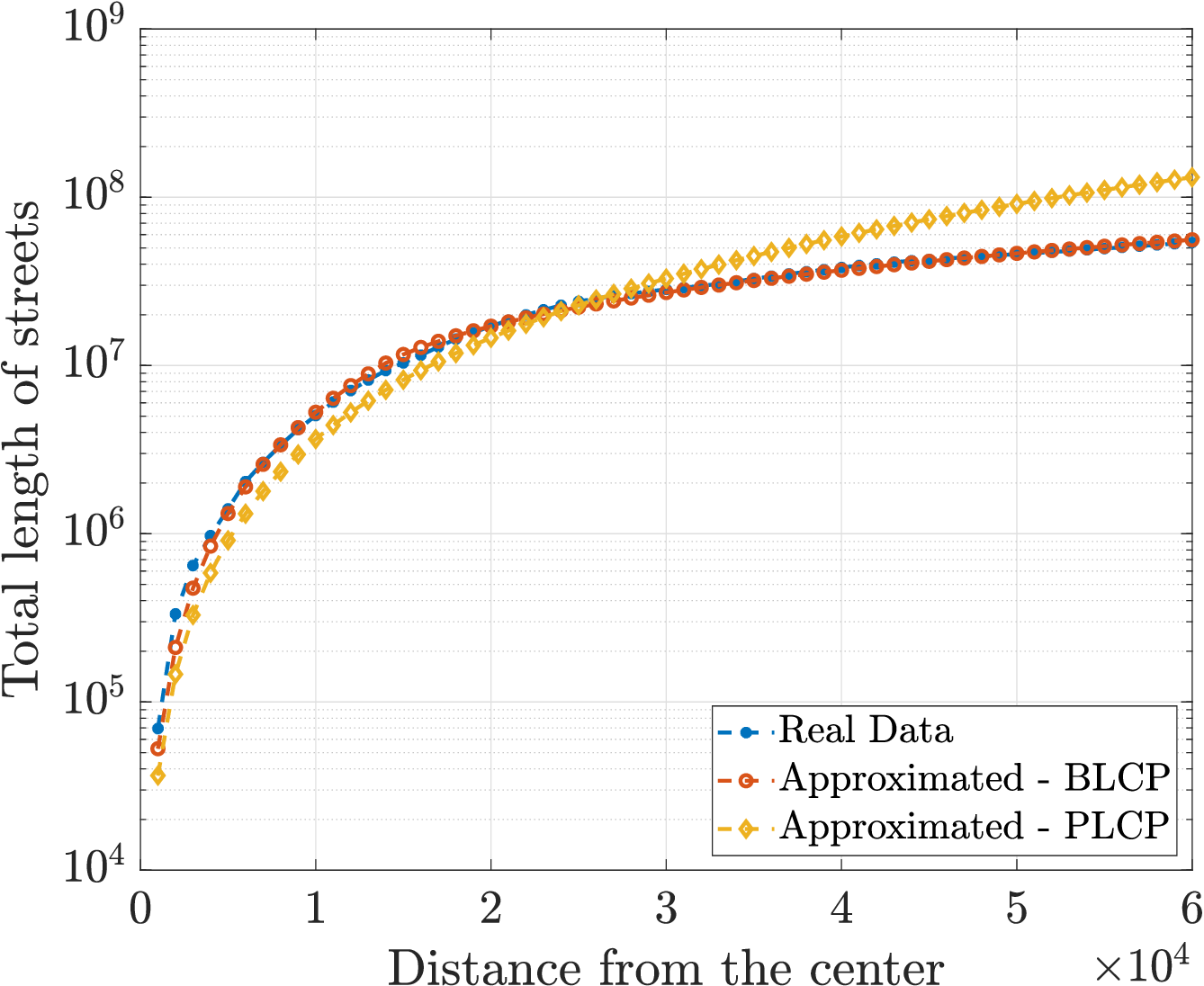} & \includegraphics[height=0.16\textwidth,width=0.2\textwidth]{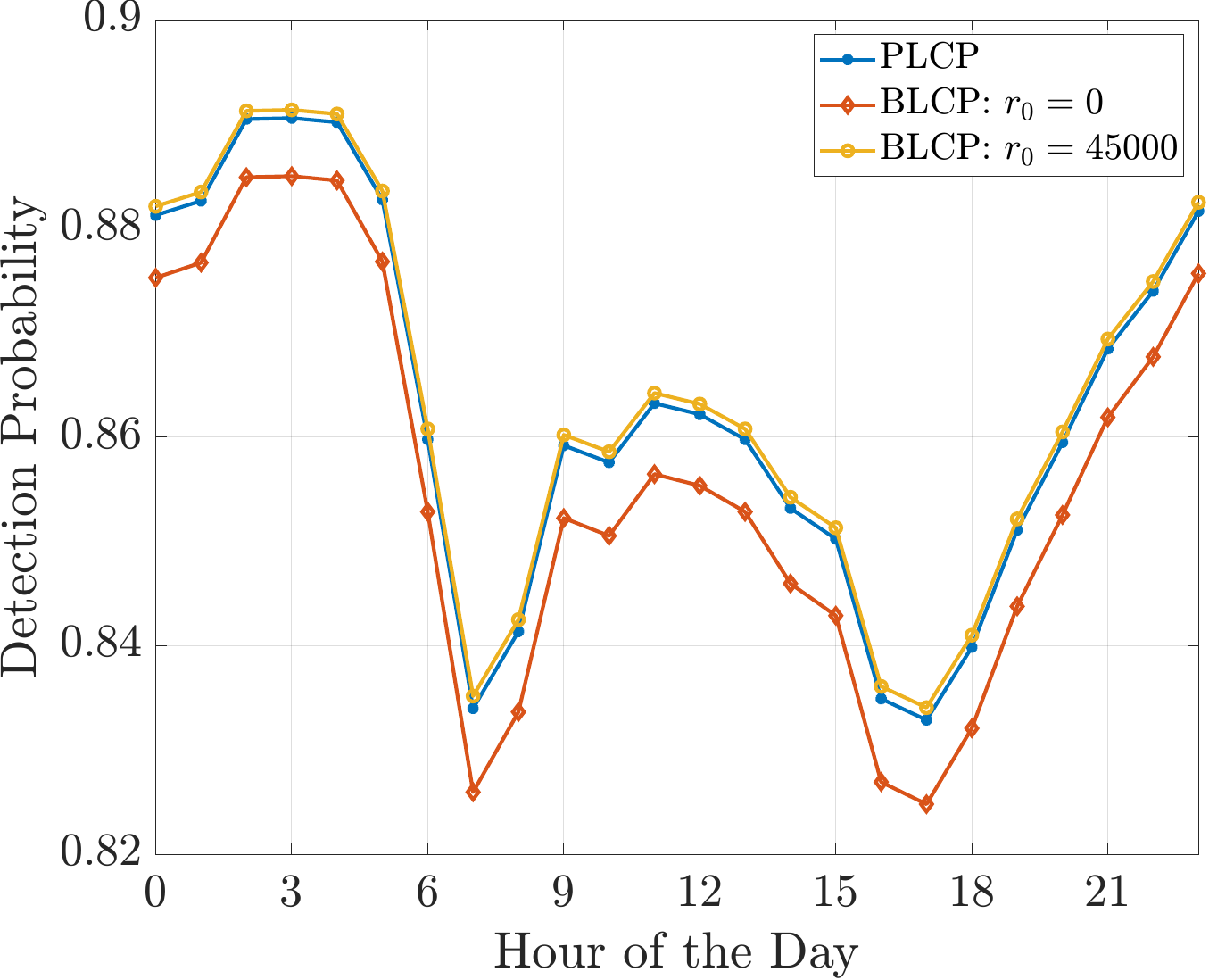} \\
\hline
\end{tabular}
\caption{City maps and their corresponding detection performance across the hour of the day}
\label{MainTable}
\end{table*}

Finally, we consider the vehicular density on the roads. This parameter is a function of the nature of a city - cities with extensive public transport generally have fewer vehicles - as well as the time of the day - peak hours have higher vehicular density than off-peak hours. 
Therefore, in the fourth column of Table~\ref{MainTable}, we plot the detection probability of an ego radar for the four cities modeled using the approximated parameters w.r.t to the hour of the day. To do so, we compute the vehicular density, i.e., $\lambda$ at each hour from the report in~\cite{hallenbeck1997vehicle}. The traffic density, or the hourly average number of vehicles, tends to have a bimodal profile with morning and evening peaks. In addition, we utilize the hourly congestion $C_i$ data provided by TomTom to factor the hourly fluctuations in congestion. From this data, we determine the hour of the day when the congestion is at its peak and then compute the average number of vehicles at this hour. Figure 12 of~\cite{hallenbeck1997vehicle} showcases the percent of average daily traffic w.r.t hour of the day, where the peak traffic occupancy percentage is around $8\%$. By taking the product of peak traffic occupancy percentage and total number of non-commercial vehicles present in the city and then dividing it by the total length of roads, it gives the approximated $\hat{\lambda}_{\max}$ at peak congestion hour. Assuming a linear relation between the intensity of vehicles and the congestion percentage, we approximate the $\lambda$ for all other hours employing the data from TomTom, using $\hat{\lambda}_i = \hat{\lambda}_{\max} C_i/C_{\max}$, where $\hat{\lambda}_i$ is the vehicular intensity, $C_i$ is the congestion percentage at $i\textsuperscript{th}$ hour, and $C_{\max}$ is the maximum congestion percentage. Fig.~\ref{fig:approxL} illustrates $\hat{\lambda}$ w.r.t hour of the day for the four cities. Here, Paris has a higher intensity of vehicles across all hours than the other three cities. In the fourth column, we plot the corresponding detection probability of an ego radar for both \ac{PLCP} and \ac{BLCP} frameworks with $r_0$ at the city center and city outskirts w.r.t to the hour of the day. As anticipated, the overall detection probability is the poorest in Paris, which has high traffic intensity, followed by Delhi, and Washingston, while Johannesburg performs best. These results are consistent with the traffic index ranking provided by TomTom. In each city, the radars perform best in the early morning hours between 3 and 5 am when vehicular traffic is at its lowest.

\subsection{Summary of Key Insights and Design Guidelines}
\label{sec:insights}
The numerical results provide several key physical insights and practical design guidelines. First, radar detection performance is not uniform but varies significantly with location, degrading by up to $40\%$ in dense urban cores compared to suburban areas. This strong spatial dependence underscores the need for cognitive radar systems that dynamically adjust their parameters based on the surrounding environment. Second, vehicular density emerges as the dominant factor affecting radar performance, outweighing other network parameters. This highlights that interference mitigation is fundamentally a congestion management challenge; reducing the number of simultaneous interferers is more effective than optimizing individual link quality alone. Third, extending $R_k$ beyond a certain threshold offers diminishing returns and can actually degrade performance due to heightened interference. Therefore, system designers should calibrate $R_k$ to match actual operational needs rather than pursuing maximal range. Finally, while the PLCP model suffices for analyzing homogeneous urban zones, the BLCP model is indispensable for city-scale planning, as it accurately captures the performance gradient from dense urban to suburban centers.

\section{Conclusion}
We presented a comprehensive analysis of large-scale automotive radar networks, considering the impact of street geometry and vehicular density. By employing \ac{SG}-based models, we account for homogeneous and inhomogeneous street distributions using \ac{PLCP} and \ac{BLCP} models, respectively, providing a complete understanding of radar performance in different urban contexts. Our models bridge simplistic, highway-based assumptions and the complexity of real urban environments and provide crucial insights into radar interference characterization. By deriving the effective interference set and the detection probability of an ego radar, we demonstrated how the radar performance can be significantly influenced by urban street layouts and traffic conditions. The \ac{PLCP} model excels at dense urban regions, while the \ac{BLCP} model captures the city center-to-outskirts transition. The validation of our models with real-world data further strengthens their applicability in urban radar network planning. This research helps automotive OEMs optimize radar systems for urban layouts and traffic conditions, improving ADAS safety and efficiency. 

Although the proposed Cox-process-based framework enables realistic and tractable modeling of large-scale radar interference, direct validation using existing real-world radar datasets or high-fidelity simulators is not feasible at the network scale considered in this work. Current public radar datasets are primarily designed for perception tasks and typically capture measurements from a single ego vehicle, without information on simultaneous multi-radar transmissions, beam interactions, or radar-to-radar interference. Consequently, they cannot be used to validate the SIR-based mutual interference statistics analyzed here. High-fidelity simulators also do not scale to dense, city-wide deployments. In contrast, our \ac{SG} framework enables scalable network-level evaluation by anchoring the spatial topology to real urban street layouts and time-of-day traffic intensities, thereby balancing realism and analytical tractability.


\bibliographystyle{ieeetr}
\bibliography{references}

\appendices
\section{Proof of Theorem 1}
\label{pr:thm_1}
The objective of this proof is to determine the interfering distance bounds $a_{\mathrm{B},i}$ and $b_{\mathrm{B},i}$ which define the segment of line $L_i$ that lies within the ego radar's bounded radar sector $\mathcal{S}^{\rightarrow}_{\mathbf{u}_0,\mathrm{B}}$ (see Definition~1). Because this sector is both \textit{angularly limited} (half-beamwidth $\Omega$) and \textit{radially bounded} (maximum range $R_{\mathrm{B}}$), the intersection of an infinite line $L_i$ with the sector depends critically on two geometric quantities:
\begin{itemize}
    \item the signed distance $d_i = r_i / \sin\theta_i - r_0$ from the ego radar to the point where $L_i$ crosses the ego street $L_0$,
    \item the angle $\theta_i \in [0,2\pi)$ that orients $L_i$ relative to the ego boresight.
\end{itemize} 
The derivation is entirely geometric and follows from the exact intersection between a bounded radar sector and a randomly oriented line. Depending on the angle $\theta_i$ and the offset $d_i$ , the line may intersect (i) the circular arc of the sector, (ii) one of the two straight radial edges, or (iii) both simultaneously at a corner point. Furthermore, the intersection may lie ahead of, behind, or outside the operational range of the ego radar. These geometric configurations partition the $\left(\theta_i, d_i\right)$ space into disjoint regions, each requiring a distinct trigonometric treatment. Crucially, due to the four-fold symmetry of the plane, we can derive expressions for $\theta_i \in [0,\pi/2]$ and extend them to all quadrants using absolute values of sine and cosine, this avoids redundant case enumeration.

Without loss of generality, we place the ego radar at any arbitrary location $(0,r_0)$ and align its boresight along the positive $y$-axis i.e. $\mathbf{a} = (-1, 0)$. The sign of $d_i$ distinguishes whether the intersection between $L_i$ and the ego street $L_0$ occurs ahead of ($d_i \ge 0$) or behind ($d_i < 0$) the ego radar. Interference is possible only when $L_i$ intersects the radar sector. We therefore restrict attention to orientations $\theta_i \in \Theta$, where $\Theta$ is the union of angular intervals defined as 
\begin{align}
\mathrm{A}_i = 
\begin{cases}
    \mathrm{A_1} =  \left\{\theta_i \colon 0 \leq \theta_i \leq \Omega \;\mathrm{and}\; 0 \leq \theta_i \leq \frac{\pi}{2} \right\} \\
    \mathrm{A_2} =  \left\{\theta_i \colon \Omega < \theta_i \leq 2\Omega \;\mathrm{and}\; 0 \leq \theta_i \leq \frac{\pi}{2} \right\} \\
    \mathrm{A_3} =  \left\{\theta_i \colon \Omega \leq \pi - \theta_i \leq 2\Omega \;\mathrm{and}\; \frac{\pi}{2} < \theta_i \leq \pi \right\} \\
    \mathrm{A_4} =  \left\{\theta_i \colon 0 \leq \pi - \theta_i < \Omega \;\mathrm{and}\; \frac{\pi}{2} < \theta_i \leq \pi \right\}\\
    \mathrm{A_5} =  \left\{\theta_i \colon 0 \leq \theta_i - \pi \leq \Omega \;\mathrm{and}\; \pi < \theta_i \leq \frac{3\pi}{2} \right\}\\
    \mathrm{A_6} =  \left\{\theta_i \colon \Omega < \theta_i - \pi \leq 2\Omega \;\mathrm{and}\; \pi < \theta_i \leq \frac{3\pi}{2} \right\}\\
    \mathrm{A_7} =  \left\{\theta_i \colon \Omega \leq 2\pi - \theta_i \leq 2\Omega \;\mathrm{and}\; \frac{3\pi}{2} < \theta_i \leq 2\pi \right\}\\
    \mathrm{A_8} =  \left\{\theta_i \colon 0 \leq 2\pi - \theta_i < \Omega \;\mathrm{and}\; \frac{3\pi}{2} < \theta_i \leq 2\pi \right\}.
\end{cases}
\label{eq:app_1}
\end{align}
Let $\Theta$ be the union of events $\mathrm{A}_i$ i.e. $\Theta = \cup_i \mathrm{A}_i$. These intervals collectively represent all street orientations for which the line can intersect the radar sector. The first two events of~\eqref{eq:app_1}, i.e., $\left[\mathrm{A_1}\cup \mathrm{A_2}\right]$ corresponds to the scenario when generating angle is in the first quadrant, likewise event $\left[\mathrm{A_3}\cup \mathrm{A_4}\right]$, $\left[\mathrm{A_5}\cup \mathrm{A_6}\right]$, and $\left[\mathrm{A_7}\cup \mathrm{A_8}\right]$ corresponding to scenario when the generating angle lies in the second, third and fourth quadrant. We now begin to derive the interfering distance, starting with case 1. 

All possible configurations fall into three mutually exclusive geometric cases, illustrated in Fig.~\ref{fig:case1_1},~\ref{fig:case1_3},~\ref{fig:case2}, and~\ref{fig:case3}. Figure~\ref{fig:case1_1},~\ref{fig:case1_3} represent \textbf{case 1} i.e. intersection occurs inside the radar sector. Figure~\ref{fig:case2} and~\ref{fig:case3} represent \textbf{case 2} i.e. intersection occurs ahead of the ego radar but outside the sector and \textbf{case 3} i.e. intersection occurs behind the ego radar respectively.

\textbf{Case 1:} Point of intersection occurs within the radar sector. This is the most common and geometrically rich case. The line $L_i$ intersects the radar sector either through the circular arc, the sector edges, or both. Accordingly, we distinguish two subcases.\\
\textbf{Case 1(a): Arc Intersection}\\
Conditions:
\begin{align*}
    & 0 \leq d_i \leq R_{\rm B}, \;\;\text{and} \\
    & \theta_i \in [0, 2\pi] \setminus \big\{[\Omega, \pi - \Omega] \cup [\pi + \Omega, 2\pi - \Omega]\big\}
\end{align*}
In this subcase, events $\mathrm{A_1} \cup \mathrm{A_4} \cup \mathrm{A_5} \cup \mathrm{A_8}$ take place. Fig.~\ref{fig:case1_1} and~\ref{fig:case1_2} illustrates an example of this subcase, wherein the $\theta_i$ is generated in first quadrant i.e. $\theta_i \in \mathrm{A_1}$. In this subcase, $L_i$ intersects the curved boundary of the radar sector. From Fig.~\ref{fig:case1_2}, we observe that the minimum interfering distance $a_{{\rm B},i}$ is calculated from point $\mathrm{X}$ to $\mathrm{Y}$ where the negative sign reflects measurement along the line relative to the intersection point. Likewise maximum interfering distance $b_{{\rm },i}$ is calculated from point $\mathrm{X}$ to $\mathrm{V}$. The purple points represent the two end values of $d_i$. For simpler representation, $\mathrm{XY}$ represents the distance between points $\mathrm{X}$ and $\mathrm{Y}$. Thus $a_{{\rm B},i}$ is then found as
\begin{align}
    a_{{\rm B},i} &= -\mathrm{XY} = -\left(\mathrm{XZ} - \mathrm{YZ}\right) = -\left(d_i \cos\theta_i - \mathrm{OZ} \cot\Omega\right) \nonumber\\
    &=  -\left(d_i \cos\theta_i - d_i\sin\theta_i \cot\Omega\right).
    \label{eq:eq_appx_1}
\end{align}
To determine $a_{{\rm B},i}$ for events $[\mathrm{A_4}\; \mathrm{A_5}\; \mathrm{A_8}]$, we take the modulus of sine and cosine to accommodate all four quadrants. Likewise we can determine maximum interfering distance $b_{{\rm B},i}$ as
\begin{align}
    b_{{\rm B},i} &= \mathrm{XV} = \mathrm{ZV} - \mathrm{ZX} = \sqrt{\mathrm{OV}^2 - \mathrm{OZ}^2} - \mathrm{ZX} \nonumber\\
    &= \sqrt{R_{\rm B}^2 - \left(d_i \sin\theta_i\right)^2} - d_i\cos\theta_i.
    \label{eq:eq_appx_2}
\end{align}
\begin{figure*}[t]
\centering
\subfloat[]
{\includegraphics[width=0.26\textwidth]{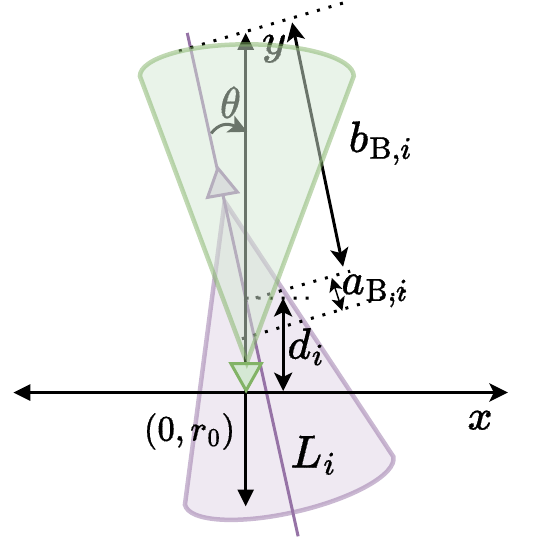}
\label{fig:case1_1}}
\hfil
\subfloat[]
{\includegraphics[trim={0cm 0cm 0cm 1cm},clip,width=0.2\textwidth]{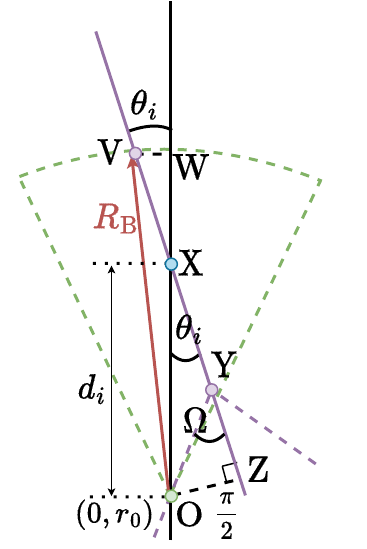}
\label{fig:case1_2}}
\hfil
\subfloat[]
{\includegraphics[width=0.28\textwidth]{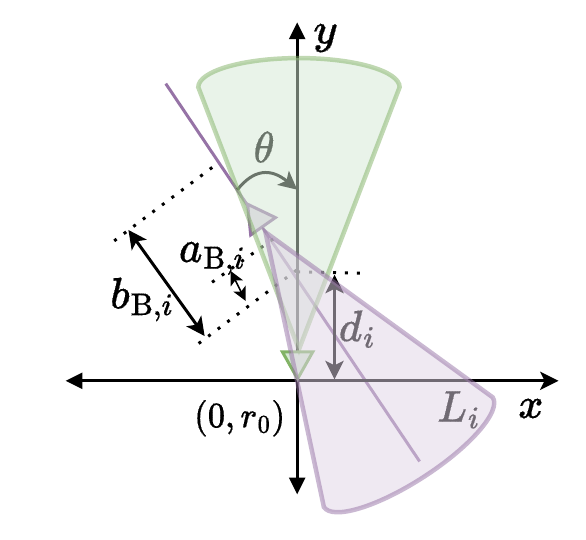}
\label{fig:case1_3}}
\hfil
\subfloat[]
{\includegraphics[trim={0cm 0cm 0cm 0.4cm},clip,width=0.2\textwidth]{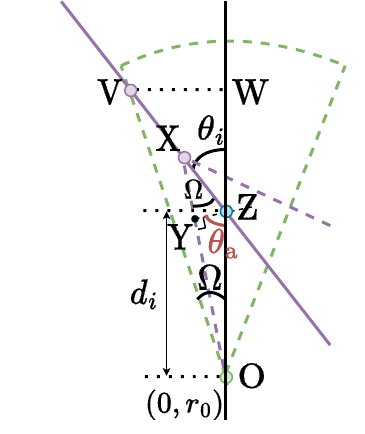}
\label{fig:case1_4}}
\caption{Geometric illustrations for Case 1: Intersection inside the radar sector. (a)-(b) depict \textbf{case  1(a)} where the line $L_i$ intersects the sector's circular arc. (c)-(d) depict \textbf{case 1(b)} where $L_i$ intersects one of the sector's edge lines. All figures show a generating angle in the first quadrant; analogous geometries hold for other quadrants after appropriate symmetry considerations.}
\label{fig:appx_01} 
\end{figure*}

\begin{figure}[t]
\centering
\subfloat[]
{\includegraphics[width=0.26\textwidth]{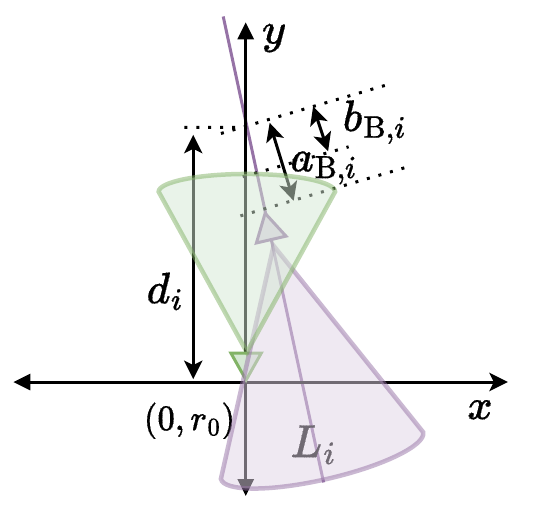}
\label{fig:case2}}
\hfil
\subfloat[]
{\includegraphics[width=0.21\textwidth]{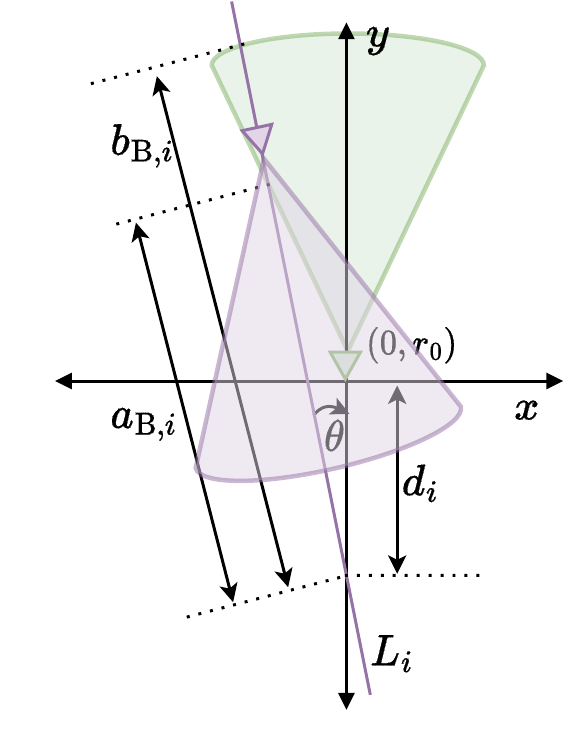}
\label{fig:case3}}
\caption{Geometric illustrations for Cases 2 and 3. (a) Case 2 with generating angle present in the first quadrant, and Case (3) with generating angle present in the third quadrant.}
\label{fig:appx_02} 
\end{figure}

\textbf{Case 1(b): Mixed Intersection}\\
Conditions:
\begin{align*}
    & 0 \leq d_i \leq c_2, \;\; \text{and} \\ 
    & \theta_i \in [0, 2\pi] \setminus \big\{[\Omega, \pi - \Omega] \cup [\pi + \Omega, 2\pi - \Omega]\big\}
\end{align*}
In this subcase events $\mathrm{A_2} \cup \mathrm{A_3} \cup \mathrm{A_6} \cup \mathrm{A_7}$ take place. Fig.~\ref{fig:case1_3} and~\ref{fig:case1_4} illustrates an example of this subcase, wherein the $\theta_i$ is generated in first quadrant i.e. $\theta_i \in \mathrm{A_2}$. In this subcase, $L_i$ may intersect one of the edge lines or the circular arc portion of the radar sector. From Fig.~\ref{fig:case1_4}, we observe that the minimum interfering distance $a_{{\rm B},i}$ is calculated from point $\mathrm{X}$ to $\mathrm{Y}$ and would be negative, likewise maximum interfering distance $b_{{\rm },i}$ is calculated from point $\mathrm{Z}$ to $\mathrm{X}$. The minimum interfering distance is found as
\begin{align*}
    \sin \Omega &= \frac{\mathrm{ZY}}{\mathrm{ZX}} = \frac{d_i \cos \left(\theta_{\rm a}\right)}{a_{{\rm B},i}} = \frac{d_i \cos \left(\frac{\pi}{2} + \Omega - \theta_i\right)}{a_{{\rm B},i}} \\
    \therefore\, a_{{\rm B},i} &= \frac{d_i \sin \left(\theta_i -  \Omega\right)}{\sin \Omega}
\end{align*}
In the above scenario where the interfering angle lies in the first quadrant, $\theta_{\rm a} = \frac{\pi}{2} + \Omega - \theta_i$, and the same is found if $\theta_i$ lies in the third quadrant. But for the second and fourth quadrant $\theta_{\rm a} = \frac{\pi}{2} + \Omega + \theta_i$. Also the modulus operator is applied in $d_i \sin \left(\theta_i -  \Omega\right)$ and $d_i \sin \left(\theta_i +  \Omega\right)$ to accommodate all four quadrants. Unlike $a_{{\rm B}, i}$, where there are two different formulas for the first/third and second/fourth quadrant, we have two different formulas for $b_{{\rm B}, i}$, but they depend on distance $d_i$. When $c_1 \leq d_i \leq c_2$, we see that line $L_i$ will intersect only the circular arc portion of the radar sector; hence, the distance $b_{{\rm B},i}$ is the same as of case 1(a). For the case of $0 \leq d_i < c_1$, the line $L_i$ does not intersect the circular arc, but one of the edge lines of the radar sector as shown in Fig.~\ref{fig:case1_3} and~~\ref{fig:case1_4}. The distance $d_i = c_1$ represents the edge case when line $L_i$ intersects simultaneously with the radar sector's circular arc and one of the edge lines. Therefore, for $d_i < c_1$, line $L_i$ intersects only the border lines. Thus by solving $y = \cot\Omega + r_0$, $x\cos\theta_i + y\sin\theta_i = r_i$, and $x^2 + (y - r_0^2) = R_{\rm B}^2$ for $y$ gives us the value of $c_1$. And as usual, take the modulus of the trigonometric functions, which are functions of generating angels. The value of $b_{{\rm B}, i}$ is then found as,
\begin{align*}
    \tan \Omega &= \frac{\mathrm{VW}}{d_i + \mathrm{WZ}} = \frac{b_{{\rm B}, i} \sin \theta_i}{d_i + b_{{\rm B}, i} \cos\theta_i} \\
    \therefore\; b_{{\rm B}, i} &= \frac{d_i \tan \Omega}{\sin\theta_i - \tan \Omega \cos\theta_i}
\end{align*}
Now, for any $\theta_i \in \Theta$, $c_2$ is the maximum value of $d_i$, beyond which any none of the interfering radars can cause interference at ego radar and is found by solving,
\begin{align*}
    c_2\sin\theta_i \cot\Omega - c_2 \cos\theta_i \leq \sqrt{R_{\rm B}^2 - \left(c_2 \sin\theta_i\right)^2} - c_2\cos\theta_i.
\end{align*}

\textbf{Case 2: Intersection Ahead but Outside the Sector}\\
Conditions
\begin{align*}
    & d_i > R_{\rm B}, \;\; \text{and} \\
    & \theta_i \in [0, 2\pi] \setminus \big\{[\Omega, \pi - \Omega] \cup [\pi + \Omega, 2\pi - \Omega]\big\}
\end{align*}
In this case, the intersection point lies ahead and outside the radar sector. This case arises only when event $\mathrm{A_1} \cup \mathrm{A_4} \cup \mathrm{A_5} \cup \mathrm{A_8}$ take place. Fig.~\ref{fig:case2} illustrates a scenario of Case 2 i.e. $\theta_i \in \mathrm{A_1}$, where we see that intersection is happening at a distance $d_i > R_{\rm B}$ i.e. the line intersects ahead of the ego radar but outside the radar sector. Despite this, portions of the line still pass through the radar sector. If we examine this scenario carefully, we notice that distance $a_{{\rm B},i}$, and $b_{{\rm B},i}$ are similar to earlier scenarios we have encountered in Case 1. The distance $a_{{\rm B},i}$ and $b_{{\rm B},i}$ is same as of distance given in~\eqref{eq:eq_appx_1} and~\eqref{eq:eq_appx_2} respectively, given in case 1(a). Although from Fig.~\ref{fig:case2}, it appears that $b_{{\rm B},i}$ is smaller than $a_{{\rm B},i}$, this is not the case. As the distance is measured from the perspective of the interfering point, both $b_{{\rm B},i}$ are negative such that $b_{{\rm B},i} > a_{{\rm B},i}$. Thus, by combining case 1(a) and case 2, we revise the range of $d_i$ as $c_1 \leq d_i \leq c_2$.

\textbf{Case 3: Intersection Behind the Ego Radar}\\ 
Conditions
\begin{align*}
    & d_i < 0, \;\; \text{and} \\
    & \theta_i \in [0, 2\pi] \setminus \big\{[\Omega, \pi - \Omega] \cup [\pi + \Omega, 2\pi - \Omega]\big\}
\end{align*}
This is the third case that arises when the point of intersection is beyond ego radar. Interference is possible only if the line still intersects the radar sector. This case arises only when event $\mathrm{A_1} \cup \mathrm{A_4} \cup \mathrm{A_5} \cup \mathrm{A_8}$ take place. Fig.~\ref{fig:case3} illustrates a scenario of Case 3, i.e., $\theta_i \in \mathrm{A_5}$ when the intersection point lies behind the ego radar. The distance $b_{{\rm B},i}$ is the same as the distance given in~\eqref{eq:eq_appx_2}, given in case 1(a), as $L_i$ is intersecting the circular arc region of radar sector. The value of $a_{{\rm B},i}$ can be derived in a similar manner to how we derived the $b_{{\rm B},i}$ in case 1(a) and 1(b). Importantly in case 3, the lines intersecting behind the ego radar cause interference only if the intersecting distance $d_i$  is greater than. If $d_i < c_1$, line $L_i$ does not fall within the radar sector of ego radar. 

Collecting the results from the three cases and their subcases, and noting the special situation of the ego radar's own street ($i=0$, where $a_{\mathrm{B},i}=R$ and $b_{\mathrm{B},i}=R_{\mathrm{B}}$), we obtain the piecewise expressions stated in Theorem~1. The thresholds $c_1$ and $c_2$ are as defined above, ensuring that the formulas are applied only in regimes where the line $L_i$ actually cuts through the radar sector.

This completes the proof of theorem~1.

\end{document}